%% file: arxiv.tex
\documentclass[12pt,twoside]{article}
\usepackage{fancyhdr}
\PassOptionsToPackage{hyphens}{url}\usepackage[bookmarks=false]{hyperref}
\usepackage{amsfonts,epsfig,graphicx}
\usepackage{afterpage}
\usepackage{amsmath,amssymb,amsthm} 
\usepackage{mathtools}
\usepackage{fullpage}
\usepackage[T1]{fontenc} 
\usepackage{epsf} 
\usepackage{graphics} 
\usepackage{amsfonts,amsmath}
\usepackage{natbib} 
\usepackage{psfrag,xspace}
\usepackage{color,etoolbox}
\usepackage{comment}

\usepackage[dvipsnames,table]{xcolor}
\hypersetup{colorlinks,
citecolor=RoyalBlue,
urlcolor=RoyalBlue,
linkcolor=RoyalBlue}

\usepackage[
    letterpaper,
    margin=1in,
    ]{geometry}

\theoremstyle{plain}
\newtheorem{theorem}{Theorem}
\newtheorem{lemma}{Lemma}
\newtheorem{proposition}{Proposition}
\newtheorem{corollary}{Corollary}
\newtheorem{assumption}{Assumption}
\theoremstyle{definition}
\newtheorem{definition}{Definition}
\newtheorem{example}{Example}

\newtheorem{game}{Game}
\theoremstyle{remark}
\newtheorem{remark}{Remark}

\usepackage{stix2}

\usepackage{subcaption}
\usepackage{makecell}
\usepackage{multirow}
\usepackage{outlines}
\usepackage{cancel}
\usepackage{footnote}
\usepackage{pifont}
\usepackage{rotating}
\usepackage{dsfont}
\usepackage[shortlabels]{enumitem}
\usepackage{accents}
\usepackage{booktabs}
\usepackage{setspace}

\input{contents/macros}

\usepackage{etoolbox}
\newtoggle{compact}
\togglefalse{compact}

\newcommand{\arxiv}{1}

\title{%
\vspace{-2em}
{\bf Betting on Bets:\\ Anytime-Valid Tests for Stochastic Dominance}
\vspace{0.5em}
}
\author{%
    {\bf Sebastian Arnold}\thanks{Co-first-author (equal contribution). The names of these authors are ordered alphabetically.} \\
    CWI \\
    \email{sebastian.arnold@cwi.nl} 
    \and
    {\bf Yo Joong Choe}\footnotemark[1] \\
    INSEAD \\
    \email{yojoong.choe@insead.edu} 
    \and
    {\bf Marco Scarsini} \\
    LUISS \\
    \email{marco.scarsini@luiss.it} 
    \and
    {\bf Ilia Tsetlin} \\
    INSEAD \\
    \email{ilia.tsetlin@insead.edu}
    \vspace{1em}
}
\date{\normalsize July 31, 2026}

\begin{document}

\onehalfspacing

\maketitle

\vspace{-2em}
\begin{abstract}
    \input{contents/abstract}

\end{abstract}

\clearpage

\input{contents/body}


\subsection*{Acknowledgements}
\input{contents/acknowledgements}

\bibliography{contents/references}
\bibliographystyle{apalike}

\clearpage

\appendix
\begin{center}\section*{Appendix}\end{center}
\input{contents/appendix}

\end{document}

%% file: contents/macros.tex
\newcommand{\email}[1]{\href{mailto:#1}{\textcolor{black}{\texttt{#1}}}}

\newcommand{\inparen}[1]{\left(#1\right)}
\newcommand{\incurly}[1]{\left\{#1\right\}}
\newcommand{\insquare}[1]{\left[#1\right]}

\DeclareMathOperator*{\indsymb}{\mathds{1}}
\newcommand{\indicator}[1]{\indsymb\inparen{#1}}

\DeclareMathOperator*{\argmax}{argmax}

\newcommand{\N}{\mathbb{N}}

\newcommand{\Q}{\mathbb{Q}}
\newcommand{\R}{\mathbb{R}}

\newcommand{\calF}{\mathcal{F}}
\newcommand{\calG}{\mathcal{G}}
\newcommand{\calH}{\mathcal{H}}

\newcommand{\calN}{\mathcal{N}}

\newcommand{\calP}{\mathcal{P}}
\newcommand{\calQ}{\mathcal{Q}}

\newcommand{\calU}{\mathcal{U}}

\newcommand{\calZ}{\mathcal{Z}}

\newcommand{\Esymb}{\mathbb{E}}
\newcommand{\Psymb}{\mathbb{P}}
\newcommand{\Qsymb}{\mathbb{Q}}

\newcommand{\prob}[1]{\Psymb\left({#1}\right)}

\newcommand{\ex}[1]{\Esymb\left[{#1}\right]}
\newcommand{\Ex}[2]{\Esymb_{{#1}}\left[{#2}\right]}

\newcommand{\bbQ}{\mathbb{Q}}

\newcommand{\frakB}{\mathfrak{B}}

\newcommand{\e}{\mathfrak{e}}

\newcommand{\fbet}{\lambda}  
\newcommand{\up}{\mathsf{UP}}

\newcommand{\sd}{\preceq}
\newcommand{\notsd}{\npreceq}

\newcommand{\fsd}{\preceq_{1}}
\newcommand{\notfsd}{\npreceq_{1}}
\newcommand{\ssd}{\preceq_{2}}

\newcommand{\tsd}{\preceq_{3}}

\newcommand{\ksd}{\preceq_{k}}

\newcommand{\fsdstrict}{\prec_{1}}

\newcommand{\ssdstrict}{\prec_{2}}

\newcommand{\tsdstrict}{\prec_{3}}

\newcommand{\diff}{\,\mathrm{d}}
\renewcommand{\epsilon}{\varepsilon}
\newcommand{\bX}{\boldsymbol{X}}
\newcommand{\bY}{\boldsymbol{Y}}

%% file: contents/abstract.tex
How can we monitor, in real time, whether one uncertain prospect has any upside over another? 
To answer this question, we develop a novel family of sequential, anytime-valid tests for \emph{stochastic dominance (SD),} a classical and popular notion for comparing entire distribution functions. 
The problem is distinct from that of testing mean dominance, and it is particularly useful when comparing distributions with similar means or with ordinal outcomes. 
We first derive powerful, nonparametric e-processes that quantify evidence against the null hypothesis that one prospect is stochastically dominated by another. 
For first-order SD, these e-processes are based on mixtures of growth-rate optimal e-variables, yielding a test of power one that retains validity under continuous monitoring. 
We then generalize the approach to sequential testing for higher-order SD and other integral stochastic orders. 
Empirically, we find that the tests are competitive in power with classical, non-anytime-valid SD tests.
Our real-world application examines a controversial phenomenon in baseball analytics, known as the ``third-time-through-the-order (3TTO) penalty,'' viewed as a monitoring problem. 
We close by sketching the complementary problem of testing whether a prospect has a \emph{definite} upside, formalizing conditions under which we can derive a powerful anytime-valid test.

%% file: contents/body.tex
\section{Introduction}\label{sec:introduction}

\emph{Stochastic dominance (SD),} also known as \emph{stochastic ordering,} is a classical framework for comparing an arbitrary pair of random variables~\citep[e.g.,][]{muller2002comparison,shaked2007stochastic}. 
In its simplest form, first-order SD (FSD) states that the CDF of one variable lies entirely below that of another; higher-order SD notions generalize this idea to iterated integrals of the CDFs.
Yet, SD is far more complex and nonparametric than dominance in means; in fact, no finite number of moment inequalities implies FSD or higher-order SD~\citep{brockett1992risk}.
SD is a cornerstone in economics and decision theory~\citep[e.g.,][]{fishburn1970utility,levy1992stochastic,muller2002comparison}, where a decision maker (DM) is tasked with comparing uncertain prospects under general classes of utility functions. 

Recent advances in the theory of \emph{sequential, anytime-valid inference (SAVI)} based on \emph{e-processes}~\citep{shafer2011test,grunwald2019safe,ramdas2020admissible,ramdas2022game} motivate us to revisit the problem of SD testing. 
While most classical SD tests are fixed-sample tests with asymptotic validity~\citep[e.g.,][]{mcfadden1989testing,barrett2003consistent,linton2010improved}, the SAVI framework warrants a new look at SD testing because it is particularly well-suited to this problem. 
First, SD testing is popular in applications where the data are inherently sequential and the goal is to compare and order uncertain prospects in the future based on past data.
This makes it desirable to use SAVI methods, which provably control the type I error under optional stopping (at adaptively chosen sample sizes). 
Second, e-processes are natural candidates for handling the composite and nonparametric nature of SD testing problems; as we shall see in this paper, these methods do not require any parametric assumptions on the distributions being compared. 
We will also see, particularly in Section~\ref{sec:testing_non_SD_null}, that SD testing itself poses unique challenges to the existing SAVI stack, as it will require both constructing new e-processes and rethinking the roles of null and alternative hypotheses.

To ground our discussion in real-world applications, consider the problem of testing whether a new prospect (treatment) $Y$ has any upside, relative to a benchmark (control) option $X$. 
This happens in a variety of contexts: a developer piloting a new subscription plan; the FDA approving a new drug for the market; a sports team manager deciding whether to take out a player in the later stages of a game; and so on. 
The DM needs to know whether the unknown prospect $Y$ deserves a look over the benchmark $X$ (e.g., the baseline subscription plan, the efficacy of an existing drug, and the player's usual performance earlier in the game). 
These random variables may be discrete or continuous, but we do not know their parametric forms. 
Over time, the DM wants to monitor the data and collect evidence on whether $Y$ has an upside and should be adopted as a viable option. 

The core statistical problem common to all of the above scenarios is to test for an ``upside'' of one uncertain outcome over another. 
We formalize this as testing whether $Y$ is stochastically dominated by $X$.
Denoting $\sd$ as an SD relation in some (e.g., usual) order, we test
\begin{equation}\label{eq:testing_sd_null}
    \calH_0: Y \sd X \quad \text{vs.} \quad \calH_1: Y \notsd X\,.
\end{equation} 
Rejecting the null would inform the DM that treatment $Y$ has an upside over $X$ and they should consider it as a contender. 
To address this problem in a rigorous and flexible manner, we develop e-processes and sequential tests that are
\begin{enumerate}[label=(\alph*)]
    \item \textbf{Anytime-valid:} The e-process is valid at \emph{any} adaptively chosen sample size, resulting in a sequential test that remains valid under continuous monitoring;
    \item \textbf{(Asymptotically) powerful:} The e-process construction is based on a mixture of growth-rate optimal (GRO) e-variables, such that under any alternative, it enjoys exponential growth and yields a sequential test with asymptotic power one;
    \item \textbf{Nonparametric:} The methods do not rely on parametric assumptions on the distributions of $X$ and $Y$, requiring only the i.i.d.~assumption (when testing FSD); 
    \item \textbf{Robust to cross-dependence:} These properties do not rely on independence between the two random variables, including when they are antimonotonic. 
\end{enumerate}

\begin{figure}[t]
    \centering
    \includegraphics[width=0.8\textwidth]{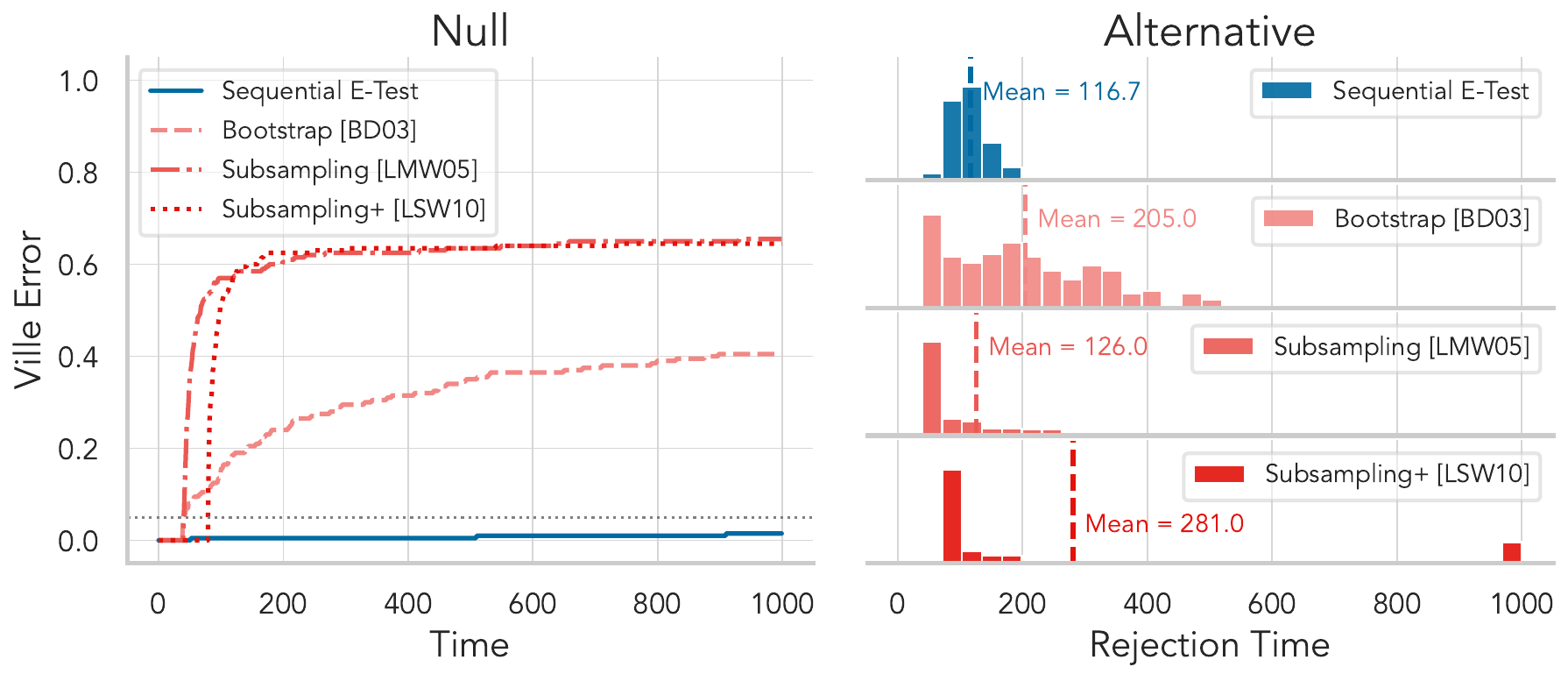}
    \caption{\emph{Sequential e-tests are anytime-valid, unlike classical resampling-based tests, and achieve competitive power.} 
    In the null case (left), we see that the \emph{cumulative} type I error under continuous monitoring (``Ville error'') is only controlled by the sequential e-test. 
    Under the alternative (right), we plot the histogram of first rejection times under continuous monitoring, along with the mean. 
    The rejection times of the e-test are no worse on average than the best-performing baseline and have a smaller spread.}
    \label{fig:comparison_classical_tests}
\end{figure}

Figure~\ref{fig:comparison_classical_tests} depicts the main methodology proposed in this paper (a ``sequential e-test''), in comparison with classical, non-anytime-valid methods for testing FSD using bootstrapping \citep[BD03]{barrett2003consistent}, subsampling \citep[LMW05]{linton2005consistent}, and ``improved'' subsampling \citep[LSW10]{linton2010improved}, on simulated data (additional details are deferred to Section~\ref{sec:simulations}). 
The left panel clearly shows that, under the null, all classical methods fail to control the type I error under continuous monitoring (Ville error), while the sequential e-test remains valid. 
The right panel shows that, under the alternative, the distribution of first rejection times under monitoring is comparable to (at the very least, not dominated by) the best-performing classical method for this case. 

\emph{Outline.}
After reviewing the background and related work in Section~\ref{sec:background_and_related_work}, we study the sequential SD testing problem~\eqref{eq:testing_sd_null} in Section~\ref{sec:FSD} for first-order SD and discuss extensions to higher-order SD.
Section~\ref{sec:simulations} contains simulations demonstrating the validity and power of the proposed methods; Section~\ref{sec:real_data_exp} applies them to examine the controversial ``3TTO penalty'' in Major League Baseball.
In Section~\ref{sec:testing_non_SD_null}, we then consider the reverse problem of affirming an SD relation before we conclude in Section~\ref{sec:discussion}. 
The Supplementary Material contains omitted proofs, detailed summaries of related work, additional experiments, and further extensions.

\section{Background and related work}\label{sec:background_and_related_work}

\subsection{Notation and background on statistical hypothesis testing}\label{sec:notation}

Throughout the paper, unless specified otherwise, we consider an i.i.d.\ sequence $(X_1,Y_1), (X_2,Y_2), \ldots$ of paired random variables (rv's) with common support $\calZ \subseteq \mathbb{R}$, equipped with the Borel $\sigma$-algebra, and write $(X,Y)$ for a generic pair distributed as $(X_1,Y_1)$. 
We assume that all random quantities are defined on an underlying measurable space $(\Omega,\calF)$ and denote by $\frakB$ the family of all probability measures on $(\Omega,\calF)$. 
For $t\in \N$, we let $(X,Y)^t=((X_i,Y_i))_{i=1}^t$, and define the data-filtration $\calF_t= \sigma((X,Y)^t)\subseteq \calF$.

We call a subset $\calP \subseteq \frakB$ a \emph{statistical hypothesis}. As usual, we denote by $\calH_0 \subseteq \frakB$ the \emph{null hypothesis}, and by $\calH_1 \subseteq \frakB$ the \emph{alternative hypothesis}. 
For any ${\Psymb}\in \frakB$, we let $\Psymb_{(X,Y)}$ be the distribution of $(X,Y)$ under ${\Psymb}$, that is the pushforward measure of ${\Psymb}$ under $(X,Y)$. Likewise, we denote by $\Psymb_X$ and $\Psymb_Y$ the marginal distributions of $X$ and $Y$ under ${\Psymb}$, respectively. 
Whenever the dependence on $\Psymb$ can be suppressed, we denote the CDFs of $X$ and $Y$ simply by $F_X$ and $F_Y$. Finally, we let $\Pi(\calZ)$ be the family of probability measures (mixtures) on $\calZ$. 

\subsection{Preliminaries on e-variables/processes and SAVI}

A nonnegative rv $E$ is called an \emph{e-variable} for $\calH_0 \subseteq \frakB$ if $\Ex{\Psymb}{E} \leq 1$ for all ${\Psymb} \in \calH_0$. 
Likewise, adapted, nonnegative rv's $(E_t)_{t\in \N}$ are called \emph{sequential e-variables} (for $\calH_0$) if $\Esymb_{\Psymb}[E_t\mid \calF_{t-1}] \leq 1$ for all $t\in \N$ and for all $\Psymb\in\calH_0$. 
The running product of a sequence of sequential e-variables (for $\calH_0$) forms a \emph{test (super)martingale (SM)}, that is, a nonnegative (super)martingale with respect to any $\Psymb\in \calH_0$ starting at one.
Often, a test SM is interpreted as the wealth process of a skeptic playing a betting game (see Remark~\ref{remark:betting_on_bets}).

The key tool for anytime-valid testing is Ville's inequality \citep{ville1939etude}, a time-uniform extension of Markov's inequality that controls the probability a test SM ever exceeds a threshold.
\begin{lemma}[Ville's inequality]\label{lem:Ville}
    Let $(M_t)_{t\in \N}$ be a test (super)martingale for $\calH_0$. Then, 
    \begin{equation}\label{eq:Ville}
        \Psymb\left(\exists t\in \N: M_t \geq \frac{1}{\alpha}\right) \leq \alpha, \quad \forall \Psymb \in \calH_0, \forall \alpha \in (0,1). 
    \end{equation}
\end{lemma}

Importantly, there are composite null hypotheses for which it is not possible to construct non-trivial test SMs~\citep[e.g.,][]{ramdas2021testing,henzi2023rank}. 
However, in such situations it may still be possible to find non-trivial e-processes. 
An \emph{e-process} for a composite hypothesis $\calH_0$ is a nonnegative adapted stochastic process $(E_t)_{t \in \N}$ such that $E_\tau$ is an e-variable, for any (possibly infinite) stopping time $\tau$, that is, $\Ex{\Psymb}{E_\tau}\leq 1$ for all $\Psymb \in \calH_0$. 
This notion of validity at arbitrary stopping times is referred to as \emph{anytime-validity}. 
Equivalently, an adapted process $(E_t)_{t \in \N}$ is an e-process for $\calH_0$ if it is upper bounded by a test SM, for each $\Psymb \in \calH_0$. 
Hence, any test SM is also an e-process, and Ville's inequality continues to hold for e-processes \citep{ramdas2022game}.
For any $\alpha\in (0,1)$ and e-process $(E_t)_{t \in \N}$ for $\calH_0$, by defining $\phi_t = \indicator{E_t \ge 1/\alpha}, t\in \N,$ we directly obtain a valid \emph{level-$\alpha$ sequential test}, that is, an adapted sequence of tests controlling the \emph{Ville error} (cumulative type I error), $\Psymb(\exists t\geq 1: \phi_t=1)$, at $\alpha$, under any $\Psymb\in \calH_0$.

\subsection{Preliminaries on stochastic dominance}\label{subsec:preliminaries_on_SD}

\emph{First-order SD} is the canonical SD relation, also known as the \emph{usual stochastic order}. 

\begin{definition}\label{def:FSD}
Given rv's $X$ and $Y$, we say that $X$ \emph{dominates $Y$ in the first order} if 
\begin{equation}\label{eq:def_FSD}
    F_X(z)\leq F_Y(z), \quad \forall z\in \calZ\,,
\end{equation}
and write in this case $Y \fsd X$ (or $\Psymb_Y \fsd \Psymb_X$ when the underlying law $\Psymb$ is specified).  
\end{definition}

FSD defines a partial order on the space of probability distributions: it is reflexive, antisymmetric and transitive, but it is not a total order as two distributions are incomparable whenever their corresponding CDFs cross. 
It is notably distinct from the complete order based on means --- in particular, $Y \fsd X$ implies $\Ex{\Psymb}{Y} \leq \Ex{\Psymb}{X}$ (if moments exist) but not vice versa.  
In practice, we do not have access to the distribution functions $F_X$ and $F_Y$, and we will have to estimate their relation in some way to draw inference on FSD. 

The following Lemma gives a well-known equivalent characterization for FSD. 
For a proof, see, e.g., Section 1.A of~\citet{shaked2007stochastic}. 
\begin{lemma}
\label{lem:FSD_equivalences}
For rv's $X$ and $Y$, we have $Y \fsd X$ if and only if $\Esymb[u(Y)] \leq \Esymb[u(X)]$ for all measurable and increasing functions $u:\calZ \to \R$ for which the expectations exist. 
\end{lemma}

Lemma~\ref{lem:FSD_equivalences} characterizes FSD with ``utility functions'': under the expected-utility framework~\citep{savage1951theory}, $Y \fsd X$ means that any utility-maximizing DM who prefers more to less would prefer $X$ over $Y$. 
In this view, SD establishes a preference ordering between prospects for \emph{any} DM with a non-decreasing, but otherwise unspecified, utility function. 
A similar characterization holds for higher-order SD, which we discuss in Section~\ref{sec:testing_kSD_summary}.

\subsection{Related work on testing stochastic dominance}\label{sec:related_work_testing_SD}

Many modern SD testing methods were developed in the context of economics and decision theory, where utility-theoretic characterizations of SD are well-studied~\citep[e.g.,][]{fishburn1970utility,levy1992stochastic,mosler1993stochastic,muller2002comparison}. 
Of particular relevance to our discussion on \emph{testing} SD is \citet{mcfadden1989testing}'s asymptotic test for FSD, based on the \emph{one-sided Kolmogorov-Smirnov (KS) statistic}:
\begin{equation}\label{eqn:ks_stat}
    \hat{D}^\mathsf{KS} = \sup_{z \in \calZ}\, [\hat{F}_X(z) - \hat{F}_Y(z)]\,,
\end{equation}
where $\hat{F}_X$ and $\hat{F}_Y$ denote the empirical CDFs for $X$ and $Y$, respectively. 
Under the i.i.d.~assumption and additional regularity conditions on the CDFs (e.g., continuity), \citet{mcfadden1989testing} derives the limiting distribution for the one-sided KS statistic. 

The KS-based approach has since been improved and generalized to higher-order SD and non-standard dominance testing~\citep[e.g.,][]{kaur1994testing,abadie2002bootstrap,barrett2003consistent,donald2016improving,zhuang2024tests}.
It is common to make regularity assumptions on CDFs, including that of bounded support, when testing higher-order SD~\citep[e.g.,][]{barrett2003consistent,donald2016improving}. 
The main approach we discuss in Section~\ref{sec:FSD} can be viewed as an anytime-valid and non-asymptotic generalization of the one-sided KS test for FSD, whose validity does not rely on any regularity condition on the CDFs. 
For higher-order SD, the validity and power of our general approach only depend on having a finite lower bound. 

Aside from the KS statistic, a parallel line of work~\citep{linton2005consistent,linton2010improved} has popularized the \emph{one-sided Cram\'er-von Mises statistic}, a mixture of squared positive differences between the CDFs:
$ \hat{D}^\mathsf{mix}_\psi = \int_\calZ \insquare{(\hat{F}_X(z) - \hat{F}_Y(z))_+}^2  \psi(z) dz\,,$
where $\psi$ is a mixture on $\calZ$. 
See~\citet{whang2019econometric} for a recent review. 
These methods rely on more specific regularity conditions and validity under resampling (e.g., bootstrapping), both of which we avoid entirely in this work.

All aforementioned methods are only asymptotically valid at \emph{pre-specified} sample sizes; many also rely on regularity conditions to validate a resampling procedure. 
In \emph{anytime-valid} SD testing, a classical reference is \citet[][DR]{darling1968some}, who derive a nonparametric sequential test for first-order SD, alongside a two-sided KS test, under independent data streams. 
In the modern SAVI literature, \citet[][SR]{shekhar2023nonparametric} develop a general testing-by-betting framework for integral probability metrics, improving upon DR's two-sided KS test in terms of power. 
Notably, their early draft (arxiv.v1) mentions an extension to SD testing, although it was not pursued further and subsequently removed from their published manuscript. 
We position our work as the first dedicated study of sequential SD testing in its various forms, providing a powerful, robust, and broadly applicable framework.
At its core is an explicit derivation of (novel) growth-rate optimal e-variables for FSD. The framework accommodates flexible choices of mixture strategies that achieve asymptotic power one, and it subsumes one of SR's key ideas as a special case (the ``greedy'' mixture).

There are two more related SAVI papers to note. 
First, as a corollary to their general construction of time-uniform CDF bands, \citet{howard2022sequential} propose a sequential test for FSD. 
Second, building on the theory of reverse submartingales, \citet{manole2023martingale} derive time-uniform lower and upper confidence bounds for the two-sided KS statistic, which can be adapted to its one-sided version \eqref{eqn:ks_stat}. 
Whereas the first paper primarily studies time-uniform CDF estimation, the second focuses on convex divergence estimation, with FSD testing appearing only as a special case in both.
We will see in simulations that these sequential tests are generally far more conservative than tests specifically designed for SD.

\section{Anytime-valid SD testing via e-processes}\label{sec:FSD}

In this section, we derive e-processes and anytime-valid tests for the null hypothesis that one prospect stochastically dominates another.
While our framework generalizes to testing higher-order SD and other integral stochastic orders~\citep{feng2025integral}, we first focus our attention on testing first-order SD, and we return to the general case at the end of the section. 

\subsection{The FSD testing problem}\label{subsec:formulation_of_FSD_null}

We consider testing the null hypothesis that $Y$ is dominated by $X$ in the first order.  
First, we embed Definition~\ref{def:FSD} into a hypothesis testing problem by defining the \emph{FSD null} as
\begin{equation}\label{eq:def_FSD_null}
    \calH_0 = \{\Psymb \in \frakB \mid \Psymb_Y \sd_1 \Psymb_X \} =\{\Psymb \in \frakB \mid \Psymb(X\leq z) \leq \Psymb(Y\leq z),\, \forall z\in \calZ\}\,.
\end{equation}
The FSD null forms an intersection null, and may be written as $\calH_0 =\bigcap_{z \in \calZ}\calH_0 (z)$,
for 
\begin{equation}\label{eq:FSD_null_at_thresholds}
 \calH_0 (z)
 =\left\{\Psymb \in \frakB \mid \Ex{\Psymb}{\indicator{X \leq z} - \indicator{Y \leq z}}\leq 0\right\}, \quad z\in \calZ\,.
\end{equation}

\begin{lemma}\label{lem:FSD-null-is-convex}
    The FSD null~\eqref{eq:def_FSD_null} is convex.
\end{lemma}

For completeness, we provide a proof in Supplement~\ref{sec:proofs}, where all other omitted proofs can also be found. 
Since any e-variable for $\calH_0$ is readily an e-variable for its convex hull, Lemma~\ref{lem:FSD-null-is-convex} implies that we can test the FSD null without having to test a much larger null. 

We remark that the FSD null~\eqref{eq:def_FSD_null} is defined solely in terms of the marginal distributions of $X$ and $Y$.
The following toy example illustrates how dependence between $X$ and $Y$ may influence the difficulty of achieving power against this null when testing it sequentially. 
\begin{example}
\label{ex:anticorr}
Consider $\Psymb \in \frakB$ with 
$\mathbb{P}_X=\frac{1}{2}\delta_{0}+\frac{1}{2}\delta_{{2/3}}$ and 
$\mathbb{P}_Y=\frac{1}{2}\delta_{1/3}+\frac{1}{2}\delta_{1}$, such that $X \fsd Y$ and $\Psymb \notin \calH_0$. 
These marginals can be realized by the \emph{comonotonic coupling} $\mathbb{P}_{(X,Y)} =\frac{1}{2}\delta_{(0, 1/3)} + \frac{1}{2}\delta_{(2/3, 1)}$ or by the \emph{antimonotonic coupling} $\mathbb{P}_{(X,Y)} = \frac{1}{2}\delta_{(0, 1)} + \frac{1}{2}\delta_{(2/3, 1/3)}$. 
We expect to reject $\calH_0$ faster in the former case, where $X_t\leq Y_t$ for all $t \in \N$, than in the second case, where $\mathbb{P}(X_t<Y_t)=\mathbb{P}(X_t>Y_t)=1/2$. 
\end{example}

\subsection{Testing the SD null hypothesis at thresholds}

We first construct e-variables for testing the FSD null for a fixed $z\in \calZ$ in a single round. 
While their construction is straightforward, these e-variables will serve as building blocks for deriving powerful e-processes that sequentially test $\calH_0$.
\begin{lemma}[The building-block e-variable]\label{lem:building-block-e-variables}
    For any $z \in \calZ$,  
    \begin{equation}\label{eq:def_coin_betting_e-variable}
        S(\lambda,z) \coloneqq
        1 + \lambda [\indicator{X \leq z} - \indicator{Y \leq z}]
    \end{equation}
    is an e-variable for $\calH_0(z)$ at \eqref{eq:FSD_null_at_thresholds}, for any $\lambda \in [0,1]$.
\end{lemma}
\begin{proof}
    $S(\lambda,z)\geq 0$, for any $\lambda \in [0,1]$, and $\Ex{\Psymb}{S(\lambda,z)}\leq 1$ for any $\Psymb\in \calH_0(z)$. 
\end{proof}
E-variables of the form~\eqref{eq:def_coin_betting_e-variable} appear ubiquitously in the literature on testing the means of bounded rv's. 
Their origins go back at least to optimal gambling strategies for coin betting due to \citet{kelly1956new}, where $\lambda$ is the gambler's choice on how much wealth to place on the coin outcome~\citep[see, e.g.,][]{orabona2023tight,waudbysmith2020estimating}.
\citet{Eugenio_JAR} recently showed that the family $\{S(\lambda,z) \mid \lambda \in [0,1]\}$ are the only admissible e-values for the hypothesis $\calH_0(z)$, in the sense that for any e-variable $E$ for $\calH_0(z)$, there exists $\lambda \in [0,1]$ such that $E\leq S(\lambda,z)$ almost surely. 
This implies that it suffices to focus on e-variables of the form \eqref{eq:def_coin_betting_e-variable}, reducing the problem of finding powerful e-variables to choosing suitable betting parameters $\lambda \in [0,1]$. 

\begin{remark}[Betting on bets]\label{remark:betting_on_bets}
    The building block e-variable~\eqref{eq:def_coin_betting_e-variable} has a particularly natural interpretation in the testing-by-betting language~\citep{shafer2019game}. 
    Given two uncertain prospects $X$ and $Y$ (say, stock and bond returns), suppose a forecaster claims that the probability of $X$ falling below a threshold $z$ (say, 0\%) is smaller than that of $Y$. 
    He then proposes a \emph{double-or-nothing-or-push} bet, on what are essentially ``bets'' themselves, with three possible outcomes. 
    If $X \leq z$ but $Y > z$, then the money is doubled; if $X > z$ but $Y \leq z$, then the money is lost; if both fall above or below $z$, then nothing happens.
    The skeptic, endowed with a dollar to spare on this bet, chooses to invest a $\lambda$-fraction of her money on this bet and keep $(1-\lambda)$ in her pocket. 
    Then, denoting the ternary outcome as $D(z) = \indicator{X \leq z} - \indicator{Y \leq z}$, the e-variable~\eqref{eq:def_coin_betting_e-variable} is precisely the money she ends up with after a single round of this bet: $S(\lambda, z) = 1 + \lambda D(z) = (1 - \lambda) + \lambda [1 + D(z)]$. 
    Over repeated rounds, the skeptic's wealth quantifies evidence against the forecaster's claim, which corresponds to $\calH_0(z)$ in~\eqref{eq:FSD_null_at_thresholds}. 
\end{remark}

We will repeatedly use the lemma below, which follows from the fact that $\calH_0$ forms an intersection hypothesis and that e-variables are closed under taking mixtures~\citep[see, e.g.,][]{vovk2021evalues}. 
\begin{lemma}\label{lem:finite_mixtures_are_sufficient_to_test_FSD_null}
    Given any $\psi \in \Pi(\calZ)$, the mixture $\int S(\lambda(z), z) \diff \psi(z)$ is an e-variable for $\calH_0$, where for each $z \in \calZ$, $S(\lambda(z), z)$ is the building-block e-variable~\eqref{eq:def_coin_betting_e-variable} for any $\lambda(z)\in [0,1]$.
\end{lemma}

By Lemma~\ref{lem:finite_mixtures_are_sufficient_to_test_FSD_null}, to test the global $\calH_0$, it suffices to come up with good e-variables $S(\lambda(z),z)$ at all thresholds and average the obtained evidence. 
In betting terms, the skeptic now splits the initial capital across thresholds $z \in \calZ$ (each representing a bet) according to weights $\psi$.

The next theorem shows how this idea translates directly when data arrives sequentially. 
For $z \in \calZ$ and $\lambda \in [0,1]$, we let $S_t(\lambda,z) \coloneqq S(\lambda, z; X_t, Y_t) = 1 + \lambda [\indicator{X_t \leq z} - \indicator{Y_t \leq z}]$.

\begin{theorem}
\label{thm:Test-supermartingale-for-the-FSD_null} 
Let $(\psi_t)_{t\in \N} \subseteq \Pi(\calZ)$, and $(\lambda_t(z))_{t\in \N} \subseteq [0,1]$, $z\in \calZ$, be predictable. Then,
 \begin{equation}\label{eq:mixture_seq_e_variable}
    S_t \coloneqq S_t(\lambda_t, \psi_t) \coloneqq \int_{\calZ} S_t(\lambda_t(z), z) \diff\psi_t( z), \quad t\in \N,
\end{equation}
are sequential e-variables for $\calH_0$, and its running product $E_t = \prod_{\ell=1}^t S_\ell,\, t\in \N,$ forms a test SM for $\calH_0$. 
In particular, $(E_t)_{t\in \N}$ yields a level-$\alpha$ sequential test for $\calH_0$ at any $\alpha \in (0,1)$.
\end{theorem}

\begin{proof}
$(S_t)_{t\in \N}$ are sequential e-variables by Lemma \ref{lem:finite_mixtures_are_sufficient_to_test_FSD_null} and predictability of $(\psi_t)_{t\in \N}$, and $(\lambda_t(z))_{t\in \N}$. 
Their running product forms thus a test SM satisfying Ville's inequality~\eqref{eq:Ville}. 
\end{proof}

Following Remark~\ref{remark:betting_on_bets}, we can interpret Theorem \ref{thm:Test-supermartingale-for-the-FSD_null} as initiating a protocol for betting games. 
\begin{game}[Testing SD by betting]\label{game:test_fsd}
    Let $E_0=1$ be the initial capital. For round $t = 1, 2, \dotsc$: 
        \begin{enumerate}[(a)]
            \item DM announces a mixture $\psi_t\in \Pi({\calZ})$ and the bets $\lambda_t(z) \in [0,1]$, for $z\in \calZ$.
            \item Reality announces the paired outcome $(x_t, y_t)$.
            \item DM's wealth multiplies by $S_t = S(\lambda_t, \psi_t; x_t, y_t)$ defined in~\eqref{eq:mixture_seq_e_variable}: $E_t= E_{t-1} \cdot S_t$. 
        \end{enumerate}
\end{game}
The DM now plays a sequence of betting games as a \emph{skeptic}, starting with unit capital. 
They can choose a different mixture at each round and bet any fraction of money earned from the previous rounds. 
Over time, DM's total wealth serves as evidence against the FSD null~\eqref{eq:def_FSD_null}. 
Theorem~\ref{thm:Test-supermartingale-for-the-FSD_null} establishes the anytime-validity of interpreting this wealth as statistical evidence across time: regardless of the choice of mixtures and betting strategies, the skeptic is not expected to make money \emph{ever} if the claim of dominance were true. 

\addtocounter{example}{-1}
\begin{example}[continued]
In the antimonotonic case in Example~\ref{ex:anticorr}, we do not expect to collect evidence at the threshold $z=1/3$, since with probability $1/2$, we will observe $X_t=2/3$ and $Y_t=1/3$, in which case $S_t(\lambda_t, 1/3) \leq 1$ for any $\lambda_t>0$. 
The DM can instead update the mixture predictably to avoid betting on this threshold (where the CDFs of $X$ and $Y$ touch) and instead bet on other thresholds, namely $z=0$ and $z=2/3$.
\end{example}

\subsection{Choosing log-optimal betting strategies}

Theorem \ref{thm:Test-supermartingale-for-the-FSD_null} shows that the running product $(E_t)_{t\in \N}$ forms a test SM for $\calH_0$, for any choice of predictable mixtures and bets. 
We now discuss sensible choices which allow the DM in Game \ref{game:test_fsd} to reject the null as fast as possible. 
We first characterize the optimal betting parameter $\lambda^\star(z)$ with respect to some fixed threshold $z\in \calZ$ and propose a natural empirical predictable plug-in version of it. 
Later, we propose how to (sequentially) average evidence over the thresholds to obtain an anytime-valid test with asymptotic power one.

\subsubsection{The growth-rate optimal (GRO) betting strategy and an e-process}\label{subsec:GRO_fixed_threshold}

For a fixed $z \in \calZ$, we want to choose $\lambda$ such that $S(\lambda, z)$ grows as fast as possible under the alternative
$\calH_1(z) = \calH_0(z)^c
    = \{\Qsymb \in \frakB \mid \Qsymb(X>z) < \Qsymb(Y>z)\}.
$
For $\Qsymb \in \calH_1(z)$, we choose $\lambda\in [0,1]$ to maximize the \emph{expected growth rate}, or \emph{e-power} \citep{shafer2011test,grunwald2019safe,vovk2024nonparametric,larsson2025numeraire}:
\begin{equation}\label{eq:power_target_FSD}
g(\lambda) = \Ex{\mathbb{Q}}{\log S(\lambda, z)}= p(z)\log(1+\lambda) + q(z)\log(1-\lambda) \,, \quad \lambda\in [0,1],
\end{equation} 
where $p(z)=\Qsymb(X\leq z<Y)$ and $q(z)=\Qsymb(Y\leq z<X)$. 

\begin{proposition}[GRO e-variable for $\calH_0(z)$ vs.\ $\calH_1(z)$]\label{ppn:GRO_betting_parameter_FSD}
    For any $\Qsymb \in \calH_1(z)$, the parameter
    \begin{equation}\label{eq:GRO_lambda}
        \lambda^\star(z) 
    = \frac{\Qsymb(X \leq z < Y) - \Qsymb(Y \leq z < X)}{\Qsymb(X \leq z < Y) + \Qsymb(Y \leq z < X)} 
    = \frac{p(z)-q(z)}{p(z)+q(z)}\in (0,1]
    \end{equation}
    yields the GRO e-variable $S(\lambda^\star(z), z)$ for testing $\calH_0(z)$ against the simple alternative $\{\Qsymb\}$. 
\end{proposition}
\begin{proof} 
   For any $\Qsymb \in \calH_1(z)$, we have $p(z)+q(z)\geq p(z)-q(z) = \Qsymb(X\leq z) - \Qsymb(Y\leq z) >0$. It follows that $g(\lambda)$ in~\eqref{eq:power_target_FSD} is strictly concave with a unique maximizer $ \lambda^\star(z) \in (0,1]$.
\end{proof}

\begin{remark}
    Given any $\Qsymb \in \calH_1(z)$, denote $\theta(z) = \Qsymb(X \leq z, Y \leq z)$ as the (lower-tail) \emph{coupling} probability between the marginals $\Qsymb_X$ and $\Qsymb_Y$ at threshold $z$. 
    Since $p(z) = \Qsymb(X \leq z) - \theta(z)$ and $q(z) = \Qsymb(Y \leq z) - \theta(z)$, the GRO bet~\eqref{eq:GRO_lambda} for $\calH_0(z)$ equivalently reads as
    \begin{equation*}
        \lambda^\star(z) = \frac{\Qsymb(X \leq z) - \Qsymb(Y \leq z)}{\Qsymb(X \leq z) + \Qsymb(Y \leq z) - 2\theta(z)}. 
    \end{equation*}
    This reveals that, even though $\calH_0(z)$ is only a statement about the marginals, the GRO bet incorporates cross-dependence, specifically through the coupling probability $\theta(z)$.
\end{remark}

\begin{example}
    Fix $c > 1$, and say the marginals are $\Qsymb_X \equiv \mathsf{Unif}[0,1]$ and $\Qsymb_Y \equiv \mathsf{Unif}[0, c]$. 
    At any $z \in (0, c)$, we have $\Qsymb \in \mathcal{H}_1(z)$; the comonotone and antimonotone cases reduce to $Y = cX$ and $Y = c(1-X)$, $\Qsymb$-a.s., respectively.
    Focusing on $z \in (0,\frac{c}{c+1}]$ for illustration, the GRO bet~\eqref{eq:GRO_lambda} is $\lambda^\star(z) = \frac{1-1/c}{1 + 1/c - 2\gamma(z)}$, where $\gamma(z) = \theta(z)/\Qsymb(X \leq z) = \Qsymb(Y \leq z \mid X \leq z) \in [0, 1/c]$ captures cross-dependence. 
    In particular, $\lambda^\star(z)$ is strictly increasing in $\gamma(z)$, smoothly interpolating between the antimonotone case ($\gamma(z)=0$), where it is the most conservative ($\lambda^\star(z) = \frac{1-1/c}{1+1/c}$), and the comonotone case ($\gamma(z)=1/c$), where it is the most aggressive ($\lambda^\star(z) = 1$). 
    (In fact, these two extreme cases respectively match the lower and upper Fr\'echet--Hoeffding bounds on $\theta(z)$, but we spare the details here.)
\end{example}

\begin{remark} 
For any $p(z)$ and $q(z)$ in the alternative, we have $p(z)>q(z)$, and the average $m(z)=(p(z)+q(z))/2$ is the closest element in the null and thus the hardest candidate to test against. 
The GRO e-variable equals $p(z)/m(z)$ if $X \leq z < Y$; $q(z)/m(z)$ if $Y \leq z < X$; and $1$ otherwise. 
This is a well-studied likelihood ratio between the given alternative and its particular projection onto the null, known as the \emph{reverse information projection} \citep{grunwald2019safe,larsson2025numeraire}.
\end{remark}

Proposition \ref{ppn:GRO_betting_parameter_FSD} yields the following natural predictable plug-in estimator for $\lambda^\star(z)$, 
\begin{equation}\label{eq:plugin_GRO_lambda}
    \hat{\lambda}_t^\textsf{GRO}(z)
    = (1-c) \land\left(\frac{\hat{\mathbb{Q}}_{t-1}(X \leq z < Y)-\hat{\mathbb{Q}}_{t-1}(Y \leq z < X)}{\hat{\mathbb{Q}}_{t-1}(X \leq z < Y)+\hat{\mathbb{Q}}_{t-1}(Y \leq z < X)}\right)_+ \in [0,1-c],
\end{equation}
for a small $c>0$, $\hat{\mathbb{Q}}_{t-1}$ the empirical distribution of $(X,Y)^{t-1}$, and $0/0:=0$. We impose the upper bound $ \hat{\lambda}_t^\textsf{GRO}(z) \leq 1-c$ to avoid e-values becoming exactly zero.

\subsubsection{Asymptotic power-one guarantee for predictably mixed GRO e-processes}\label{sec:power_one_theorem}

Following Theorem~\ref{thm:Test-supermartingale-for-the-FSD_null} and the derivations from the previous section, at $t\geq 1$, for some predictable mixture distribution $\psi_t \in \Pi(\calZ)$, we consider the $\calZ$-mixture sequential e-variable
\begin{equation}\label{eq:mixture_e-variable}
    S_t = \int S_t(\hat{\lambda}_t^\textsf{GRO}(z),z)\diff\psi_t( z), 
\end{equation}
as our measure of evidence against the global null $\calH_0$. 

Our main result shows that, under a mild assumption on the predictable mixture distributions, the running product of the e-variables \eqref{eq:mixture_e-variable} yields powerful e-processes.  
We say that an e-process $(E_t)_{t\in\N}$ is \emph{powerful}~\citep{ramdas2021testing} if it yields a sequential test of asymptotic power one, that is, for any $\mathbb{Q}\in\calH_1$ and $\alpha\in(0,1)$, we have $\Q(\tau_\alpha<\infty)= 1$ for $\tau_\alpha=\inf\{t\geq 1: E_t \geq 1/\alpha\}$. 

\begin{theorem}[Powerful GRO e-process with predictable $\calZ$-mixtures]\label{thm:asymptotic_power_one} 
Let $\Qsymb\in \calH_1$. 
Suppose we have a predictable sequence $(\psi_t)_{t\in \N}\subseteq \Pi(\calZ)$ for which there exist $\epsilon,\delta>0$ and $N_0 \in \N$ s.t.
\begin{equation}\label{eq:cond_power_one_theorem}
\psi_t(\{z:\Qsymb_X(z) -\Qsymb_Y(z)> \epsilon\})>\delta, \quad \textrm{for all }t\geq N_0\,.
\end{equation}
Then, the running product of e-variables $E_t=\prod_{\ell=1}^t S_\ell$, $t\in \N$, is a powerful e-process that grows at an exponential rate under $\Qsymb$, that is
\begin{equation*}
\Qsymb\inparen{\liminf_{t\to\infty}t^{-1}\log E_t >0} = 1\,.
\end{equation*}
\end{theorem}

Condition~\eqref{eq:cond_power_one_theorem} is rather mild and merely requires that the mass of $\psi_t$ eventually remains bounded away from zero on the \emph{violation set} $V=\{z : \Qsymb_X(z) > \Qsymb_Y(z)+\epsilon\}$, the region where the violation of SD is detectable. 
We remark that, even if the mixture ``wastes'' at most a fraction $(1-\delta)$ of wealth to non-violation regions, we obtain a powerful e-process because the GRO bets~\eqref{eq:plugin_GRO_lambda} on these regions tend to zero, while GRO bets on the violation set yield exponentially growing wealth.

In the special case of choosing a fixed mixture distribution $\psi_t=\psi$ with full support, condition~\eqref{eq:cond_power_one_theorem} is guaranteed for small enough $\varepsilon$, and Theorem~\ref{thm:asymptotic_power_one} may be seen as an instance of the \emph{method of mixtures}~\citep{robbins1970boundary,howard2021time}. 
However, in the latter, one typically integrates over betting parameters or model parameters at the level of entire processes rather than individual increments. Thus, there remains a distinction to this work, where we instead consider roundwise mixtures over the testing thresholds, and choose the betting parameters via the empirical plug-in version of the GRO bets introduced in the previous section.

\subsubsection{On predictable weights for aggregating over the testing thresholds}\label{sec:agg_over_z}

Next, we propose an adaptive weighting method from online learning to effectively combine evidence across $\calZ$. 
At $t\in \N$, we fix $\tilde{\calZ}_t\subseteq \calZ$ given by $z_1^t < z_2^t < \ldots < z^t_{n_t-1} <z_{n_t}^t$, 
for some $n_t\geq 1$. 
Then, we use exponential weights based on the empirical CDF difference: 
\begin{equation}\label{eq:predictable_weights}
    w_t(z_i^t) \propto \exp\left\{\eta_{t}(z_i^t)\big[\hat{F}_{X}^{t-1}(z_i^t)-\hat{F}_{Y}^{t-1}(z_i^t)\big]\right\}, \quad \quad i=1, \ldots, n_t\,,
\end{equation}
for some predictable tuning parameters $\eta_{t}(z_i^t)>0$. We then suggest using
\begin{equation}\label{eq:pred_mixture}
    \psi_t=\sum_{i=1}^{n_t}  w_t(z_i^t) \delta_{z_i^t}, 
\end{equation}
to compute the e-variable $S_t$ at \eqref{eq:mixture_e-variable} and $E_t=\prod_{\ell=1}^t S_\ell$ as discussed in Theorems~\ref{thm:Test-supermartingale-for-the-FSD_null} and~\ref{thm:asymptotic_power_one}.

The choice of weights in~\eqref{eq:predictable_weights} is inspired by the exponential weights aggregation (EWA) method~\citep{vovk1990aggregating}, which achieves minimax-optimal regret in online learning~\citep{cesabianchi2006prediction}, 
and has been adapted to analogous sequential testing problems~\citep[e.g.,][]{orabona2023tight,waudbysmith2020estimating,waudbysmith2025universal}.
The rationale behind~\eqref{eq:predictable_weights} is to place larger bets on the region where SD is violated—that is, on thresholds for which $F_X(z)-F_Y(z)$ is large—and smaller bets on other regions. 

The original EWA scheme is known to be computationally expensive in similar settings~\citep{cesabianchi2006prediction}, as it requires calculating all past gains at each newly selected threshold $z$. 
Instead, we use exponential weights of the standardized CDF difference, by setting $\eta_t(z_i^t)$ in~\eqref{eq:predictable_weights} as
\begin{equation}\label{eq:standardized_lr}
    \eta_t(z_i^t) =
    \frac{\eta}{ \widehat{\textrm{sd}}[\hat{F}_{X}^{t-1}(z_i^t)-\hat{F}_{Y}^{t-1}(z_i^t)]},
\end{equation}
for a single pre-specified parameter $\eta>0$, where the estimated standard deviation is clipped from below to avoid division by zero. 
The resulting weights involve a self-normalized test statistic, standard in martingale concentration inequalities \citep[e.g.,][]{howard2021time} and common in the literature on SD \citep[e.g.,][]{davidson2013testing,zhuang2024tests}.  
Later, in our simulations (Section~\ref{sec:simulations}), we will compare the adaptive approach~\eqref{eq:pred_mixture} with exponential weights against other baselines. 

\begin{remark}[On the choice of the testing thresholds $\tilde\calZ_t$]
Since $z \mapsto S_t(\hat{\lambda}_t^\textsf{GRO}(z),z)$ is piecewise constant with potential jumps at $(X,Y)^t$, 
it suffices to have access to the probabilities $\psi_t([z_j^t,z_{j+1}^t)),\, j=1, \ldots, n_t-1\,$, where $z_1^t < z_2^t < \ldots < z^t_{n_t}$ are the ordered observations in $(X,Y)^{t}$ such that the mixture~\eqref{eq:mixture_e-variable} can be computed exactly.
Here, $\psi_t$ needs to be predictable, and needs thus in particular to be fixed \emph{before} observing $(X_t,Y_t)$. 
Since we can always interpret the weights in \eqref{eq:predictable_weights} as weights on the whole intervals $[z_j^t,z_{j+1}^t)$ and distribute the mass uniformly over them, any choice of weights on the ordered observations in $(X,Y)^{t-1}$ can be used to derive the target probabilities $\psi_t([z_j^t,z_{j+1}^t))$, including the interval probabilities with the non-predictable endpoints $X_t$ and $Y_t$. 
However, as the sample size increases, the impact of the additional interval splitting at $X_t$ and $Y_t$ becomes negligible; we simply apply the method as described above, taking $\tilde\calZ_t$ to be the ordered observations in $(X,Y)^{t-1}$. 
By subsequently adding more observations from the pooled sample distribution $(F_X+F_Y)/2$ and exploiting known optimality properties of the proposed weighting schemes, it is reasonable to expect that the regularity condition~\eqref{eq:cond_power_one_theorem} in Theorem~\ref{thm:asymptotic_power_one} is satisfied. 
\end{remark}

\begin{remark}[Finite known support]\label{remark:finite_support} 
An important special case is when the outcomes have finite and known support $\calZ = \{z_1, \dotsc, z_m\}$, as is the case for ordinal data. 
(The baseball application in Section~\ref{sec:real_data_exp} is a concrete example.)
In such cases, we can simply fix the choice of thresholds to be $\{z_1, \dotsc, z_{m-1}\}$ (there is no violation at $z_m$).  
The adaptive weighting scheme~\eqref{eq:predictable_weights} with weights~\eqref{eq:standardized_lr} is still useful, although it can only outperform simple uniform weighting, $\psi_t \equiv \mathsf{Unif}(\{z_1, \dotsc, z_{m-1}\})$, by a constant factor $O(\log m)$ at most. 
\end{remark}

\subsection{Testing higher-order SD and integral stochastic orders}\label{sec:testing_kSD_summary}

While first-order SD is a natural benchmark, many applications require weaker notions of dominance that remain meaningful when the distribution functions cross.
One such problem is testing second-order SD (SSD), which corresponds to dominance in expected utility according to any increasing \emph{and concave} function. 
For two prospects with the same mean, SSD implies that the dominant prospect has no larger variance (if it exists), reflecting \emph{risk aversion}~\citep{rothschild1970risk}. 
FSD implies SSD; more generally, higher-order SD is obtained by recursively integrating the CDFs and requiring a pointwise ordering of the resulting functions, or equivalently, by considering progressively smaller classes of utility functions.

We can analogously, although nontrivially, generalize our approach and the power-one result of Theorem~\ref{thm:asymptotic_power_one} to sequentially testing the null of $k$-th order SD for any $k \in \N$.
We briefly summarize our results here and defer a comprehensive treatment of the problem to Supplement~\ref{sec:testing_higher_order_SD}. 
As in our FSD testing framework, we select a fixed or predictable mixture strategy over thresholds and multiply over time the mixture of asymptotically log-optimal e-variables to obtain a powerful e-process. 
The main difference is in the construction of the building-block e-variables: unlike in Proposition~\ref{ppn:GRO_betting_parameter_FSD}, we no longer have a closed-form expression of the GRO bet in terms of $\bbQ$ at each time/threshold. 
We thus propose to use a betting strategy based on universal portfolios (UP)~\citep{cover1991universal}, which yield an asymptotically log-optimal sequence of bets~\citep{waudbysmith2025universal}. 
Our main result is that the resulting e-process is powerful when the support has a finite lower bound.  
We also sketch extensions to other popular integral stochastic orders, namely the increasing convex order and the infinite order.

\section{Simulations}\label{sec:simulations}

In our simulations, we validate five e-process variants, three of which use predictable mixtures as summarized in Table~\ref{tbl:methods_summary}. 
The method used in Figure~\ref{fig:comparison_classical_tests} is the \emph{AdaGRO-Exp} method based on the aforementioned idea of exponential weights~\eqref{eq:predictable_weights}. 
We remark that the \emph{AdaGRO-Greedy} method is inspired by \citet{shekhar2023nonparametric}'s approach, although there are at least two differences from its vanilla adaptation, as summarized in Remark~\ref{rem:sr_comparison}. 

\begin{table}[t]
    \centering
    \begin{tabular}{llll}
    \toprule
    \textbf{Method} & $\lambda_t(z)$ & \textbf{Use \eqref{eq:pred_mixture}?} & \textbf{Mixture weights $w_t(z_i^t)$} (if applicable) \\
    \midrule
    AdaGRO-Exp & $\hat\lambda_t^\textsf{GRO}(z)$ & Yes & Exponential w/ {standardized} ECDF diff.~\eqref{eq:standardized_lr} \\
    AdaGRO-Linear & $\hat\lambda_t^\textsf{GRO}(z)$ & Yes & Linear in bets, $w_t(z_i^t) \propto \hat\lambda_t^\textsf{GRO}(z_i^t)$ \\
    AdaGRO-Greedy & $\hat\lambda_t^\textsf{GRO}(z)$ & Yes & All-in on $z_i^t$ w/ maximum ECDF difference \\
    GRO & $\hat\lambda_t^\textsf{GRO}(z)$ & No & n/a (equal weights on a grid of $z$) \\
    Constant & $0.1$ & No & n/a (equal weights on a grid of $z$) \\
    \bottomrule
    \end{tabular}
    \caption{E-processes examined in simulations. $\hat\lambda_t^\textsf{GRO}(z)$ is the plug-in GRO bet~\eqref{eq:plugin_GRO_lambda}.}
    \label{tbl:methods_summary}
\end{table}

To examine the validity and power of the e-processes in sequential experiments, we compute two metrics. 
Both metrics are standard in the SAVI literature~\citep[e.g.,][]{howard2021time,vovk2024nonparametric} and can be estimated, for each sample size $t$, as averages over repeated simulations. 
\begin{enumerate}
    \item \emph{Ville error}: $\Psymb\inparen{\exists \ell \leq t: E_\ell \geq {1}/{\alpha}}$. 
    This is the notion of type I error under continuous monitoring, capturing the anytime-validity of sequential tests under the null $\Psymb$. 
    \item \emph{e-power}: $\Ex{\Qsymb}{\log E_t}$. 
    This captures how quickly the e-process is expected to grow under the alternative $\Qsymb$.
    We expect powerful e-processes to have linearly increasing e-power over time, indicating a positive growth rate and an almost sure rejection at any level $\alpha$ in the long run (i.e., yield a test of power one).
\end{enumerate}

\subsection{Simulation \#1: Robustness against contact sets}\label{sec:sim_kuniform}

Our main simulation examines the challenging case where there is substantial ``contact'' between the CDFs. 
This setup was popularized by \citet{linton2010improved} for non-sequential SD testing, where the presence of contact sets is known to determine the limiting distribution of the test statistic. 
While the classical framework does not extend to sequential testing, we borrow the setup to examine the behavior of different e-processes in such scenarios.

We consider distributions on $[0,1]$ whose CDFs can have a ``kink'' at a specified value. 
Fix $F_Y$ as the CDF of the $\mathsf{Unif}[0,1]$ distribution (no kink). 
For each kink point $z_0 \in \{0.0, 0.1, 0.2, 0.5, 1.0\}$, we define a mixed distribution for $X$, having the form $F_X(z) = \frac{1}{2} z + \frac{1}{2}z_0$ for $z \in [0, z_0]$ and $F_X(z) = z$ for $z \in (z_0, 1]$.
Figure~\ref{subfig:kuniform_cdfs} shows the CDF of $X$ alongside that of $Y$.
For $z_0 > 0$, we have that $X \fsd Y$ and $Y \notfsd X$; when $z_0=0$, the two distributions coincide. 
As $z_0$ increases, the non-dominance (violation) region $[0, z_0]$ increases in size, so it should be easier to reject the FSD null $\calH_0: Y \fsd X$; as $z_0$ decreases, the size of the contact between the two CDFs increases, so the ``upside'' of $Y$ over $X$ becomes harder to detect. 

In each of 100 repeated simulations and for each value of $z_0$, we generate a sequence of i.i.d.~observations $(X_t, Y_t)_{t=1}^T$, $T=2,000$, where $X_t \sim F_X$ and $Y_t \sim F_Y$ independently, and compute the e-processes for testing $\calH_0: Y \fsd X$ and $\calH_0: X \fsd Y$ at each time $t$. 
The initial search region, $\tilde\calZ_0$, is placed at $K=101$ equidistant points on $[0, 1]$ including the endpoints. 
For the non-adaptive e-processes, fixed mixtures are computed over this search region. 
For the adaptive e-processes, predictable mixtures~\eqref{eq:pred_mixture} are computed on an updated search region $\tilde\calZ_t$ for $t>50$ after a burn-in period. 
Here, instead of using all observations in $(X,Y)^{t-1}$ to compute the mixtures, we choose $\tilde\calZ_t$ to be the empirical quantiles of $(X,Y)^{t-1}$ on an equidistant quantile grid of size $K=101$, thereby reducing the time complexity from $\mathcal{O}(t^2)$ to $\mathcal{O}(Kt)$ while obtaining nearly identical results. 

\begin{figure}[t]
    \centering
    \begin{subfigure}[t]{\textwidth}
        \centering
        \includegraphics[width=\textwidth, keepaspectratio]{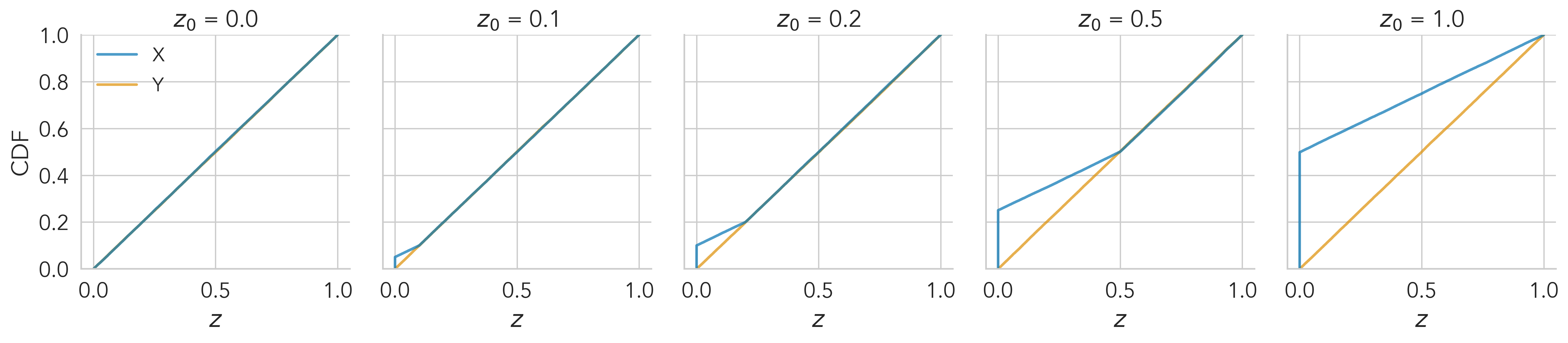}
        \caption{CDFs of $X$ and $Y$ for $z_0 \in \{0.0, 0.1, 0.2, 0.5, 1.0\}$.}
        \label{subfig:kuniform_cdfs}
    \end{subfigure}
    \begin{subfigure}[t]{\textwidth}
        \centering
        \includegraphics[width=\textwidth, keepaspectratio]{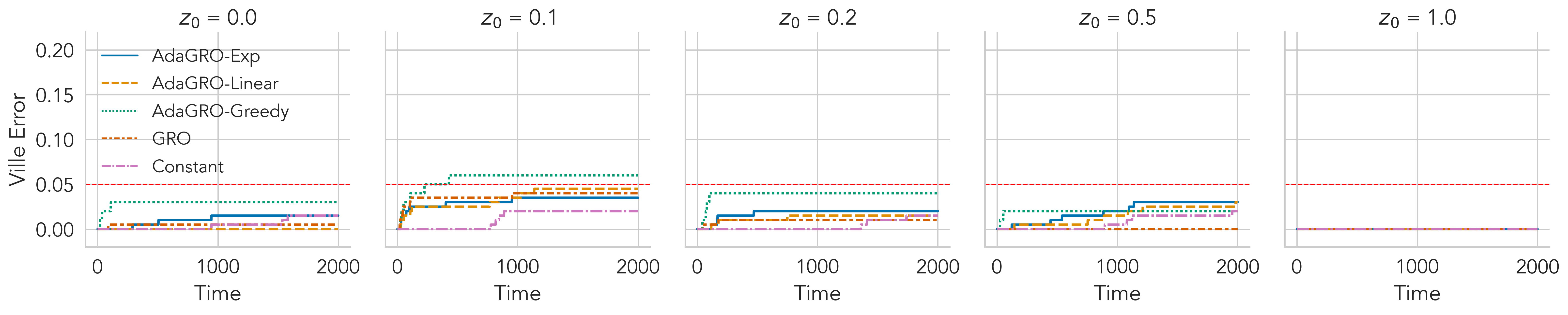}
        \caption{Ville error for testing $\calH_0: X \fsd Y$.}
        \label{subfig:kuniform_ville}
    \end{subfigure}
    \begin{subfigure}[t]{\textwidth}
        \centering
        \includegraphics[width=\textwidth, keepaspectratio]{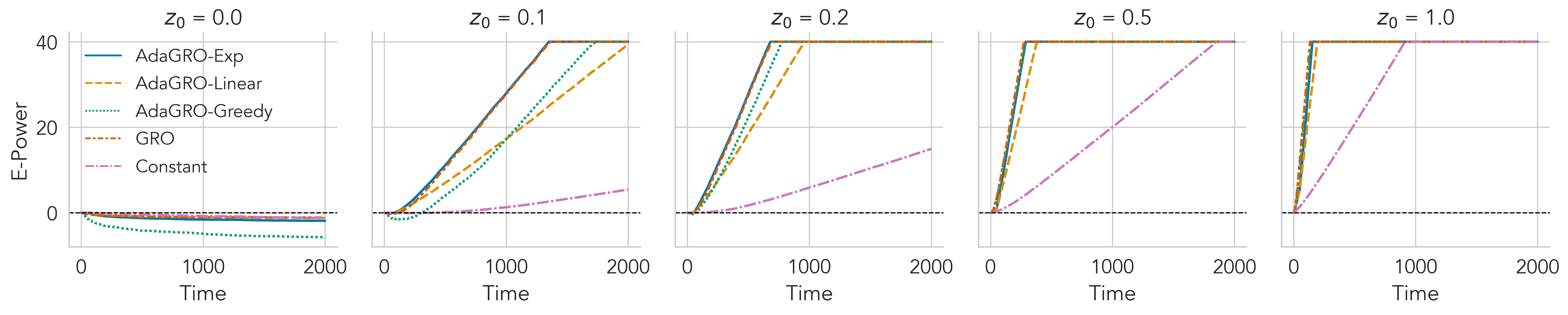}
        \caption{E-power against $\calH_0: Y \fsd X$.}
        \label{subfig:kuniform_epower}
    \end{subfigure}
    \caption{\textit{All GRO e-processes grow quickly even when the contact set is large, with the adaptive variant with exponential weights having the largest e-power.}
    These plots summarize simulations where $F_Y \equiv \mathsf{Unif}[0,1]$ and $F_X$ is a piecewise uniform distribution with a kink point at $z_0 \in [0,1]$. 
    As $z_0$ decreases, the two CDFs have more ``contact'' and the testing problem becomes more challenging (under the null, to not reject; under the alternative, to reject quickly). 
    The e-power is truncated at 40 for ease of visualization.}
    \label{fig:kuniform}
\end{figure}

Figure~\ref{subfig:kuniform_ville} shows that all e-processes maintain anytime-validity across all values of $z_0$, as the Ville error is controlled at level $\alpha=0.05$ (red horizontal line) at all times. 
When $z_0=1.0$ (no contact except at $z=1$; the ``easy'' case), all e-processes have zero Ville error up to $T$. 
Figure~\ref{subfig:kuniform_epower} further shows that all GRO e-processes grow quickly against the FSD null, much faster than the constant variant, even when the contact between the two CDFs is very large. 

Overall, two variants achieve the highest e-power across scenarios: the adaptive GRO e-process with exponential weights~\eqref{eq:predictable_weights} and the non-adaptive GRO e-process. 
The adaptive methods with different weighting schemes (AdaGRO-Linear and AdaGRO-Greedy) exhibit suboptimal growth rates, especially when the contact is large (e.g., $z_0=0.1$). 
The constant-bet e-process performs the worst in all scenarios, suggesting that the GRO betting strategy is crucial for achieving good power. 
The reason why both the adaptive and non-adaptive GRO e-processes work well in e-power terms is that, given bounded support, the initial set of thresholds already includes $z=0$, where there is the biggest gap between the two CDFs. 
Thus, even without predictable weights, putting constant weight on this threshold is enough (up to logarithmic terms). 
In our next simulation, we explore this gap in a different setup.

\emph{Comparisons with non-anytime-valid tests (Figure~\ref{fig:comparison_classical_tests}).} 
In Section~\ref{sec:introduction}, we compared the sequential test induced by the adaptive GRO e-process (AdaGRO-Exp) against three non-anytime-valid tests for SD, with the null being the $z_0=0.0$, $\calH_0: X \fsd Y$ case and the alternative the $z_0=0.2$, $\calH_0: Y \fsd X$ case.  
BD03 and LMW05 represent paired bootstrap and subsampling, respectively; LSW10 augments the subsampling strategy with a separate estimation scheme for the non-dominance region. 
These methods are asymptotically valid under several regularity conditions, and by design, they are not valid under continuous monitoring. 
We saw this clearly in Figure~\ref{fig:comparison_classical_tests} (left): the Ville errors of these three baselines, unlike that of the sequential e-test, quickly surpass the significance level ($\alpha=0.05$) and continue to increase (in logarithmic scale of time). 
The figure (right) also showed the rejection times, i.e., $\tau_\alpha = \inf\{t \geq 1: E_t \geq 1/\alpha\}$ at level $\alpha=0.05$, under the alternative. 
There, the mean rejection time of the e-test (116.7) was smaller than that of the best-performing baseline (126.0 for LMW05). 
Further, the distributions of rejection times for resampling-based procedures exhibited heavier tails, suggesting that they can sometimes take much longer to reject than sequential e-tests.

\emph{Comparisons with KS-type sequential tests.} 
In Supplement~\ref{sec:empirical_comparison_ks_tests}, we include additional experiments comparing the GRO e-processes with KS-type sequential tests~\citep{darling1968some,howard2022sequential,manole2023martingale,clerico2026uniform}. 
These baselines are described in depth in Supplement~\ref{app:ks_test}; we remark that the latter three are designed for time-uniform CDF estimation rather than SD testing. 
To summarize, we find that GRO e-processes (adaptive or not) are generally more powerful than these baselines, particularly when the non-dominance region is small (or when contact is large). 
For example, when $z_0=0.2$, both the AdaGRO-Exp and (non-adaptive) GRO e-processes reject the FSD null in $110$ observations on average, while the sequential tests based on CDF bands are expected to require anywhere from $390$ to $3,800$ observations to reject.

\emph{Testing higher-order SD.} In Supplement~\ref{app:sim_kuniform_ksd}, we validate the UP e-processes for 2-SD and 3-SD testing in the same setups, arriving at generally similar conclusions as in the FSD case. Unlike for FSD, the greedy mixture is the superior choice for $k$-SD.

\subsection{Simulation \#2: Adaptivity to non-dominance regions}\label{sec:sim_normal}

In this simulation, we examine the value of adaptive strategies via predictable weights (Theorem~\ref{thm:asymptotic_power_one}) using data with unbounded and continuous support ($\calZ = \R$). 
For each scenario, we generate $(X_t, Y_t)_{t=1}^T$, $T=2,000$, from an i.i.d.~bivariate Gaussian distribution with a high negative cross-correlation ($\rho=-0.9$).
While we use parametric data-generating distributions here for the sake of clarity, the theory suggests that the methods are largely agnostic to such parametric forms. 
In the case of Gaussian rv's, $X \sim \calN(\mu_X, \sigma_X^2)$ and $Y \sim \calN(\mu_Y, \sigma_Y^2)$, we know that $X \fsd Y$ if and only if $\mu_X \leq \mu_Y$ and $\sigma_X = \sigma_Y$. 

To examine adaptivity to non-dominance regions, we vary the initial search region $\tilde\calZ_0 \subset \R$, the initial set of testing thresholds placed at $K=101$ equidistant points. 
As before, non-adaptive e-processes use a fixed mixture over these initial points; for adaptive e-processes, after 50 rounds, these thresholds are updated based on the observed sample quantiles.

\begin{figure}[t]
    \centering
    \begin{subfigure}[t]{0.26\textwidth}
        \centering
        \includegraphics[width=\textwidth, keepaspectratio]{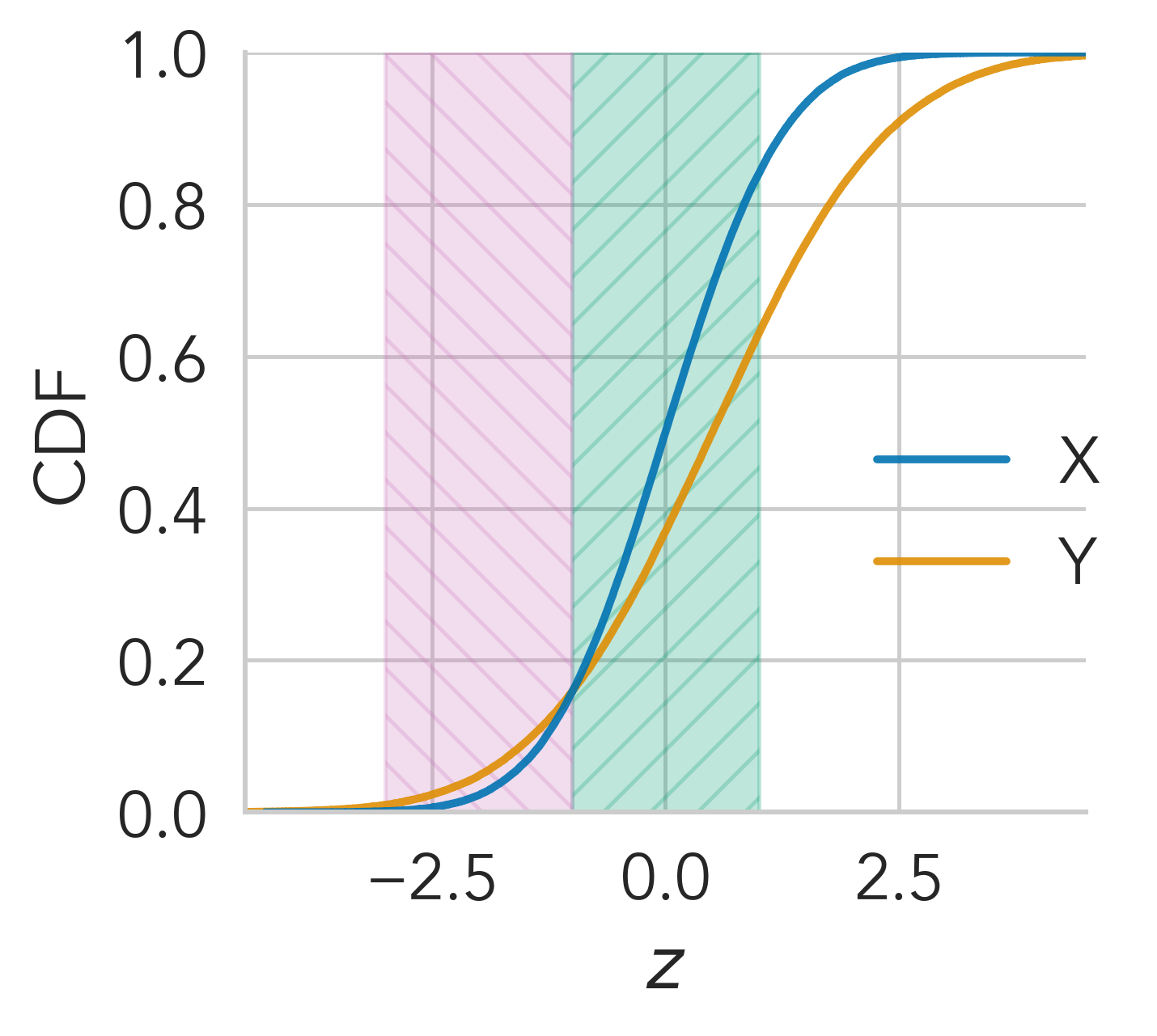}
        \caption{CDFs of $X$ and $Y$.}
        \label{subfig:normal_cdfs}
    \end{subfigure}
    ~
    \begin{subfigure}[t]{0.68\textwidth}
        \centering
        \includegraphics[width=\textwidth, keepaspectratio]{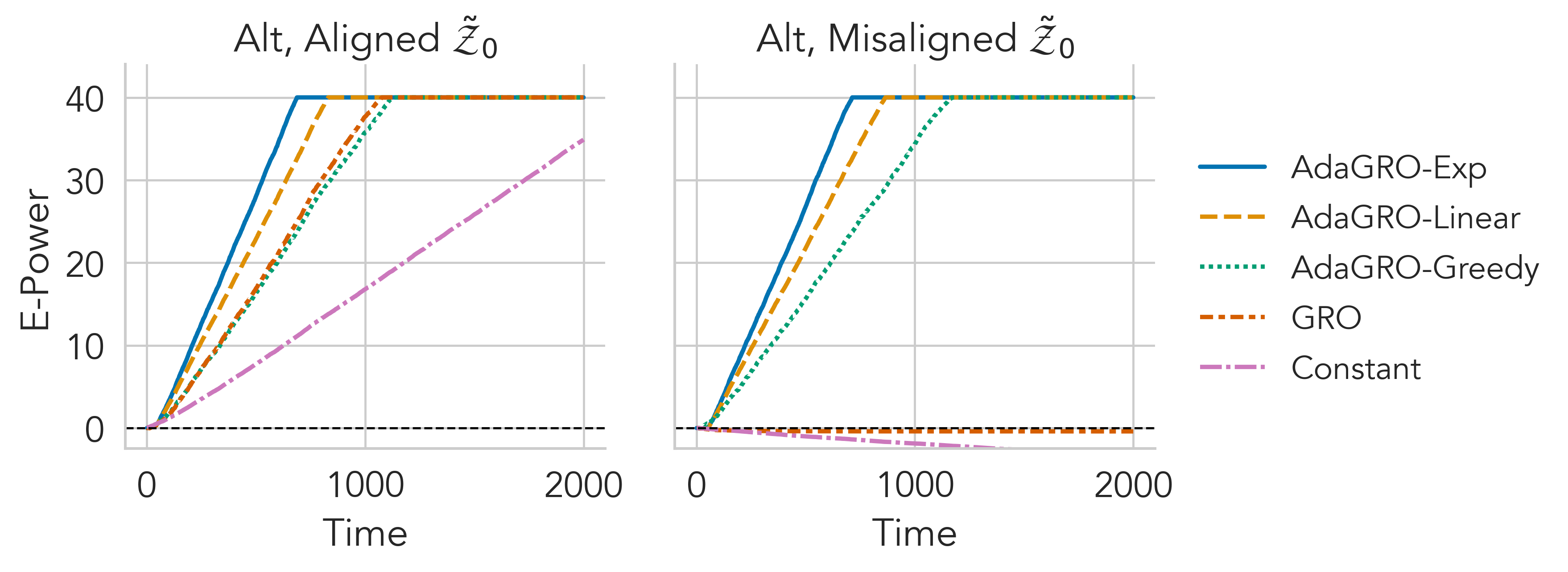}
        \caption{E-powers in the two scenarios.}
        \label{subfig:normal_epower}
    \end{subfigure}
    \caption{\textit{Adaptive GRO e-processes grow exponentially fast under the alternative, even when the non-dominance region is misaligned with the initial search interval. 
    } 
    In~\ref{subfig:normal_cdfs}, we plot the CDFs of $X$ and $Y$ and the two initial search intervals (green: aligned; purple: misaligned).
    In~\ref{subfig:normal_epower}, we plot the e-powers in two cases, depending on whether the initial search region is aligned/misaligned.
    The e-power is truncated at 40 for ease of visualization.}
    \label{fig:normal}
\end{figure}

Below, we let $X \sim \calN(0, 1)$ and $Y \sim \calN(0.25, 1.5^2)$ and test $\calH_0: Y \fsd X$; a powerful test should reject $\calH_0$.  
We consider two scenarios: one where the initial search region is \emph{aligned} with the non-dominance region, or $\tilde\calZ_0 \subset [-1, 1]$; another where the search region is \emph{misaligned}, or $\tilde\calZ_0 \subset [-3, -1]$. 
In Figure~\ref{fig:normal}, we plot the CDFs and the e-powers of all e-process variants in these two cases. 
The e-powers are averaged over 100 repeated simulations.

The results show a stark contrast between the adaptive and non-adaptive GRO e-processes. 
When the initial search region $\tilde\calZ_0$ is aligned, both adaptive and non-adaptive GRO e-processes grow quickly, with AdaGRO-Exp having a slight edge in terms of e-power ($\geq 40$ by $t\approx 750$). 
In contrast, when $\tilde\calZ_0$ is misaligned with the non-dominance region, the non-adaptive GRO e-process fails to grow at all, while all adaptive GRO e-processes grow almost as quickly as they did in the aligned case. 
This shows the crucial role of adaptivity in identifying and placing more weight (larger bets) on non-dominance regions.
We thus recommend using adaptive GRO e-processes when the support is unbounded or unknown. 

\section{Application: Monitoring the 3TTO effect in baseball}\label{sec:real_data_exp}

To illustrate a realistic use case, we apply our approach to a popular and ever-controversial decision problem in baseball analytics, known as the \emph{third-time-through-the-order (3TTO)} problem. 
In a baseball game, a starting pitcher cycles through the opposing lineup of nine batters, usually multiple (often 2--4) times, and the manager has to decide when to pull the pitcher out of the game, if at all. 
Whereas conventional wisdom viewed the starting pitcher as the ``ace'' of the team and thus advocated for keeping them in the game as long as possible, modern sabermetric evidence suggested that the performance of even the best pitchers often drops off sharply when they face the same batter for the third time~\citep{tango2007book}. 
This observed ``3TTO penalty'' became a hotly debated topic during the 2020 World Series, in which a manager took out the starting pitcher Blake Snell after his second TTO, citing the 3TTO penalty as the main reason for his decision~\citep{3TTO}. 
Snell's replacement immediately gave up two crucial runs that ultimately cost the team the series.
While the 3TTO penalty has become a popular belief among managers and fans, the question is far from settled: \citet{brill2023ttop} recently found that, after adjusting for confounders, there is little evidence of the purported strong discontinuity in performance between 1/2TTO and 3TTO. 

Beyond its relevance as a sequential decision problem, the 3TTO penalty problem provides an unorthodox yet effective use case for FSD testing. 
Since at-bat outcomes are ordinal, it is unclear how to aggregate them into a single metric; FSD avoids this issue by comparing the outcome distributions directly, without requiring a particular weighting scheme, the study of which is itself a major topic in baseball analytics~\citep[e.g.,][]{tango2007book}.

We formalize the sequential testing problem using five ordinal at-bat outcomes: an out (say, $z=0$), a walk ($z=1$), a single (1B; $z=2$), a double or triple (2B/3B; $z=3$), and a home run (HR; $z=4$). 
The larger the outcome, the better it is for the opposing batter, and thus the worse it is for the pitcher. 
As with any ordinal variable, the magnitude of $z$ does not matter (we are agnostic to, say, whether a HR is twice as valuable as a 1B). 
We exclude contextual events such as intentional walks, sacrifice flies/bunts, and fielding errors. 
Over the course of multiple at-bats across games, whenever a pitcher faces the same opposing batter for the third time, we observe the pair of 1TTO and 3TTO outcomes $(X_t, Y_t)$. 

To demonstrate our method, we proceed under an i.i.d.~assumption with some remarks. 
Within each game, whether outcomes against multiple opposing batters are i.i.d.~may depend on whether the pitcher is ``streaky''; on the other hand, controlling for the pitcher, batter-level outcomes are largely determined at the individual level. 
Across different games, one potential known issue is selection bias: pitchers can be pulled out of the game before their 3TTO. 
\citet[][Section 2.3]{brill2023ttop} note that, on average, worse pitchers are ``slightly more likely'' to be removed earlier in the game, and they focus on comparisons conditional on the pitcher getting to 3TTO. 
Here, we also focus on such a conditional comparison, and specifically on a single pitcher (Snell) rather than an average effect across many pitchers. 
At a high level, the selection bias would make it harder to detect (unconditional) 3TTO penalty \emph{if} the counterfactual 3TTO outcomes were worse than the observed ones. 

Motivated by our earlier discussion, we (retrospectively) apply our GRO e-process for testing FSD to the 3TTO problem for the pitcher Blake Snell in all his regular season starts from 2016 to 2025. 
This gives us a total of $T=849$ at-bat pairs that we observe sequentially in the order they occurred in the games. 
Since we have five discrete outcomes, we use exponential weights~\eqref{eq:predictable_weights} to adaptively split our bets across the first four at-bat outcomes.

\begin{figure}[t]
    \centering
    \includegraphics[width=0.9\textwidth,keepaspectratio]{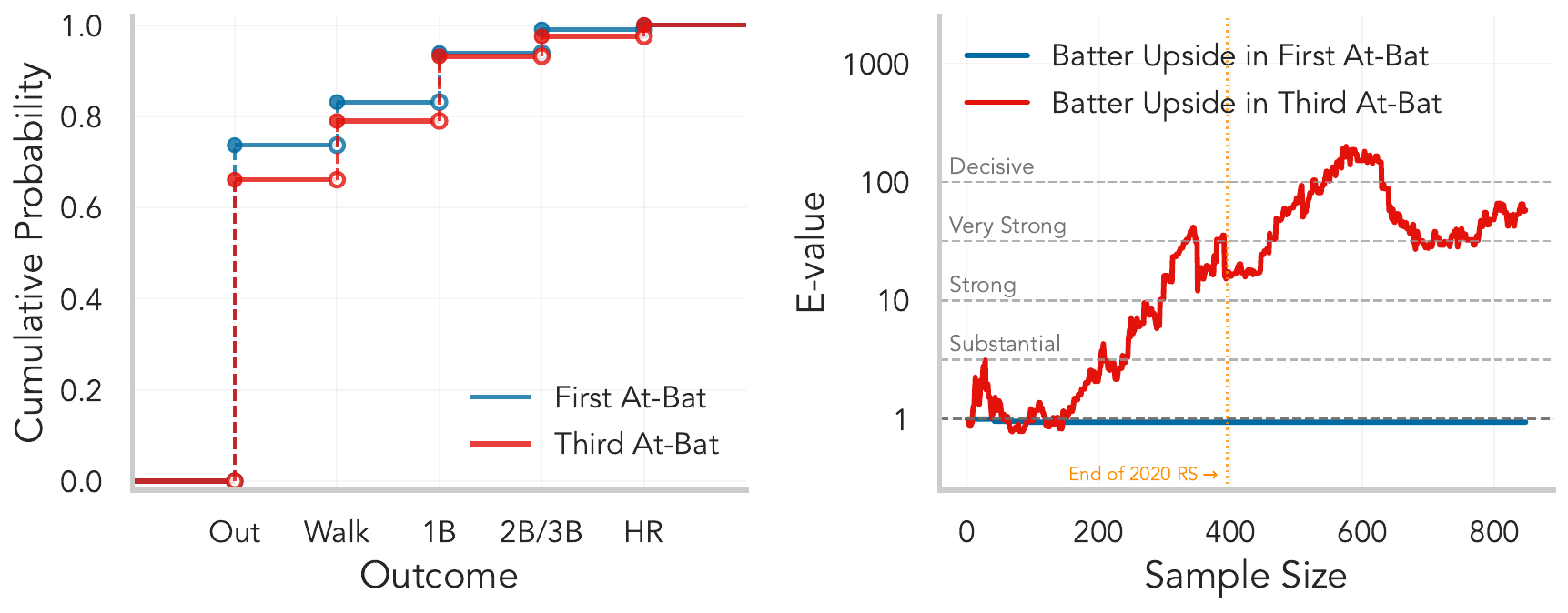}
    \caption{\textit{GRO e-processes for testing FSD provide strong evidence of the 3TTO penalty for MLB pitcher Blake Snell.} 
    The figure shows the empirical CDFs (left) and the e-processes for testing FSD in either direction (right). ``Batter upside in third at-bat'' refers to the relative advantage of opposing batters facing Snell for the third time than the first time.}
    \label{fig:baseball_tto_snell_1sd}
\end{figure}

In Figure~\ref{fig:baseball_tto_snell_1sd} (left), we plot the empirical CDFs of the observed outcomes in 1TTO and 3TTO, alongside the GRO e-processes for testing the FSD null in either direction. 
A first look at the CDFs shows that, over the ten seasons, Snell has allowed more hits and fewer outs in 3TTO than 1TTO. 
On the right, the evidence against the FSD null $\calH_0: Y \fsd X$ quantifies the \emph{batter's} upside in 3TTO relative to 1TTO, labeled in red; the evidence for batter's upside in 1TTO relative to 3TTO is labeled in blue. 
For reference, we add Jeffreys' rule-of-thumb thresholds for interpreting the strength of evidence, as recommended in \citep{ramdas2025ebook}.

After 150 observations, we see that only the e-process for testing $\calH_0: Y \fsd X$ (red) grows quickly, exceeding 10 around 300 observations and surpassing 100 after around 500 observations. 
This provides strong evidence that opposing batters have an upside in 3TTO against Snell relative to 1TTO, consistent with the 3TTO penalty hypothesis (conditional on staying in for the 3TTO). 
Had the manager been monitoring this e-process, by the end of the 2020 regular season (orange vertical line), the e-process would have been sitting at $17.41$, favoring his decision to pull Snell out of the game before his 3TTO in the World Series game. 
In practice, of course, the managerial decision takes into account several other factors: the quality of replacement; knowledge about the pitcher's physical and mental conditions on that day; belief about streakiness; and even public sentiment. 
The statistical evidence serves as one source of guidance for the eventual decision. 

In Supplement~\ref{app:real_data_exp}, we include additional experimental results, including a fine-grained analysis of batter upside against Snell and a comparison of Snell's outcomes with other top MLB pitchers. 
The analysis suggests that the presence of a 3TTO upside for batters can be heterogeneous, with Snell appearing rather anomalous among the very best pitchers.

\section{Affirming SD: Sequential tests for \emph{definite} upside}\label{sec:testing_non_SD_null}

So far, we have focused on testing the SD null hypothesis~\eqref{eq:testing_sd_null}, where rejection indicates that $Y$ has an upside over $X$. In some applications, however, it is important for the decision maker to have evidence that $Y$ is \emph{definitely} superior to $X$. 
A prominent example is online A/B testing, where a company compares a new design $B$ against an incumbent design $A$ and replaces $A$ only if there is sufficient evidence that $B$ dominates $A$. 
Anytime-valid methods are particularly attractive in such settings because they permit optional stopping while avoiding premature deployment decisions~\citep[e.g.,][]{johari2021always,lindon2022anytime}. 

The above example motivates the reversed testing problem involving the \emph{non-SD null}
\begin{equation}\label{eq:testing_non_sd}
\calH_0': X \notsd Y
\quad \text{vs.} \quad
\calH_1': X \sd Y,
\end{equation}
where $\sd$ denotes a SD relation. \emph{Affirming SD}, by rejecting the non-SD null in~\eqref{eq:testing_non_sd}, provides a strong form of guidance to the DM, as there is no point in keeping the incumbent over the challenger. 
In contrast, rejecting the SD null~\eqref{eq:testing_sd_null} provides a weaker form of guidance, suggesting $Y$ has its merits but does not necessarily disqualify $X$. 

Testing the non-SD null~\eqref{eq:testing_non_sd} is statistically more challenging than testing the SD null~\eqref{eq:testing_sd_null}, and many existing methods only address the latter. 
While the early work by~\citet{kaur1994testing} derived a consistent test for the non-2SD null, later work by~\citet{davidson2013testing} made clear that, for FSD, such tests can only effectively affirm a \emph{restricted} notion of SD on a pre-specified range that is a strict subset of the support. 
Indeed, with continuous or unbounded support, any two CDFs will eventually get close enough at the boundaries or in the tails, such that it is impossible to establish definite dominance across the entire support. 

We provide a formal exposition of the problem in Supplement~\ref{app:testing_non_SD_null}, focusing on FSD, and here we briefly summarize the main results. 
The first finding is that constructing a nontrivial e-variable for testing the non-SD null~\eqref{eq:testing_non_sd} is generally challenging:
\begin{corollary}
[No nontrivial e-variable exists for the non-SD null]
\label{cor:no_e_variable_for_non_sd_null}
There exists no e-variable $E$ for $\calH_0'$ which is nontrivial for all $\Qsymb \in (\calH_0')^c$ (that is, $\Esymb_{\Qsymb} \log(E) > 0$) simultaneously. 
\end{corollary}
Corollary~\ref{cor:no_e_variable_for_non_sd_null} directly follows from \citet[][Theorem 6.2]{zhang2024existence} and the fact that $\calH_0'$ is \emph{not} convex (contrasting with Lemma~\ref{lem:FSD-null-is-convex}). 
It suggests that we need to restrict ourselves to a smaller support $\tilde\calZ \subsetneq \calZ$ of interest, for which we can construct nontrivial e-variables. 
Given $\tilde\calZ \subset \calZ$, we define the restricted non-SD null as $\calH_0'(\tilde\calZ) = \incurly{\Psymb : \exists z_0\in \tilde\calZ, \Psymb_Y(z_0) > \Psymb_X(z_0) }$. 
Then, for $\varepsilon>0$, consider the restricted SD alternative with $\varepsilon$-separation in the CDFs:
\begin{align}\label{eq:alt_cdf_sep}
    \quad \calQ(\tilde\calZ, \epsilon) = \incurly{\Qsymb \in \frakB: \Qsymb_X(z) \geq \Qsymb_Y(z) + \varepsilon, \forall z \in \tilde\calZ}.
\end{align}
We can now construct sequential tests for the restricted non-SD null that are powerful against alternatives in $\calQ(\tilde\calZ, \epsilon)$, thereby affirming restricted SD (given minimal separation).
\begin{proposition}
    For any $\tilde\calZ \subset \calZ$ and $\epsilon > 0$ with $\calQ(\tilde\calZ, \epsilon)\neq\emptyset$, there exist sequential tests for the restricted non-SD null $\calH_0'(\tilde\calZ)$ that achieve asymptotic power one against 
    $\calQ(\tilde\calZ, \epsilon)$.
\end{proposition}
In Supplement~\ref{sec:testing_union_directly}, we construct such tests when $\calZ$ is finite and known by taking the minimum over the threshold-specific GRO e-processes.
In Supplement~\ref{sec:testing_nonsd_tvCS}, focusing instead on the case where $\calZ$ is continuous, we show that we can use time-uniform CDF bands~\citep[e.g.,][]{howard2022sequential,manole2023martingale,clerico2026uniform} to construct sequential tests that are powerful against $\calQ(\tilde{\calZ}, \epsilon)$. 
In either case, the tests are more difficult to reject than the analogous tests for the SD null~\eqref{eq:testing_sd_null}.

\section{Discussion}\label{sec:discussion}

This work presents a powerful and flexible framework for anytime-valid SD testing, including extensions to higher-order SD and non-SD null hypotheses. 
For SD testing, the core strategy is to express the composite null hypothesis as an intersection, enabling the predictable mixture strategy over thresholds and the derivation of a GRO e-variable (for FSD) that incorporates cross-dependence. 
The work motivates interesting future directions. 
In Section~\ref{sec:testing_non_SD_null}, we sketched an adaptation to the challenging case of affirming SD, finding that tests for affirming SD cannot achieve power against all alternatives without additional assumptions. 
\if\arxiv1%
{This difficulty has motivated earlier work on relaxed SD notions, such as almost SD \citep{baillo2024tests,leshno2002preferred} and restricted SD \citep{davidson2013testing}.} 
\else%
{This difficulty has motivated earlier work on relaxed SD notions, such as almost stochastic dominance \citep{baillo2024tests,leshno2002preferred} and restricted stochastic dominance \citep{davidson2013testing}.} \fi
Developing powerful, anytime-valid procedures for relaxed notions of SD remains an open problem.

In closing, we refer the reader to Supplement~\ref{sec:extensions}, where we discuss additional extensions to other popular integral and multivariate stochastic orders~\citep{shaked2007stochastic}, discuss validity under certain non-i.i.d.\ settings where the null hypothesis is defined using time-varying marginals~\citep{mineiro2023nonstationarity}, and explain how our methods can be extended to unpaired data.

%% file: contents/acknowledgements.tex
The authors thank Aaditya Ramdas, Shubhanshu Shekhar, Wouter Koolen, and Ruodu Wang for providing valuable early feedback.

%% file: contents/appendix.tex
\section{Omitted proofs}\label{sec:proofs}
\subsection{Proof of Lemma~\ref{lem:FSD-null-is-convex}}\label{app:proof_FSD_null_is_convex}
Consider $\Psymb,\Psymb'\in \calH_0$, and $\gamma \in (0,1)$ arbitrary. Then, for any $z\in \calZ$,
    \begin{align*}
        \left(\gamma \Psymb + (1-\gamma) \Psymb'\right)_X((-\infty,z])&= \Esymb_{\gamma \Psymb + (1-\gamma) \Psymb'} [\indicator{X \leq z}]\\ &= \gamma\Esymb_{ \Psymb } [\indicator{X \leq z}]+(1-\gamma)\Esymb_{  \Psymb'} [\indicator{X \leq z}] \\
        &\leq \gamma\Esymb_{ \Psymb } [\indicator{Y\leq z}]+(1-\gamma )\Esymb_{  \Psymb'} [\indicator{Y \leq z}] \\
        &= \left(\gamma \Psymb + (1-\gamma) \Psymb'\right)_Y((-\infty,z]),
    \end{align*}
    and it follows that $\gamma \Psymb + (1-\gamma) \Psymb' \in \calH_0.$

\subsection{Proof of Theorem \ref{thm:asymptotic_power_one}}\label{app:proof_asymptotic_power_one}

    Let $\Qsymb\in \calH_1$ be fixed. For any $t\geq 1$, by Jensen 
    \begin{align}
        \frac{\log(E_t)}{t}= \frac{1}{t}{\sum_{\ell=1}^t \log(S_\ell)}\geq
\frac{1}{t}\sum_{\ell=1}^t\int  \log(S_\ell(z))\diff \psi_\ell(z)=I_t\label{eq:lb_proof}
    \end{align}
    We show that, for sufficiently large $t\geq 1$, the quantity $I_t$ at \eqref{eq:lb_proof} is positive almost surely, which implies that $E_t\to \infty$ almost surely as $t\to \infty$. 
    
    By assumption, we may choose $\epsilon,\delta>0$ small and $N_0\geq 1$ such that 
    \begin{equation*}
        \psi_t(V)>\delta, \quad \textrm{for all }t\geq N_0
    \end{equation*}
    for
    \begin{equation*}
        V=\{z \mid F_X(z)-F_Y(z) > \epsilon\}= \{z \mid p(z)-q(z) > \epsilon\}\,.
    \end{equation*}
    For $z\in \calZ$, and $\lambda \in [-1,1]$, we define
  \begin{equation*}
      g_z(\lambda)=p(z) \log(1+\lambda) + q(z) \log(1-\lambda),  
  \end{equation*}
  for $ p(z) = \Qsymb(X\leq z < Y)$ and $ q(z) = \Qsymb(Y\leq z < X).$ 
  For any $z$ with $p(z)+ q(z)\neq 0$, $g_z$ is strictly concave and continuously differentiable with unique maximizer
\begin{equation*}
    \lambda^\star(z) = \frac{p(z)-q(z)}{p(z)+q(z)}\in [-1,1]\,,
\end{equation*}
and maximum value
\begin{align}
     g_z(\lambda^\star(z)) &= p(z) \log\left(\frac{2 p(z)}{p(z)+q(z)}\right)+q(z)\log\left(\frac{2 q(z)}{p(z)+q(z)}\right)\nonumber\\ &=\big(p(z)+q(z)\big)\cdot\textrm{kl}\left(\frac{p(z)}{p(z)+q(z)},\frac{1}{2}\right)\,,\label{eq_kl}
\end{align}
for kl denoting the KL divergence for Bernoulli random variables.

\textbf{Claim 1:} There exists $ \epsilon' >0$ and $N_1\geq 1$ such that 
\begin{equation}\label{eq:claim1}
\Esymb_{\Qsymb}[\log (S_t(z)) \mid \calF_{t-1}]=g_z(\hat{\lambda}_t^\textsf{GRO}(z))>\epsilon', \quad \forall t\geq N_1, \forall z \in V\,.
 \end{equation}
\emph{Proof:} By Glivenko–Cantelli, the empirical CDFs converge almost surely, uniformly for all $z\in V$. By the continuous mapping theorem, this implies that $\hat{\lambda}_t^\textsf{GRO}(z) \to \lambda^\star(z)_+ \land(1-c)=\lambda^\star(z) \land(1-c)$ almost surely uniformly for all $z\in V$. 
Minimizing the KL distance at \eqref{eq_kl} over all such $z$ yields
\begin{equation*}
    g_z(\lambda^\star(z)) > \epsilon \cdot \textrm{kl}\left(\frac{1+\epsilon}{2},\frac{1}{2}\right):= \epsilon_1 >0, \quad \forall z \in V\,.
\end{equation*}
Next, we show that also $g_z(\min\{\lambda^\star(z), 1-c\})>0$ for all $z \in V$. For this, we show that there exists $\epsilon_2>0$ such that, for $z \in A= V \cap \{z :\lambda^\star(z)>1-c\}$, we have
\begin{equation}\label{eq:proof1}
   g_z(\min\{\lambda^\star(z), 1-c\})= g_z(1-c) = p(z) \log(2-c) + q(z) \log(c) > \epsilon_2\,.
\end{equation}
Rewrite $A = \{z : q(z) < \min\{\gamma p(z),p(z)-\epsilon\}\}$ for $\gamma= c/(2-c)$. For any $z$, the left-hand side of \eqref{eq:proof1} is decreasing in $q(z)$. By equalizing $\gamma p(z)=p(z)-\epsilon$ and thus making $q(z)$ as large as possible, we obtain the positive lower bound 
\begin{equation*}
   g_z(\min\{\lambda^\star(z), 1-c\})> \frac{\varepsilon}{2(1-c)}
\Bigl((2-c)\log(2-c) + c\log(c)\Bigr):=\epsilon_2, \quad \forall z \in A\,.
\end{equation*}
Set $\epsilon'= \min\{\epsilon_1,\epsilon_2\}>0.$ Then, we have $g_z(\min\{\lambda^\star(z), 1-c\})>\epsilon'$ for all $z \in V$. Since $g_z$ is continuous, it follows that there exists $N_1\geq 0$ such that \eqref{eq:claim1} holds true.

 \textbf{Claim 2:} There exists $N_2\geq 1$ such that 
    \begin{equation}\label{eq:claim2}
\Esymb_{\Qsymb}[\log (S_t(z)) \mid \calF_{t-1}]=g_z(\hat{\lambda}_t^\textsf{GRO}(z))>-\frac{\epsilon'}{2}\cdot\frac{\delta}{\psi_t(V^c)}, \quad \forall t\geq N_2, \forall z \in V^c\,.
    \end{equation}
\emph{Proof:} For $0<\epsilon_3< \epsilon' \delta/(2|\log c|)$, let $B_1=\{z\in V^c : p(z)+q(z)=0\}$ and $B_2=\{z \in V^c: p(z)+q(z)>\epsilon_3\}$. For $z\in B_1$, we have $\hat{\lambda}_t^\textsf{GRO}(z)=0$ for all $t$ and thus $g_z(\hat{\lambda}_t^\textsf{GRO}(z))=0$ for all $t$. For $z$ with $0<p(z)+q(z) < \epsilon_3$, that is, $z\in V^c \setminus (B_1 \cup B_2)$, we have, for any $\lambda\in [0,1-c]$, 
\begin{align*}
    g_z(\lambda) \geq q(z)\log (c)\geq  (p(z)+q(z))\log (c)\geq \epsilon_3 \log (c)\geq -\frac{\epsilon' \delta}{2}\geq-\frac{\epsilon'}{2}\cdot\frac{\delta}{\psi_t(V^c)}\,,
\end{align*}
and thus in particular $g_z(\hat{\lambda}_t^\textsf{GRO}(z))\geq -{\epsilon' \delta}/{(2 \psi_t(V^c))}$ for all $z\in V^c\setminus (B_1 \cup B_2)$. Finally, for $z \in B_2$, by Glivenko–Cantelli and the continuous mapping theorem, we have 
$\hat{\lambda}_t^\textsf{GRO}(z) \to \lambda^\star(z)_+ \land(1-c)$ almost surely. Since $g_z(\lambda^\star(z)_+ \land(1-c))\geq 0$, for all $z$, it follows by continuity of $g_z$ that we can find $N_2\geq 1$ large enough such that \eqref{eq:claim2} is satisfied. 

Let $N= \max\{N_0,N_1,N_2\}$, and let 
\begin{equation*}
    Z_t = \int \log (S_t(z)) \diff \psi_t ( z)\,,\quad t\geq N_0.
\end{equation*}
Then, by Fubini, for any $t\geq N$,
\begin{align}
    \Esymb_{\Qsymb}(Z_t \mid \calF_{t-1}) &= \int  \Esymb_{\Qsymb}(\log (S_t(z)) \mid \calF_{t-1}) \diff \psi_t ( z)\nonumber \\
    &= \int_V  \Esymb_{\Qsymb}(\log (S_t(z)) \mid \calF_{t-1}) \diff \psi_t ( z) +\int_{V^c}  \Esymb_{\Qsymb}(\log (S_t(z)) \mid \calF_{t-1}) \diff \psi_t ( z) \nonumber\\
    &\geq \epsilon'\psi_t (V) +\left(-\frac{\epsilon'}{2}\cdot\frac{\delta}{\psi_t(V^c)}\right)  \psi_t (V^c)\nonumber\\
    &\geq\frac{\epsilon' \delta}{2}>0\,.  \label{eq:proof2}
\end{align}
Moreover, for any $t\geq N$, by Jensen
\begin{equation*}
  Z_t^2 \leq \int \log (S_t(z))^2 \diff  \psi_t ( z) < \infty\,,
\end{equation*}
as $\log (S_t(z))^2$ is bounded for all $z\in \calZ$ and all $t\geq N$. Hence, by the martingale SLLN
 \begin{equation*}
        \lim_{t\to \infty}\frac{1}{t}\sum_{\ell=N}^t  Z_\ell -  \Esymb_{\Qsymb}(Z_\ell \mid \calF_{\ell-1}) =0\,,
    \end{equation*}
and by \eqref{eq:proof2}, we may conclude that
\begin{align*}
   \liminf_{t\to \infty} I_t = \liminf_{t\to \infty} \frac{1}{t} \sum_{\ell=N}^{t} \Esymb_{\Qsymb}(Z_\ell \mid \calF_{\ell-1}) +\liminf_{t\to \infty} \frac{1}{t} \sum_{\ell=N}^{t} \big(Z_\ell-\Esymb_{\Qsymb}(Z_\ell \mid \calF_{\ell-1})\big) >\frac{\epsilon' \delta}{2}>0 \,.
\end{align*}

\section{Testing higher-order SD}\label{sec:testing_higher_order_SD}

Testing for first-order SD is a useful starting point. 
While conceptually simple, however, testing FSD may be too restrictive in some decision-making contexts, as rejection simply indicates there is some DM with an increasing utility function who prefers $X$ to $Y$. 
This motivates us to develop sequential tests for SD under known subclasses of increasing utility functions, making the SD null less restrictive.

\subsection{Higher-order SD characterizations}\label{sec:kSD}

We first review the standard characterizations of higher-order SD~\citep[e.g.,][]{shaked2007stochastic}.
Mathematically, the $k$-th order SD relation is defined in terms of iterated integrals of CDFs. 
\begin{definition}[$k$-th order SD]\label{def:kSD}
    Let $X$ and $Y$ be two rv's with CDFs $F_X$ and $F_Y$, respectively.
    For $k\geq 2$, we say that $X$ \emph{stochastically dominates $Y$ in k-th order ($k$-SD)}, denoted by $Y \ksd X$, if
    \begin{equation}\label{eq:ksd_definition}
        F_X^{[k]}(z) \leq F_Y^{[k]}(z), \quad \forall z \in \R,
    \end{equation}
    where $F^{[k]}$ denotes the $k$-th iterated integral of the CDF $F$, defined recursively: $F^{[1]} \equiv F$, and for $k\geq 2$,
    \begin{equation*}
        F^{[k]}(z) = \int_{-\infty}^z F^{[k-1]}(u)\diff u, \quad z\in \R.  
    \end{equation*}
    The relation is well-defined when the iterated integrals exist and are finite for all $z\in \R$. 
\end{definition}

The key to our e-process construction is to leverage the known characterization of these iterated integrals of CDFs in terms of lower partial moments~\citep[e.g.,][]{shaked2007stochastic}: 
\begin{lemma}[Generator characterization of $k$-SD]\label{lem:ksd_generator}
    Let $X$ and $Y$ be two rv's with CDFs $F_X$ and $F_Y$, respectively. For each $k\geq 2$, the following statements are equivalent:
    \begin{enumerate}[(a)]
        \item $Y \ksd X$ in the sense of Definition~\ref{def:kSD}, that is, $F_X^{[k]}(z) \leq F_Y^{[k]}(z)$ for all $z\in \R$;
        \item $\Esymb[g_z^{[k]}(Y)] \leq \Esymb[g_z^{[k]}(X)] $ for all $z \in \R$, where $g_z^{[k]}(x) = -(z - x)_+^{k-1}$ for $x\in \calZ$.\label{item_b:Lemma6}
    \end{enumerate}
\end{lemma}
As before, each characterization is well-defined assuming that the associated expectations are finite for each $z\in \R$. 
We refer to the functions $g_z^{[k]}$ as the \emph{(unnormalized) generator functions} of $k$-SD, as they are sufficient for generating the SD relation in the $k$-th order~\citep{marshall1991multivariate}. 
As with FSD, we now have a simple characterization of $k$-SD in terms of a class of functions indexed by the outcome space $\R$. 

\begin{remark}[On the range of the testing threshold $z$] Comparing with Definition~\ref{def:FSD} for FSD and the equivalent characterization in Lemma~\ref{lem:FSD_equivalences}, 
one might be tempted to restrict the thresholds $z$ for higher-order SD only to the common support $\calZ$ rather than requiring the conditions to hold for all $z\in \R$. For second-order stochastic dominance this would indeed be sufficient.
However, for $k\geq 3$, pointwise ordering of the iterated CDF integrals does not, in general, imply the corresponding ordering of lower-order moments required for $k$-th order SD~\citep{whitmore1970third,fishburn1976continua}.
For a concrete counterexample, consider, e.g., the case where $P_X=(1/2)(\delta_{1/2}+\delta_{1})$ and $P_Y=(1/5)\delta_{0}+(4/5)\delta_{1}$, for which one can verify that $F_X^{[3]}(z)\leq F_Y^{[3]}(z)$, for all $z\in \calZ=[0,1]$, even though $\Esymb(X)=3/4<4/5=\Esymb(Y)$. 
That is, if $\sup \calZ=b<\infty$, the iterated CDF integrals also carry moment information to the right of $b$, and one would need to impose the additional lower-moment conditions $F_X^{[j]}(b) \leq F_Y^{[j]}(b)$, for all $j \leq k-1$, alongside the iterated integral condition \eqref{eq:ksd_definition}, when restricting attention simply to $z\in \calZ$.
\end{remark}

\begin{remark}[The expected utility characterization for $k$-SD]\label{remark:eu_ksd}
    There is also a standard expected utility characterization of $k$-th order SD~\citep{fishburn1970utility}: 
    for $k\geq 2$, $Y \ksd X$ if and only if $\ex{u(Y)} \leq \ex{u(X)}$ for all $u \in \calU^{[k]}$ whenever the expectations exist, where $\calU^{[k]}$ consists of all $k$-times-differentiable functions (denoted by $C^k(\R)$) whose first $k$ derivatives alternate in sign:
    \begin{equation*}\label{eq:utility_class_kSD}
        \calU^{[k]} = \{u\in C^k(\R) \mid (-1)^{j-1} u^{(j)} (z) \geq 0, \textrm{ for all }z \in \R, j=1, \ldots, k\}.
    \end{equation*}
    For example, when $k=2$ (second-order SD or SSD), the utility class $\calU^{[2]}$ consists of all increasing and concave functions, which represents risk aversion in
the expected-utility framework
commonly used in economics and decision theory~\citep{pratt1964risk,rothschild1970risk}. 
    In particular, given two bets $X$ and $Y$ with the same expected value, a DM with a concave utility function would prefer the bet with less variability.\footnote{For example, given a 50-50 bet that pays either $x_1=-\$100$ or $x_2=\$100$, a DM with a linear utility function would be indifferent between
taking the bet and not taking it.
    Yet, a DM with a concave utility function would prefer to avoid the bet, reflecting risk aversion. Analogously, a DM with a convex utility function would prefer to take the bet, reflecting risk-seeking behavior.}
    Similarly, the $k=3$ case corresponds to a ``prudent'' DM whose marginal utility is further convex, such that the DM prefers to avoid negative downside risk (i.e., negative skew); the $k=4$ case corresponds to a ``temperate'' DM who prefers to avoid heavy tails on either side~\citep{whitmore1970third,fishburn1976continua,kimball1990precautionary}. 
    In this view, rejecting the $k$-SD null $\calH_0: Y \ksd X$ indicates there exists at least one such DM for whom $Y$ has an upside over $X$. 
\end{remark}

\subsection{Testing for higher-order SD with e-variables}\label{sec:higher_order_evariable}

We now consider sequentially testing the $k$-SD null hypothesis. By Lemma~\ref{lem:ksd_generator}, we can express the $k$-SD null as another intersection null over $\calZ$:
\begin{equation}\label{eq:ksd_null}
    \calH_0^{[k]} = \{ \Psymb \in \frakB^{[k]} \mid Y \ksd X \text{ under } \Psymb\} = \bigcap_{z \in \R} \calH_0^{[k]}(z), 
\end{equation}
where
\begin{equation*}\label{eq:ksd_null_z}
    \calH_0^{[k]}(z) = \incurly{\Psymb \in \frakB^{[k]} \mid \Ex{\Psymb}{g_z^{[k]}(Y)} \leq \Ex{\Psymb}{g_z^{[k]}(X)}},
\end{equation*}
and $\frakB^{[k]}$ denotes the family of all probability measures on $(\Omega,\calF)$ such that all $(k-1)$-th order lower partial moments of $X$ and $Y$ exist. 
Rejecting the 2-SD null, for example, indicates that $Y$ has an upside over $X$ in the second order (in the expected utility view, $Y$ is preferred to $X$ by at least one risk-averse DM).

Testing $k$-th order SD is only sensible when the associated lower partial moments exist, such as when the support~$\calZ$ is bounded from below (see Assumption~\ref{assn:finite_lower_bound} later).
In such cases, Definition~\ref{def:kSD} implies that we have an \emph{increasing} sequence of nulls,
\begin{equation}\label{eq:ksd_relations}
    \tilde{\calH}_0^{[1]} \subset \tilde{\calH}_0^{[2]} \subset \dotsb \subset \tilde{\calH}_0^{[k-1]} \subset \tilde{\calH}_0^{[k]} \subset \dotsb \subset {\calH}_0^{[\infty]} \subset \calH_0^{[\mu]},
\end{equation}
where $\tilde{\calH}_0^{[k]} = \calH_0^{[k]} \cap \frakB^{[\infty]}$ is the $k$-SD null restricted to distributions with finite moments of all orders, $\calH_0^{[\infty]}$ is the infinite-order SD null (the limit of the $k$-SD nulls; see Section~\ref{sec:extensions}), and $\calH_0^{[\mu]} = \incurly{\Psymb \in \frakB: \Ex{\Psymb}{Y} \leq \Ex{\Psymb}{X}}$ is the mean dominance null. 

Using Lemma~\ref{lem:ksd_generator}, we can construct e-variables for any $k$-SD null $\calH_0^{[k]}(z)$ for each $z\in \calZ$ analogous to the FSD case. 
Unlike in the FSD case, however, we will now require a lower bound on the random variables. 
Recalling that $\calZ$ denotes the common support of $X$ and $Y$, 
\begin{assumption}[Finite lower bound]\label{assn:finite_lower_bound}
    $\calZ \subseteq [a, \infty)$ for some $a \in \R$. 
\end{assumption}

Assumption~\ref{assn:finite_lower_bound} implies in particular that the iterated CDF integrals at \eqref{eq:ksd_definition} and the truncated moments
in Lemma~\ref{lem:ksd_generator}\ref{item_b:Lemma6} exist for all $z\in \R$ and $k\geq 1$. 
\begin{theorem}
\label{thm:Test-supermartingale-for-the-kSD_null} 
Suppose that Assumption~\ref{assn:finite_lower_bound} holds. For $k \geq 2$, define the $k$-th order normalized utility generator for each $z \in \calZ$:
\begin{equation*}\label{eq:def_kSD_normalized_generator}
    u_z^{[k]}(x) = \frac{g_z^{[k]}(x)}{(z-a)^{k-1}} = -\insquare{\frac{(z-x)_+}{z-a}}^{k-1}, \quad x \in \calZ,
\end{equation*}
if $z \in (a, \infty)$, and $u_a^{[k]}(x) = 0, \, x \in \calZ$. 
Then, for each $k\geq 2$, the following statements hold: 
\begin{enumerate}[(a)]
    \item For each $\lambda \in [0,1]$,
    \begin{equation*}\label{eq:evalue_ksd}
        S^{[k]}(\lambda, z) = 1+\lambda [{u}_z^{[k]}(Y)-{u}_z^{[k]}(X)]
    \end{equation*}
    is an e-variable for $\calH_0^{[k]}$. 
    \item For any predictable sequence of mixtures $(\psi_t)_{t\in \N} \subseteq \Pi(\calZ)$, and predictable bets $(\lambda_t(z))_{t\in \N} \subseteq [0,1]$ for each $z\in \calZ$, 
    \begin{equation*}
        S_t^{[k]} = 1 +\int_a^\infty \lambda_t(z) [{u}_z^{[k]}(Y_t)-{u}_z^{[k]}(X_t)] \psi_t(\diff z), \quad t\in \N, 
    \end{equation*}
    forms sequential e-variables for $\calH_0^{[k]}$, and its running product $E_t^{[k]} = \prod_{\ell=1}^t S_\ell^{[k]}$
    forms a test SM, and thus an e-process, for $\calH_0^{[k]}$. 
\end{enumerate}
\end{theorem}
\begin{proof}
    For any $k \geq 2$ and any rv $X$ with CDF $F$ and finite $(k-1)$-th lower partial moments, by Fubini
\begin{equation*}\label{eq:equiv_characterization_HO_FSD}
    F^{[k]}(z) = \frac{1}{(k-1)!} \ex{(z-X)_+^{k-1}}, \quad z\in \calZ\,.
\end{equation*}
The null hypothesis $\calH_0^{[k]}: F_X^{[k]}(z) \leq F_Y^{[k]}(z)$ for all $z \in \calZ$, is thus equivalent to
\begin{equation*}
    \ex{(z-X)_+^{k-1}} \leq \ex{(z-Y)_+^{k-1}}, \quad \forall z\in \calZ.
\end{equation*}
Using the lower bound of $a$ (Assumption~\ref{assn:finite_lower_bound}), we can rewrite this condition as
\begin{equation*}
\ex{-\inparen{\frac{(z-Y)_+}{z-a}}^{k-1}} \leq \ex{-\inparen{\frac{(z-X)_+}{z-a}}^{k-1}}, \quad \forall z\in \calZ.
\end{equation*}
Because the differences $u_z^{[k]}(x) - u_z^{[k]}(y)$ take values in $[-1, 1]$, each $S^{[k]}(\lambda, z)$ is an e-variable for $\calH_0^{[k]}$ for any fixed choice of $\lambda \in [0, 1]$. 
The second part follows exactly analogously to the case of Theorem~\ref{thm:Test-supermartingale-for-the-FSD_null}. 
\end{proof}

\begin{remark}
While the assumption of a finite lower bound is common in classical tests for SD~\citep[e.g.,][]{barrett2003consistent,linton2005consistent}, it may be restrictive when the data are unbounded from both sides.  
In such cases, it is possible to instead impose a tail condition on the data distributions, in particular that they are sub-exponential, and derive a different building-block e-variable for testing 2-/3-SD based on the theory of \emph{sub-$\psi$ processes}~\citep{howard2020chernoff,howard2021time}. 
Sub-exponential distributions have tails that decay at least as fast as the exponential distribution, and thus include any bounded and sub-Gaussian distributions as special cases. 
\end{remark}

\subsection{Portfolio-based bets for higher-order SD testing}\label{sec:ksd_portfolio}

As with FSD, we begin by choosing an optimal betting strategy for the ``atomic'' hypothesis $\calH_0(z)$ of the intersection null $\calH_0 = \bigcap_{z \in \calZ} \calH_0(z)$. 
Recall that, for FSD, we were able to derive a GRO betting strategy $\lambda^\star(z)$ for each threshold $z \in \calZ$ (Proposition~\ref{ppn:GRO_betting_parameter_FSD}) and obtain a log-optimal e-process for each atomic null $\calH_0(z)$. 
This derivation leveraged the relatively simple form of the building-block e-variables in the FSD case, something we do not have for higher-order SD as we do not have ternary building-block e-variables. 

Here, instead of focusing on a particular SD relation, we propose a general SD betting strategy with respect to arbitrary utility generators. 
We can do so for the higher-order SD relations, as well as for other SD relations such as the increasing convex order or the infinite-order SD discussed in Section~\ref{sec:extensions}. 
We leverage recent advances that derive portfolio-based betting strategies that \emph{asymptotically} achieve the log-optimal growth rate. 

Suppose we sequentially test $\calH_0 = \bigcap_{z \in \calZ} \calH_0(z)$, for $\calH_0(z) = \{\Psymb \mid \Ex{\Psymb}{u_z(Y)} \leq \Ex{\Psymb}{u_z(X)}\}$ and utility generator $u_z$ at threshold $z \in \calZ$.  
Given an adapted sequence of e-variables $(S_t(z))_{t\in \N}$ for each null hypothesis $\calH_0(z)$, $z \in \calZ$, consider the wealth process
\begin{equation}\label{eq:wealth_fbet}
    W_t(z) = \prod_{\ell=1}^t \insquare{(1 - \fbet_{\ell}(z)) + \fbet_{\ell}(z) S_{\ell}(z)}, \quad t\in \N,
\end{equation}
where $\fbet_{\ell}(z) \in [0,1]$ is the skeptic's predictable bet in the $\ell$th round on the $z$th game. 
We can recover our earlier e-process construction by setting $S_t(z) = 1 + [u_z(Y_t) - u_z(X_t)]$. 

Prior work establishes that we can construct betting strategies $(\lambda_t(z))_{t\in \N}$ that yield an \emph{asymptotically log-optimal}~\citep{waudbysmith2025universal,wang2025ebacktesting} wealth process under any alternative $\Qsymb$ for which $(S_t(z))_{t\in \N}$ are not e-variables. 
Essentially, this means that the process achieves the optimal growth rate in large samples under i.i.d.~data.
When compared to using the GRO betting strategy that we derived explicitly for the FSD case, we are sacrificing finite-sample optimality, at each time $t$, in exchange for the flexibility of the strategy across different utility classes without having to explicitly derive the optimal bet.

\citet{waudbysmith2025universal} find that a canonical betting strategy that achieves asymptotic log-optimality is the \emph{universal portfolio (UP)} strategy~\citep{cover1991universal}. 
In essence, the UP strategy begins with an (uninformative) prior $\pi$ over the fractional bet parameter $\lambda \in [0, 1]$, such as the $\mathsf{Beta}(1/2, 1/2)$ distribution, and updates the posterior over $\lambda$ at each time $t$ by treating the wealth process $W_t(z; \lambda)$ itself as a likelihood function. 
More formally, let $W_t(z; \lambda)$ denote the wealth process~\eqref{eq:wealth_fbet} at time $t$ for index $z \in \calZ$ where the fractional bet is fixed to $\lambda$ across time, and compute the posterior mean for $\lambda_t(z)$: for $t \geq 1$,
\begin{equation*}
    \fbet_t^\up(z) = \frac{\int_0^1 \fbet W_{t-1}(z; \fbet) \, \diff\pi(\fbet)}{\int_0^1 W_{t-1}(z; \fbet) \, \diff \pi(\fbet)}.
\end{equation*}
In practice, we use the clipped UP strategy 
\begin{equation}\label{eq:universal_portfolio}
    \fbet_t^\up(z) = \frac{\int_0^{1-c} \fbet W_{t-1}(z; \fbet) \, \diff\pi(\fbet)}{\int_0^{1-c} W_{t-1}(z; \fbet) \, \diff\pi(\fbet)}.
\end{equation}
for some small $c>0$ in order to prevent e-variables from being zero, and approximate the integrals by discretizing the space of fractional bets and keeping track of unnormalized posterior weights for each candidate bet. 
We remark that it is possible to use other portfolio-based betting strategies that also achieve asymptotic log-optimality~\citep[e.g.,][]{cover1996universal,orabona2023tight}. 

The UP strategy described here takes the place of the GRO strategy for FSD testing in constructing the building-block e-variables (Lemma~\ref{lem:building-block-e-variables}) for each threshold $z \in \calZ$.
To aggregate wealth across thresholds, we recommend using the adaptive approach with predictable weights, described in Section~\ref{sec:agg_over_z}.  
In the next section, we show that the asymptotic power-one guarantee of Theorem \ref{thm:asymptotic_power_one} for FSD translates for higher-order SD when replacing the empirical GRO plug-in with the UP strategy.

\begin{remark}\label{rem:sr_comparison}
    In Section~\ref{sec:related_work_testing_SD}, we mentioned that the preliminary draft of \citet[SR;][Appendix F.1 of arxiv.v1]{shekhar2023nonparametric} sketched a related approach for $k$-SD testing. 
    In this variant, at each time $t$, the threshold $z_t^*$ at which the mean \emph{unnormalized} generator difference is the largest is first chosen; then, on that threshold, it is suggested that a portfolio bet is made. 
    This can be viewed (conceptually) as a special case of our $k$-SD framework, such that the predictable weights are given entirely to the single threshold $z_t^*$. 
    Yet, there are also two subtle differences. 
    First, the e-variable in this paper is based on \emph{normalized} utility generator differences, bounded in $[-1,1]$ for each threshold $z$, ensuring \citet{cover1991universal}'s standard UP bets are well-defined under a finite lower bound. 
    Second, the UP bets in our framework for each $z_t^*$, even when using a greedy mixture, are calculated based on the prior wealth \emph{on that particular threshold}, as opposed to SR's variant that maintains a single posterior on bets across time even as the selected threshold moves around. 
    To implement SR's approach faithfully (as it was never implemented and removed in their subsequent versions), we would need to find a different prior on $\lambda$ suitable for unbounded and unnormalized generators; instead, we adapt their greedy strategy within our framework (``AdaUP-Greedy''), and find it to be an effective adaptation for $k$-SD testing (see Section~\ref{app:sim_kuniform_ksd}). 
\end{remark}

\subsection{Asymptotic power-one guarantee for higher-order SD}\label{app:power-one_for_higher_order}
For $k\geq 2$, following Theorem~\ref{thm:Test-supermartingale-for-the-kSD_null}, at $t\geq 1$, for some predictable mixture distribution $\psi_t \in \Pi(\calZ)$, we consider the $\calZ$-mixture sequential e-variable
\begin{equation}\label{eq:mixture_e-variable_higher_order}
    S_t = \int S_t^{[k]}(  \fbet_t^\up(z),z)\diff\psi_t( z), 
\end{equation}
as our measure of evidence against the global k-SD null $\calH_0^{[k]}$~\eqref{eq:ksd_null} with the UP betting parameter $\fbet_t^\up(z)$~\eqref{eq:universal_portfolio}, where we impose the following mild condition on the prior:

\begin{assumption}[Prior mass]\label{ass:prior}
For $c\in(0,1/2)$, the prior $\pi$ on $[0,1-c]$ admits a Lebesgue density
bounded above and below by constants $0<\underline\pi\le\overline\pi<\infty$.
\end{assumption}

Assumption~\ref{ass:prior} is mild and is satisfied, in particular, by the uniform prior on $[0,1-c]$.

The next theorem shows that the asymptotic power-one guarantee of Theorem~\ref{thm:asymptotic_power_one} for FSD extends to higher-orders when replacing the empirical GRO plug-in with the UP strategy.  

\begin{theorem}[UP e-process with predictable $\calZ$-mixture is powerful against the higher-order SD null]\label{thm:asymptotic_power_one_for_higher_order} Suppose that Assumptions~\ref{assn:finite_lower_bound} and~\ref{ass:prior} hold. 
Let $k\geq 2$ and $\Qsymb\in (\calH_0^{[k]})^c$. For $(\psi_t)_{t\in \N}\subseteq \Pi(\calZ)$ predictable, assume that there exist $\epsilon,\delta>0$ small and $N_0 \in \N$ such that, 
\begin{equation*}
\psi_t\left(\left\{z:\Ex{\Qsymb}{{u}_z^{[k]}(Y)-{u}_z^{[k]}(X)}> \epsilon\right\}\right)>\delta, \quad \textrm{for all }t\geq N_0\,.
\end{equation*} 
Then, the e-process $E_t=\prod_{\ell=1}^t S_\ell$, for $t\in \N$ and $S_\ell$ defined at \eqref{eq:mixture_e-variable_higher_order}, is powerful against $\calH_1$ and the corresponding sequential test has asymptotic power one. 
\end{theorem}

\subsection{Proof of Theorem~\ref{thm:asymptotic_power_one_for_higher_order}}\label{seq:proofs_higher_order_SD}
The argument follows the structure of the proof of
Theorem~\ref{thm:asymptotic_power_one}, with two additional ingredients that we spell out
in the form of two lemmas (Lemmas~\ref{lem:uniformGC}
and~\ref{lem:posterior_concentration} below). All almost sure statements in the proof are with respect to the measure $\Qsymb$. 

Fix $k\geq 2$ and $c\in(0,1/2)$. For each
$z\in\calZ$, let
\begin{equation*}
    D(z) = u_z^{[k]}(Y) - u_z^{[k]}(X)\in[-1,1]\,,
\end{equation*}
and 
\begin{equation*}
    g_z(\lambda)=\Ex{\Qsymb}{\log(1 + \lambda D(z))}\,,
    \qquad \lambda\in[0, 1-c].
\end{equation*}
The interval $[0,1-c]$ is chosen so that $1 + \lambda D(z)\in[c,\,2-c]$. Hence, $g_z$ is finite and thus infinitely differentiable by dominated convergence. Whenever $D(z)$ is not identically zero, the map $\lambda\mapsto g_z(\lambda)$ is strictly concave on $[0,1-c]$. In this case, let
\begin{equation*}
    \fbet^\star_c(z) \;=\; \argmax_{\lambda\in[0,1-c]} g_z(\lambda)\in[0,1-c]
\end{equation*}
denote the constrained maximizer, and let $\lambda^\star(z)=
\argmax_{\lambda\in[0,1]}g_z(\lambda)\in[0,1]$ denote the unconstrained
maximizer. By concavity, $\fbet^\star_c(z)=\min\{\lambda^\star(z),\,1-c\}$. Moreover, let
\begin{equation*}
    W_t(z;\lambda) \;=\; \prod_{\ell=1}^{t}\bigl(1+\lambda(u_z^{[k]}(Y_\ell)-u_z^{[k]}(X_\ell))\bigr),
    \qquad \hat g_t(z,\lambda) \;=\; \frac{1}{t}\log W_t(z;\lambda),
\end{equation*}
and denote the clipped UP posterior at time $t$ by
$\pi_{t-1}(z;\mathrm{d}\lambda)\propto W_{t-1}(z;\lambda)\,\pi(\mathrm{d}\lambda)$,
restricted to $[0,1-c]$, so that
$\fbet_t^\up(z)=\int_0^{1-c}\lambda\,\pi_{t-1}(z;\mathrm{d}\lambda)$.

\begin{lemma}[Uniform Glivenko--Cantelli]\label{lem:uniformGC}
Under Assumptions~\ref{assn:finite_lower_bound} and~\ref{ass:prior}, almost surely,
\begin{equation*}
    \sup_{z\in\calZ,\,\lambda\in[0,1-c]}
    \left|\hat g_t(z,\lambda)-g_z(\lambda)\right|
    \to 0,\quad \textrm{for }t\to \infty\,.
\end{equation*}
\end{lemma}

\begin{proof}
Consider the function class
\begin{equation*}
    \calG_c \;=\; \bigl\{(x,y)\mapsto \log\bigl(1+\lambda[u_z^{[k]}(y)-u_z^{[k]}(x)]\bigr):\,z\in\calZ,\,\lambda\in[0,1-c]\bigr\}.
\end{equation*}
The truncated-power family $\{x\mapsto(z-x)_+^{k-1}:z\in\calZ\}$ is
VC-subgraph, by a polynomial-class argument applied to the affine
family $\{z-x:z\in\calZ\}$ together with the preservation of
VC-subgraph under pointwise maxima with $0$ and monotone composition
\citep[Lemmas 2.6.15 and 2.6.18]{vanderVaart_empirical_processes}. The same properties together with a Lipschitz composition with $\log(1+\cdot)$ on the bounded range $[c,2-c]$ imply that $\calG_c$ is itself VC-subgraph with a uniformly bounded envelope $\max(|\log c|,|\log(2-c)|)$. By the
uniform LLN for VC-subgraph classes with integrable envelope
\citep[Theorem 2.4.1]{vanderVaart_empirical_processes}, $\calG_c$ is Glivenko--Cantelli, which is the claim.
\end{proof}

\begin{lemma}[Uniform posterior concentration]\label{lem:posterior_concentration}
Let Assumptions~\ref{assn:finite_lower_bound} and~\ref{ass:prior} hold. Let $A\subseteq\calZ$ be any
subset on which there exists $c^\star>0$ such that
\begin{equation}\label{eq:uniform_curvature}
    \inf_{z\in A,\;\lambda\in[0,1-c]}\bigl(-\partial_\lambda^2 g_z(\lambda)\bigr) \;\ge\; c^\star.
\end{equation}
Then, almost surely
\begin{equation*}
    \sup_{z\in A}\left|\fbet_t^\up(z)-\fbet^\star_c(z)\right|
   \to0,\quad \text{for }t\to \infty\,.
\end{equation*}
\end{lemma}

\begin{proof}
Under \eqref{eq:uniform_curvature}, we have
\begin{equation}\label{eq:quadratic_well}
    g_z(\fbet^\star_c(z))-g_z(\lambda) \;\ge\; \tfrac{c^\star}{2}\,(\lambda-\fbet^\star_c(z))^2,
    \qquad \forall z\in A,\;\lambda\in[0,1-c]\,.
\end{equation}
by Taylor's theorem with integral remainder for $h(s)=g_z(\fbet^\star_c(z)+s(\lambda-\fbet^\star_c(z)))$, $s\in[0,1]$, using $h'(0)\le 0$ and $h''(s)\le -c^\star(\lambda-\fbet^\star_c(z))^2$ for all $s\in[0,1]$.

Fix $\rho>0$. We bound the posterior tail
$\pi_{t-1}(z;\{|\lambda-\fbet^\star_c(z)|>\rho\})$ uniformly in $z\in A$.
Let $M_c=\sup_{z\in\calZ,\lambda\in[0,1-c]}|\partial_\lambda^2 g_z(\lambda)|\le 1/c^2$, and $\Delta_t=\sup_{z\in\calZ,\,\lambda\in[0,1-c]}|\hat g_t(z,\lambda)-g_z(\lambda)|$. 

For
$|\lambda-\fbet^\star_c(z)|>\rho$, \eqref{eq:quadratic_well} gives
$g_z(\lambda)\le g_z(\fbet^\star_c(z))-c^\star\rho^2/2$. Combined with
$\hat g_{t-1}(z,\lambda)\le g_z(\lambda)+\Delta_{t-1}$ and
Assumption~\ref{ass:prior},
\begin{equation*}
    \int_{|\lambda-\fbet^\star_c(z)|>\rho} e^{(t-1)\hat g_{t-1}(z,\lambda)}\pi(\mathrm{d}\lambda)
    \;\le\; \overline\pi\cdot(1-c)\cdot
    \exp\left((t-1)[g_z(\fbet^\star_c(z))-c^\star\rho^2/2+\Delta_{t-1}]\right)\,.
\end{equation*}

The map $\lambda\mapsto g_z(\lambda)$
is $(1/c)$-Lipschitz on $[0,1-c]$ uniformly in $z$, because
$|\partial_\lambda g_z(\lambda)|\le 1/c$. Hence, for any $0<r<1-c$ and $\lambda\in [\fbet^\star_c(z)-r,\fbet^\star_c(z)+r]\cap[0,1-c]$,
\begin{equation*}
    g_z(\lambda) \;\ge\; g_z(\fbet^\star_c(z)) - r/c.
\end{equation*}
The
intersection of the $r$-ball with $[0,1-c]$ has Lebesgue measure at
least $r$. Thus
\begin{equation*}
    \int_0^{1-c} e^{(t-1)\hat g_{t-1}(z,\lambda)}\pi(\mathrm{d}\lambda)
   \geq \underline\pi\cdot r\cdot
    \exp\left((t-1)[g_z(\fbet^\star_c(z))-r/c-\Delta_{t-1}]\right).
\end{equation*}

Combining these two bounds (and canceling
$e^{(t-1)g_z(\fbet^\star_c(z))}$) yields
\begin{equation*}
    \pi_{t-1}\left(z;\{|\lambda-\fbet^\star_c(z)|>\rho\}\right)
   \leq\frac{\overline\pi(1-c)}{\underline\pi\, r}\,
    \exp\left((t-1)[-c^\star\rho^2/2 + r/c + 2\Delta_{t-1}]\right)\,.
\end{equation*}
Let $r = cc^\star\rho^2/4$. Then, the
exponent equals $(t-1)[-c^\star\rho^2/4+2\Delta_{t-1}]$. As $\Delta_t\to0$, almost surely, by Lemma~\ref{lem:uniformGC}, we have
$\Delta_{t-1}<c^\star\rho^2/16$ almost surely for all $t$ large. Hence,
\begin{equation}\label{eq:refine}
    \sup_{z\in A}\pi_{t-1}\bigl(z;\{|\lambda-\fbet^\star_c(z)|>\rho\}\bigr)
    \;\le\;\frac{4\overline\pi(1-c)}{\underline\pi\, c\, c^\star\rho^2}\,
    e^{-(t-1)\,c^\star\rho^2/8}\;\to\;0\,,
\end{equation}
almost surely. Finally, since, for any $z\in A$,
\begin{equation*}
    \bigl|\fbet_t^\up(z)-\fbet^\star_c(z)\bigr|
    \;\le\; \int_0^{1-c}|\lambda-\fbet^\star_c(z)|\,\pi_{t-1}(z;\mathrm{d}\lambda)
    \;\le\; \rho + (1-c)\,\pi_{t-1}(z;\{|\lambda-\fbet^\star_c(z)|>\rho\}),
\end{equation*}
it follows by~\eqref{eq:refine}, that $\limsup_{t\to\infty} \sup_{z\in A}|\fbet_t^\up(z)-\fbet^\star_c(z)\bigr| \leq \rho.$ Taking $\rho\downarrow 0$ yields the claim.
\end{proof}

Let
\begin{equation*}
    V=\left\{z:\Ex{\Qsymb}{{u}_z^{[k]}(Y)-{u}_z^{[k]}(X)}> \epsilon\right\}\,.
\end{equation*}
Then, for any $z\in V$ and $\lambda\in[0,1-c]$,
\begin{equation*}
    -\partial_\lambda^2 g_z(\lambda)
    \;=\;\Ex{\Qsymb}{\frac{D(z)^2}{(1+\lambda D(z))^2}}
    \;\ge\;\frac{\Ex{\Qsymb}{D(z)^2}}{(2-c)^2}
    \;\ge\;\frac{\Ex{\Qsymb}{D(z)}^2}{(2-c)^2}
    \;\ge\;\frac{\varepsilon^2}{(2-c)^2}\,.
\end{equation*}
Hence,~\eqref{eq:uniform_curvature} holds on $A=V$ and thus
Lemma~\ref{lem:posterior_concentration} applies:
\begin{equation}\label{eq:uniform_conv_V}
    \sup_{z\in V}\left|\fbet_t^\up(z)-\fbet^\star_c(z)\right|\to0,
    \qquad\text{almost surely as }t\to\infty\,.
\end{equation}

\textbf{Claim 1:} There exist $\varepsilon'>0$
and $N_1\ge 1$ such that \begin{equation*}
    g_z(\fbet_t^\up(z))>\varepsilon'/2, \quad \forall t\geq N_1, \forall z\in V\,.
\end{equation*}

\emph{Proof of claim 1:} We first derive a uniform positive lower
bound on $g_z(\fbet^\star_c(z))$ over $z\in V$:
For $\lambda^\star(z)\le 1-c$, we have,
by Taylor expanding around $0$ using $|g_z''|\le 1/c^2$,
\begin{equation}\label{eq:taylor_lower}
    g_z(\lambda) \;\ge\; \lambda\,\Ex{\Qsymb}{D(z)} - \frac{\lambda^2}{2c^2},
    \qquad \lambda\in[0,1-c].
\end{equation}
Wlog assume that $\varepsilon>0$ is chosen small enough such that
$\lambda_0=\varepsilon c^2 \le 1-c$. Then, by~\eqref{eq:taylor_lower}
\begin{equation}\label{eq:g_star_bound_interior}
    g_z(\lambda^\star(z)) \geq g_z(\lambda_0)
    \;\ge\; \varepsilon^2 c^2 - \frac{\varepsilon^2 c^2}{2}
    \;=\; \frac{\varepsilon^2 c^2}{2}.
\end{equation}
Otherwise, assume
 $\lambda^\star(z)> 1-c$. Then, since $g_z$ is
non-decreasing on $[0,\lambda^\star(z)]$, we obtain
\begin{equation}\label{eq:g_star_bound_boundary}
    g_z(1-c) \;\ge\; g_z(\lambda_0) \;\ge\; \frac{\varepsilon^2 c^2}{2}.
\end{equation}
\eqref{eq:g_star_bound_interior} and \eqref{eq:g_star_bound_boundary} together give the uniform bound
\begin{equation*}
    g_z(\fbet^\star_c(z)) \;\ge\; \frac{\varepsilon^2 c^2}{2}\;=:\;\varepsilon'
    \;>\;0,\qquad\forall z\in V.
\end{equation*}

The claim follows by the continuous mapping theorem and uniform convergence
\eqref{eq:uniform_conv_V}.

\textbf{Claim 2:} There exists $N_2\ge 1$
such that \begin{equation*}
    g_z(\fbet_t^\up(z))>-\varepsilon'\delta/4, \quad \forall t\ge N_2, \forall z\in V^c\,.
\end{equation*}

\emph{Proof of claim 2:} For $\varepsilon_3\in(0,1)$ small, partition $V^c$ into three sets:
\begin{align*}
    B_1 &= \bigl\{z\in V^c:\, D(z)=0\ \Qsymb\text{-a.s.}\bigr\},\\
    B_2 &= \bigl\{z\in V^c\setminus B_1:\,\Ex{\Qsymb}{D(z)^2}>\varepsilon_3\bigr\},\\
    B_3 &= V^c\setminus(B_1\cup B_2) \;=\;
    \bigl\{z\in V^c:\,0<\Ex{\Qsymb}{D(z)^2}\le\varepsilon_3\bigr\}.
\end{align*}

Then, $g_z\equiv 0$ for $z\in B_1$ and hence $g_z(\fbet_t^\up(z))=0$
for all $t$, which trivially satisfies the required bound.

Lemma~\ref{lem:posterior_concentration}
applied to $A=B_2$ (using $\inf_{z\in B_2,\,\lambda\in[0,1-c]}(-\partial_\lambda^2 g_z(\lambda))
\ge\varepsilon_3/(2-c)^2$), yields
\begin{equation*}
    \sup_{z\in B_2}\bigl|\fbet_t^\up(z)-\fbet^\star_c(z)\bigr|\;\to\;0
    \qquad\Qsymb\text{-a.s.}
\end{equation*}
Since $g_z(\fbet^\star_c(z))\ge g_z(0)=0$ for every $z$, and $g_z$ is uniformly continuous in $\lambda$, we get
$\liminf_{t}\inf_{z\in B_2}g_z(\fbet_t^\up(z))\ge 0$ $\Qsymb$-a.s. Hence,
there is $N_2$ such that $g_z(\fbet_t^\up(z))>-\varepsilon'\delta/4$
uniformly on $B_2$ for $t\ge N_2$.

Finally, by Jensen and $(1/c)$-Lipschitz continuity of $\log(1+\cdot)$ on $[-(1-c),1-c]$,
\begin{equation*}
    |g_z(\lambda)|
    \;\le\;\Ex{\Qsymb}{|\log(1+\lambda D(z))|}
    \;\le\;\frac{\lambda}{c}\,\Ex{\Qsymb}{|D(z)|}
    \;\le\;\frac{\lambda}{c}\,\sqrt{\Ex{\Qsymb}{D(z)^2}}
    \;\le\;\frac{1-c}{c}\,\sqrt{\varepsilon_3},
\end{equation*}
for all $z\in B_3$ and $\lambda\in[0,1-c]$. Choosing $\varepsilon_3 = [{(c\varepsilon'\delta)}/{(4(1-c))}]^2$, gives $|g_z(\fbet_t^\up(z))|\le\varepsilon'\delta/4$
for all $t$ and all $z\in B_3$, which concludes the proof of claim 2.

To conclude the proof, let $N=\max\{N_0,N_1,N_2\}$ and
$Z_t=\int\log S_t^{[k]}(\fbet_t^\up(z),z)\,\mathrm{d}\psi_t(z)$. By Fubini
and Claims 1 and 2, for all $t\ge N$,
\begin{align*}
    \mathbb{E}_\Qsymb(Z_t\mid\calF_{t-1})
    &=\int_V \Ex{\Qsymb}{\log S_t^{[k]}(\fbet_t^\up(z),z)\mid\calF_{t-1}}\mathrm{d}\psi_t(z) \\
    &\quad+\int_{V^c}\Ex{\Qsymb}{\log S_t^{[k]}(\fbet_t^\up(z),z)\mid\calF_{t-1}}\mathrm{d}\psi_t(z) \\
    &\ge\;\tfrac{\varepsilon'}{2}\,\psi_t(V) \;-\; \tfrac{\varepsilon'\delta}{4}\,\psi_t(V^c) \\
    &\ge\;\tfrac{\varepsilon'\delta}{2}\;-\;\tfrac{\varepsilon'\delta}{4}
    \;=\;\tfrac{\varepsilon'\delta}{4}\;>\;0\,.
\end{align*}
As $Z_t$ is bounded uniformly in $t$, the martingale SLLN gives
$t^{-1}\sum_{\ell=N}^t[Z_\ell-\mathbb{E}_\Qsymb(Z_\ell\mid\calF_{\ell-1})]\to0$
almost surely. Since
$\log(E_t)/t\ge t^{-1}\sum_{\ell=1}^t Z_\ell$ by Jensen, we conclude that
$\liminf_{t\to\infty}\log(E_t)/t\ge \varepsilon'\delta/4>0$ eventually, and hence
$E_t\to\infty$ almost surely.

\section{Affirming SD by testing the non-SD null hypothesis}\label{app:testing_non_SD_null}

\subsection{Problem formulation and an impossibility result}\label{sec:impossibility_non_sd_null}

The \emph{non-SD null} (in first order) at \eqref{eq:testing_non_sd} can be written as the union hypothesis
\begin{equation*}\label{eq:def_non_sd_null}
    \calH_0' = \{\Psymb \in \frakB \mid \Psymb_X \notsd \Psymb_Y \} =\{\Psymb \in \frakB \mid \exists z \in \calZ: \Psymb(X \leq z) < \Psymb(Y \leq z)\}=\bigcup_{z\in \calZ}  \calH_0'(z)\,,
\end{equation*}
for $\calH_0'(z) = \{\Psymb \mid \Psymb(X \leq z) < \Psymb(Y \leq z) \}$.
Moreover, $\calH_0'(z) \subseteq \calH_0(z)$, for all $z$, where the only difference between $\calH_0'(z)$ and $\calH_0(z)$ at \eqref{eq:FSD_null_at_thresholds} is that the former imposes a strict inequality whereas the latter allows for equality. That is, $\calH_0'(z)$ and $\calH_0(z)$ are equivalent from a statistical perspective, and testing for $\calH_0'$ is equivalent to testing for $\bigcup_{z\in \calZ} \calH_0(z)$. 

Intuitively, $\calH_0'$ is a very large and thus intractable null. 
Moreover, and in contrast to Lemma \ref{lem:FSD-null-is-convex}, the non-SD null $\calH_0'$ is \emph{not} convex, as is easily seen by the following example. 
\begin{example}
    Let $\calZ= \{z_1,z_2,z_3\}$ for $z_1<z_2<z_3$, and let $\Psymb^1,\Psymb^2 \in \frakB$ with
    \begin{equation*}
        \Psymb^1_X = \delta_{z_2}, \quad  \Psymb^2_X = \frac{3}{5}\delta_{z_1}+\frac{2}{5}\delta_{z_3}, \quad  \Psymb^1_Y= \Psymb^2_Y= \frac{1}{4}\delta_{z_1}+\frac{1}{2}\delta_{z_2}+\frac{1}{4}\delta_{z_3}\,.
    \end{equation*}
    Then, $\Psymb^i \in \calH_0'$, for $i=1,2$, but $(\Psymb^1+\Psymb^2)/2\notin \calH_0'$.
\end{example}
Combined with \citet[][Theorem 6.2]{zhang2024existence}, the non-convexity of $\calH_0'$ implies that we cannot expect to find a nontrivial e-variable for $\calH_0'$ as stated in Corollary~\ref{cor:no_e_variable_for_non_sd_null}.

\subsection{Testing the non-SD null under finite support}\label{sec:testing_union_directly}

First, we propose an e-value–based method for testing the non-SD null for finite support.
\begin{assumption}[Finite support]\label{assn:finite_support}
    For $m\geq 2$, let $\calZ=\{z_i\}_{i=1}^m$ with $z_1<z_2< \ldots < z_m$.
\end{assumption}
As $F(z_m)=1$ for any distribution on $\calZ$, we have $\calH_0'(z_m)=\emptyset$ and $\calH_0' = \cup_{i=1}^{m-1} \calH_0'(z_i)$. That is, it is sufficient to consider the thresholds $z_1<\ldots <z_{m-1}$ in the sequel.
Since we are facing an intersection-union test, a natural approach is to base our test decision on a minimum test statistic~\cite[e.g.,][]{kaur1994testing,davidson2013testing}. 
It is routine to check that, for $\calH_0'(z_i)$-e-variables $E_i$, the minimum $\min_{i\leq m-1} E_i$ is an e-variable for $\calH_0'$. 
This mimics the standard procedure to take the maximum of the individual p-values for an intersection-union test. 
Although statistically valid, the minimum approach can be very conservative in practice, as it is often possible to design more powerful e-variables by incorporating dependencies between the individual hypotheses. 
Nevertheless, using the growth-rate optimal e-variables for the individual hypotheses $\calH_0'(z_i)$ from Section \ref{subsec:GRO_fixed_threshold}, we can show that the min-approach already achieves the best possible asymptotic growth rate.

As argued in Section~\ref{sec:testing_non_SD_null}, making the testing problem feasible requires a minimal separation between the distributions of $X$ and $Y$. 
In the finite case, the ``restriction'' of the null is not necessary as we can test $\calH_0' = \calH_0'(\tilde\calZ)$ with $\tilde\calZ = \{z_1, \dotsc, z_{m-1}\}$. 
For $0<\varepsilon<1$, we let
\begin{equation*}
    \calQ(\varepsilon) = \incurly{\Qsymb \in \frakB \mid \Qsymb(Y \leq z_i) +\varepsilon \leq \Qsymb(X \leq z_i), \; \forall i=1, \ldots, m-1}.
\end{equation*}

\begin{proposition}
[Validity and asymptotic optimality of the minimum e-process under finite support]
\label{ppn:min_approach}
Suppose that Assumption \ref{assn:finite_support} is satisfied, i.e., $X$ and $Y$ have finite support. 
\begin{enumerate}[(a)]
    \item Let $S_t(\lambda,z_i)=1+\lambda(\indicator{X_t\leq z_i}-\indicator{Y_t\leq z_i})$, and
    $(\lambda_t(z_i))_{t\in \N} \subseteq [0,1]$ predictable for all $i=1,\ldots, m-1$. Then, 
    \begin{equation}\label{eq:min_e_process}
        \e_t = \min_{i\leq m-1} \left(\prod_{\ell=1}^t S_\ell(\lambda_\ell(z_i),z_i)\right), \quad t\in \N, 
    \end{equation}
    is an e-process for $\calH_0'$. \label{item:min_eprocess}
    \item If, for any $i$, $\lambda_{t}(z_i)$ is chosen as the plug-in GRO bet~\eqref{eq:plugin_GRO_lambda}, then, for any $\Qsymb \in \calQ(\varepsilon)$,\label{item:min_eprocess_has_power_one}
    \begin{equation*}\label{eq:asymptotic_power_one_for_min_approach}
        \liminf_{t\to \infty} \e_t = \infty, \quad \text{$\Qsymb$-a.s.}
    \end{equation*}
    \item 
    For any $\Qsymb\in \calQ(\varepsilon)$ and $\epsilon>0$, there exists $c_0>0$ such that for any clipping constant $0<c<c_0$ in~\eqref{eq:plugin_GRO_lambda} and for any other e-process $(W_t)_{t\in \N}$ for $\calH_0'$, we have, \label{item:min_eprocess_is_asymptotically_log-optimal}
    \begin{equation*}
        \liminf_{t\to \infty} \frac{1}{t}\left(\log{\e_t}-\log{W_t}\right) \geq -\epsilon, \quad \Qsymb\textrm{-a.s.}
    \end{equation*}
\end{enumerate}
\end{proposition}

 \begin{proof}
Since $\calH_0'(z_i) \subset \calH_0(z_i)$, each \begin{equation*}
    E_t(z_i) = \prod_{\ell=1}^t S_{\ell}(\lambda_{\ell}(z_i),z_i), \quad t\in \N,
\end{equation*} 
is an e-process for $\calH_0'(z_i)$ by Theorem~\ref{thm:Test-supermartingale-for-the-FSD_null}. 
The anytime-validity of the minimum e-process~\eqref{eq:min_e_process} follows immediately by Jensen: for any $\Psymb \in \calH_0'$ and any stopping time $\tau$,
\begin{equation*}
    \Ex{\Psymb}{\e_\tau} = \Ex{\Psymb}{\min_{i \in \{1, \dotsc, m-1\}} E_\tau(z_i)} \leq \min_{i \in \{1, \dotsc, m-1\}} \Ex{\Psymb}{E_\tau(z_i)} \leq 1. 
\end{equation*}

For the second claim, fix any $\Qsymb\in \calQ(\varepsilon)$. Then, $\Qsymb \in (\calH_0'(z_i))^c$ for all $i$. In the proof of Theorem~\ref{thm:asymptotic_power_one}, it is shown that, for all $i=1, \ldots, m-1$, $\hat{\lambda}_t^\textsf{GRO}(z_i) \to \lambda^\star(z_i) \land (1-c)$ almost surely, and thus, by the SLLN
\begin{equation*}
    \frac{\log(E_t(z_i))}{t} \to g_{z_i}(\lambda^\star(z_i) \land (1-c))>0,
\end{equation*}
for $g_{z_i}$ and $\lambda^\star(z_i)$ defined as in the proof of Theorem \ref{thm:asymptotic_power_one}. As $m$ is finite, it follows by uniform convergence that, as $t\to\infty$,
\begin{equation}\label{eq:proof_prop_eq1}
        \frac{\log(\e_t)}{t}=\min_{i=1, \ldots, m-1}\frac{\log\left(E_{t}(z_i)\right)}{t}\to \min_{i=1, \ldots, m-1}\{g_{z_i}(\lambda^\star(z_i) \land (1-c))\} >0, \quad \Qsymb-\textrm{a.s.}
    \end{equation}   
Regarding the last claim, let $\Qsymb\in \calQ(\varepsilon)$, and let $(W_t)_{t\in \N}$ be any e-process for $\calH_0'$. Then, 
\begin{equation*}
    \limsup_{t\to \infty} \frac{1}{t}\log W_t \leq \sup_{S\in \mathcal{E}(\calH_0')}\Esymb_{\Qsymb}\log(S), \quad \Qsymb-\textrm{a.s.},
\end{equation*}
for $\mathcal{E}(\calH)$ denoting the family of all e-variables for some null $\calH\subseteq \frakB$, see e.g.\ \citet[Theorem 2.1]{waudbysmith2025universal}.
As $\calH_0'(z_i) \subseteq \calH_0'$, and thus $\mathcal{E}(\calH_0'(z_i))\supseteq\mathcal{E}(\calH_0')$, for all $i$, we have
\begin{equation*}
    \sup_{S\in \mathcal{E}(\calH_0')}\Esymb_{\Qsymb}\log(S) \leq \sup_{S\in \mathcal{E}(\calH_0'(z_i))}\Esymb_{\Qsymb}\log(S)= g_{z_i}(\lambda^\star(z_i)),
\end{equation*}
for all $i$, and thus 
\begin{equation}\label{eq:proof_prop_eq2}
    \limsup_{t\to \infty} \frac{1}{t}\log W_t \leq \min_{i=1, \ldots, m-1}\{g_{z_i}(\lambda^\star(z_i))\},  \quad \Qsymb-\textrm{a.s.}
\end{equation}
By continuity of $\lambda \mapsto g_{z_i}(\lambda)$ for all $i\leq m-1$, we may choose $c_0>0$, such that, \begin{equation*}
    \left| \min_{i=1, \ldots, m-1}\{g_{z_i}(\lambda^\star(z_i) \land (1-c))\} -\min_{i=1, \ldots, m-1}\{g_{z_i}(\lambda^\star(z_i))\} \right| < \epsilon\,,
\end{equation*}for all $0<c<c_0$. That is, by \eqref{eq:proof_prop_eq1}, the minimum process $(\e_t)_{t\in \N}$ approximately (up to $\epsilon$-difference) achieves  the upper bound on the asymptotic growth rate in~\eqref{eq:proof_prop_eq2}. 
\end{proof}

Part~\ref{item:min_eprocess} of Proposition~\ref{ppn:min_approach} shows that the minimum process $(\e_t)_{t\in \N}$ yields an anytime-valid test for the non-SD null, and \ref{item:min_eprocess_has_power_one} ensures that we will eventually reject if the distribution of $X$ and $Y$ is properly separated from the null. 
Part~\ref{item:min_eprocess_is_asymptotically_log-optimal} shows that, asymptotically, one cannot hope for evidence to grow faster than the minimum e-process. 
Following~\citet{waudbysmith2025universal}, we call any process satisfying~\ref{item:min_eprocess_is_asymptotically_log-optimal} \emph{asymptotically log-optimal}.

\subsection{Affirming restricted SD via time-uniform CDF bands}\label{sec:testing_nonsd_tvCS}

Given sequentially observed data $X_1, X_2, \ldots \sim F_X$, a $(1-\alpha)$-level \emph{time-uniform CDF band \emph{or} CDF confidence sequence (CDF-CS)} is a set of intervals $\{(\hat{L}_t(z), \hat{U}_t(z))\}_{t \geq 1, z \in \calZ}$ such that
\begin{equation*}\label{eqn:tvCS}
    \prob{\forall t \geq 1, \, \forall z \in \calZ: \hat{L}_t(z) \leq F_X(z) \leq \hat{U}_t(z) } \geq 1-\alpha, \quad \forall \Psymb\in \frakB.
\end{equation*}
We refer to $\hat{L}_t(z)$ and $\hat{U}_t(z)$ as the \emph{lower} and \emph{upper confidence sequence} of the CDF band (\emph{LCS} and \emph{UCS}), respectively.
\citet{howard2022sequential} and \citet{mineiro2023nonstationarity} are among the first to derive time-uniform CDF bands, while \citet{clerico2026uniform} recently derived tighter CDF-CS's based on PAC-Bayes techniques. 
When testing the SD null (not the non-SD null), the empirical results in Section~\ref{sec:empirical_comparison_ks_tests} show these methods achieve rather suboptimal statistical power, relative to the sequential e-tests proposed in this paper. 

Nevertheless, we may reasonably suspect that, if we have a tight CDF-CS for each CDF, then we can affirm SD when the two CDF-CS's do not intersect with each other at all $z$. 
However, we quickly run into an issue when operationalizing this idea. 
Suppose, for example, that $X$ and $Y$ are bounded in $[0, 1]$. 
Then, by definition, there is no difference in the CDF values at the boundaries: $F_X(1) - F_Y(1) = 0$. 
For unbounded rv's, say taking values in $\R$, we analogously know that $\lim_{z \to -\infty}[F_X(z) - F_Y(z)] = \lim_{z \to +\infty}[F_X(z) - F_Y(z)] = 0$. 
Even outside these ``boundary cases,'' whenever $F_X$ and $F_Y$ touch at an arbitrary point $z$ or in a small tail region, then we cannot distinguish the two CDFs; for a concrete example, recall the simulation example of Section~\ref{sec:sim_kuniform}. 

To make progress, we need to invoke a relaxation of SD in the alternative. Following \citet{davidson2013testing}, we say that $Y$ \emph{stochastically dominates $X$ restricted to} $\tilde\calZ \subseteq \calZ$, denoted by $X \sd_{\tilde\calZ} Y$, if
    \begin{equation*}
        F_Y(z) \leq F_X(z), \quad \forall z \in \tilde\calZ.
    \end{equation*}
For any $\tilde\calZ \subsetneq \calZ$, SD on $\tilde\calZ$ is a strict relaxation of SD on $\calZ$. 
Thus, by restricting ourselves to a smaller support, we can define a smaller null hypothesis that is easier to reject. 

Next, we specify the set of alternatives for which affirming SD is possible by generalizing the $\varepsilon$-separated alternative for finite support. 
For $\tilde\calZ \subseteq \calZ$, and $\varepsilon \in (0,1)$, let
\begin{equation*}
    \calQ(\tilde\calZ, \varepsilon) = \incurly{\Qsymb \in \frakB \mid \Qsymb(Y \leq z) + \varepsilon \leq \Qsymb(X \leq z), \; \forall z \in \tilde\calZ} \subseteq \incurly{\Qsymb \in \frakB \mid X \sd_{\tilde\calZ} Y \text{ under } \Qsymb}.
\end{equation*}

Here, we remain agnostic about the size of the witness set $\tilde\calZ$ and the value of $\epsilon$. 
As \citet{davidson2013testing} point out, affirming SD is empirically sensible only over a restricted range of outcomes in many practical applications. 
For example, if $X$ and $Y$ represent the incomes in two economies, then it may be more relevant to affirm SD over the range of incomes that are likely to be observed in practice. 

Given a sensible alternative, we can construct a powerful test that rejects the \emph{restricted} non-SD null, $\calH_0'(\tilde\calZ): X \notsd_{\tilde\calZ} Y$, as soon as the two CDF-CS's do not intersect at all $z \in \tilde\calZ$. 
\begin{proposition}[Affirming restricted SD using CDF-CS's]\label{ppn:affirming_sd_tvcs_witness_set}
Let $\alpha \in (0,1)$ and $\tilde\calZ \subseteq \calZ$. 
Let $\{\hat{L}_{X,t}(z)\}_{t\in\N, z \in \tilde\calZ}$ be a $(1-\alpha/2)$-level LCS for $F_X$, and $\{\hat{U}_{Y,t}(z)\}_{t\in\N, z\in \tilde\calZ}$ be a $(1-\alpha/2)$-level UCS for $F_Y$, on $\tilde{\calZ}$. 
Consider the test that rejects $\calH_0'(\tilde\calZ)$ at time $t$ only when the UCS on $F_Y$ is entirely below the LCS on $F_X$:
\begin{equation}\label{eq:seqtest_tvCS}
    \phi_t = \indicator{\forall z \in \tilde\calZ: \hat{U}_{Y,t}(z) < \hat{L}_{X,t}(z)}, \quad  t \in \N\,.
\end{equation}
Then, the following statements hold.
\begin{enumerate}[(a)]
    \item (Anytime-validity.) $(\phi_t)_{t\geq 1}$ is an anytime-valid test for $\calH_{0}'(\tilde\calZ)$ at level $\alpha$. \label{item:seqtest_tvCS_validity}
    \item (Test of power one.) Define the stopping time $\tau_\alpha = \inf\{t \geq 1: \phi_t = 1\}$.
    Let $\varepsilon \in (0,1)$. For any $\Qsymb \in \calQ(\tilde\calZ, \varepsilon)$, suppose that the UCS for $F_Y$ and the LCS for $F_X$ approach their respective CDFs, uniformly on $\tilde\calZ$, in large samples: as $t \to \infty$,
    \begin{equation}\label{eq:tvCS_uniform_convergence}
        \sup_{z \in \tilde\calZ} \insquare{\hat{U}_{Y,t}(z) - F_Y(z)} \to 0, \quad \text{and} \quad \sup_{z \in \tilde\calZ} \insquare{F_X(z) - \hat{L}_{X,t}(z)} \to 0, \quad \text{$\Qsymb$-a.s.}
    \end{equation}
    Then, $(\phi_t)_{t\in\N}$ is a sequential test of power one against $\Qsymb$, that is, $\Qsymb\inparen{\tau_\alpha < \infty} = 1$. \label{item:seqtest_tvCS_power}
\end{enumerate}
\end{proposition} 

\begin{proof}
(a) Define the ``good'' events $E^X$ and $E^Y$ as
\begin{align*}
    E^X = \incurly{\forall t \geq 1,\, \forall z \in \tilde\calZ: \hat{L}_{X,t}(z) \leq F_X(z)}, \;\; \textrm{and}\;\;
    E^Y = \incurly{\forall t \geq 1,\, \forall z \in \tilde\calZ: F_Y(z) \leq \hat{U}_{Y,t}(z)},
\end{align*} 
such that $\prob{E^X} \geq 1-\alpha/2$ and $\prob{E^Y} \geq 1-\alpha/2$ for any $\Psymb\in \frakB$. 

Suppose that there exists a $t\geq 1$ such that $\hat{U}_{Y,t}(z) < \hat{L}_{X,t}(z)$ for all $z \in \tilde\calZ$. 
Then, on the event $E^X \cap E^Y$, we have
\begin{equation*}
    F_Y(z) \leq \hat{U}_{Y,t}(z) < \hat{L}_{X,t}(z) \leq F_X(z), \quad \forall z \in \tilde\calZ,
\end{equation*}
which contradicts the definition of $\calH_0'(\tilde\calZ)$.
Thus, for any $\Psymb \in \calH_0'(\tilde\calZ)$,
\begin{equation*}
    \prob{\exists t \geq 1: \phi_t = 1} \leq \prob{\inparen{E^X \cap E^Y}^c} \leq \prob{(E^X)^c} + \prob{(E^Y)^c} \leq \alpha.
\end{equation*}
(b) Let $\Qsymb \in \calQ(\tilde{\calZ}, \epsilon)$. 
Given that the LCS and UCS uniformly converge to the true CDFs, we can choose $N( \epsilon)$ large such that, $\Qsymb$-almost surely,
\begin{equation*}
    \hat{L}_{X,t}(z) > F_X(z) -  \epsilon/2 \quad \text{and} \quad \hat{U}_{Y,t}(z) < F_Y(z) +  \epsilon/2, \quad \forall z \in \tilde\calZ,\, \forall t \geq N( \epsilon).
\end{equation*}

Then, for any $t \geq N( \epsilon)$ and $z \in \tilde{\calZ}$, we have $\Qsymb$-almost surely
\begin{equation*}
    \hat{U}_{Y,t}(z) - \hat{L}_{X,t}(z) < F_Y(z) +  \epsilon/2 - (F_X(z) -  \epsilon/2) = F_Y(z) - F_X(z) +  \epsilon \leq 0,
\end{equation*}
where the last inequality uses the fact that $\Qsymb$ is in the $ \epsilon$-witness set $\calQ(\tilde{\calZ},  \epsilon)$ restricted to $\tilde{\calZ}$. Thus, $\Qsymb(\tau_\alpha < \infty) = 1$ for any $\Qsymb \in \calQ(\tilde{\calZ}, \epsilon)$. 
\end{proof}

Note that the condition~\eqref{eq:tvCS_uniform_convergence} is a mild, one-sided condition that we expect to hold for any reasonable CDF band for i.i.d. data.
One practical choice is the recently developed CDF-CS by \citet{clerico2026uniform}, which is empirically tight for time-uniform CDF estimation and has widths shrinking at an asymptotic rate of $O(\sqrt{T^{-1}\log T})$ under ``regular'' conditions (see their remark on asymptotic rates in Section 4).

\begin{remark}[Necessity of $\epsilon$-separation for power] 
    The above result, along with the counterexamples discussed earlier, establishes that tests of the form~\eqref{eq:seqtest_tvCS} are valid for the non-SD null and achieve power only when the associated CDFs are $\epsilon$-separated. 
    See further discussion in Section~\ref{app:ks_test}.
\end{remark}

\section{Kolmogorov-Smirnov-type sequential tests for FSD}\label{app:ks_test}

In the non-anytime-valid setting, the classical approach to testing FSD is based on the one-sided KS statistic~\eqref{eqn:ks_stat}. In this section, we first review the pioneering work of \citet{darling1968some}, who were among the first to propose an anytime-valid test for FSD based on the KS statistic. Alternatively, we can adapt existing methods for time-uniform CDF estimation~\citep{howard2022sequential,manole2023martingale} to obtain simple, KS-type sequential tests for FSD. 
We briefly describe how to obtain such tests, and show how they ultimately fall short in terms of both applicability and power compared to our proposed methods.

\subsection{The original sequential KS test of Darling and Robbins}
\label{app:KS-darling-robbins}

For two independent i.i.d.\ samples
$X_1, X_2, \dots \sim F_X$ and $Y_1, Y_2, \dots \sim F_Y$, \citet{darling1968some} propose to test for FSD by
comparing the one-sided empirical KS statistic
\begin{equation*}
    \hat{D}^\textsf{KS}_t = \sup_{z \in \mathcal{Z}}
  \big[\widehat{F}_{X,t}(z) - \widehat{F}_{Y,t}(z)\big], \quad t\in \N,
\end{equation*}
against a deterministic, shrinking boundary $f(t)/t$, $t\in \N$. Here $f$ is any
continuous, positive, nondecreasing and concave function on $[m,\infty)$, for
some user-chosen $m \in \mathbb{N}$, such that $f(x) \le x$, $f(x)/x$
decreases strictly to $0$, and, for $\alpha\in (0,1)$,
\begin{equation} \label{eq:DR-boundary-condition}
  \sum_{t \ge m} \exp\left(-\frac{f^2(t)}{t+1}\right) \le \alpha .
\end{equation}
They propose the sequential test
\begin{equation}
  \phi_t = \mathbf{1}\!\left(\hat{D}^\textsf{KS}_t > f(t)/t\right), \quad t \ge m,
  \label{eq:DR-test}
\end{equation}
and show that~\eqref{eq:DR-test} defines a valid level-$\alpha$ sequential test for $\calH_0 : Y \preceq_1 X$, with asymptotic power one under any distribution in the alternative $\calH_1 : Y \not\preceq_1 X$; see their Theorem~1.
Type-I error control follows from the exact combinatorial null distribution of the two-sample KS statistic given by \citet{GnedenkoKorolyuk1951}, together with a union bound over all $t \geq m$. Consequently, the procedure is anytime-valid, although its validity is not derived from a (super)martingale or e-process construction.

As a concrete instantiation, their Remark~4(d) suggests the family
\begin{equation*}
f(x) = [(x+1)(a \log x + \log b)]^{1/2},
\qquad a > 1, b \ge 1,
\end{equation*}
and shows that, for $\alpha = 0.05$, choosing $a = 2$, $b = 4$, and $m = 6$ ensures that condition~\eqref{eq:DR-boundary-condition} is satisfied. The resulting rejection boundary is
\begin{equation*}
     \frac{f(t)}{t}
  = \frac{\sqrt{(t+1)\,(2\log t + \log 4)}}{t}
  \sim \sqrt{\frac{2\log t}{t}} .
\end{equation*}
Comparing with the boundary $2\omega_t(A,\alpha) \sim
2A\sqrt{\log\log t / t}$ of the CDF-band construction in the next subsection, the two are of comparable width for moderate $t$, 
but the Darling--Robbins boundary is asymptotically wider by a factor of order
$\sqrt{\log t / \log\log t}$. This is unavoidable for their construction:
condition \eqref{eq:DR-boundary-condition} forces $f(x) \gtrsim \sqrt{x\log x}$,
as they observe in their Remark~4(f), where they also note that boundaries of
the LIL order $\sqrt{x \log\log x}$ should exist by the results of
\citet{chung1949estimate}. This gap was closed later by the finite-LIL bounds of \citet{howard2022sequential} discussed
below.

We remark that the boundary $f(t)/t$ depends on $t$ alone and is therefore not adaptive to the underlying distributions, which may result in low power, as we illustrate in Section~\ref{sec:empirical_comparison_ks_tests}. Moreover, their validity argument requires two independent samples of equal size and therefore cannot exploit paired observations that are correlated. Finally, because their proof relies on rank-based combinatorial arguments for the empirical CDF process, the approach does not naturally extend to higher-order stochastic dominance.

\subsection{KS-type sequential tests based on time-uniform CDF bands}\label{sec:general_ks_test_from_tvcs}

\citet[][Appendix B.2]{howard2022sequential} propose a time-uniform bound on the empirical process $(\hat{F}_t(\cdot) - F(\cdot))_{t\in \N}$, referred to as a ``finite law-of-iterated-logarithm (LIL) bound,'' that yields a sequential test for FSD. 
We summarize the construction in the following proposition, which parallels Proposition~\ref{ppn:affirming_sd_tvcs_witness_set}.

\begin{proposition}[KS-type sequential test for FSD]\label{ppn:testing_sd_from_tvcs}
    For $\alpha \in (0,1)$,
let $\{\hat{L}_{X,t}(z)\}_{t \in \N, z \in \calZ}$ be a $1-\alpha/2$-LCS for $F_X$ and $\{\hat{U}_{Y,t}(z)\}_{t \in \N, z \in \calZ}$ a $1-\alpha/2$-UCS for $F_Y$.
    Then,
    \begin{equation*}
        \phi_t = \indicator{ \sup_{z \in \calZ} \insquare{\hat{L}_{X,t}(z) - \hat{U}_{Y,t}(z)} > 0} , \quad t\in \N,
    \end{equation*}
    is a valid level-$\alpha$ sequential test for $\calH_0: Y \fsd X$ against $\calH_1: Y \notfsd X$. 
\end{proposition}
The test will reject as soon as the LCS and UCS cross at some point $z \in \calZ$. 
As a concrete example, \citet[][Theorem 2]{howard2022sequential} derive the LCS $\hat{L}_{X,t}(z) = \hat{F}_{X,t}(z) - \omega_t(A, \alpha)$ and UCS $\hat{U}_{Y,t}(z) = \hat{F}_{Y,t}(z) + \omega_t(A, \alpha)$, where $\omega_t$ is the uniform width function
\begin{equation*}
    \omega_t(A, \alpha) = A\sqrt{\frac{\log\log(et/t_{\text{min}}) + C(A, \alpha)}{t}}, 
\end{equation*}
$C(A, \alpha)$ is a constant that is calculated numerically, $A \geq 1/2$ is a tuning parameter (set to $0.85$ in their paper), and $t_{\text{min}}$ is a user-specified minimum time after which the bound starts to hold.
Note that the width merely depends on $t$ and not on the data, so the resulting test is not adaptive to the underlying distributions. 
The corresponding level-$\alpha$ test~\citep{darling1968some} is simply
\begin{equation*}
    \phi_t = \indicator{\sup_{z \in \calZ} \insquare{\hat{F}_{X,t}(z) - \hat{F}_{Y,t}(z)} > 2\omega_t(A, \alpha)}, \quad t\in \N,
\end{equation*}
involving the empirical one-sided KS statistic on the left-hand side. 

Interestingly, \citet[Corollary S3]{howard2022sequential} instead propose 
\begin{equation}\label{eq:infimum_test_tvCS}
    \phi_t = \indicator{\inf_{z \in \calZ} \insquare{\hat{F}_{X,t}(z) - \hat{F}_{Y,t}(z)} > 2\omega_t(A, \alpha)}=\indicator{ \inf_{z \in \calZ} \insquare{\hat{L}_{X,t}(z) - \hat{U}_{Y,t}(z)} > 0},
\end{equation}
for the FSD null $\calH_0: Y \sd X$. However, as discussed in the previous section, this test is actually valid for the strictly larger \emph{non-FSD} null $\calH_0': X \not\prec Y$, and is therefore expected to be conservative when used for testing the FSD null. 
It achieves no power against any non-FSD alternative where neither $X$ dominates $Y$ nor $Y$ dominates $X$ (whenever the two CDFs cross). 
In fact, even when $Y$ strictly dominates $X$ but $F_Y$ has a small contact region with $F_X$, the infimum difference $\inf_{z \in \calZ} [F_X(z) - F_Y(z)]$ will be zero, so the test~\eqref{eq:infimum_test_tvCS} will not retain power (as long as the LCS and UCS are valid and uniformly approach these CDFs). 
This problem persists when we instead treat~\eqref{eq:infimum_test_tvCS} as a test for the non-FSD null: as we have demonstrated in the previous section, it can only achieve power when $F_X$ and $F_Y$ are sufficiently separated. 

While we focused here on \citet{howard2022sequential}'s sequential test for SD (or non-SD), the analogous construction can work with any CDF-CS, such as those developed in \citet{mineiro2023nonstationarity} and \citet{clerico2026uniform}.
A final remark is that again it is unclear how these methods generalize beyond FSD.

\subsection{A KS-type sequential test based on reverse submartingales}\label{sec:ks_test_from_reverse_submtg}

In another line of work, building on the theory of \emph{reverse submartingales}, \citet{manole2023martingale} derive lower and upper confidence sequences for the two-sided KS statistic
\begin{equation*}
    D^\textsf{KS,2} = \lVert F_X - F_Y \rVert_{\infty} = \sup_{z \in \calZ} |F_X(z) - F_Y(z) |
\end{equation*}
based on the empirical CDFs from i.i.d.~samples; see Corollary 14 in their paper. 
To adapt their result to FSD testing, as discussed in Section~\ref{sec:related_work_testing_SD}, we may rewrite the FSD null as $ \calH_0 =\{\Psymb \in \frakB \mid  D^\mathsf{KS}(\Psymb) \leq 0\}$ with the one-sided KS statistic
\begin{equation*}
    D^\mathsf{KS} = D^\mathsf{KS}(\Psymb) = 
    \sup_{z \in \calZ}\, [\Psymb_X(z) - \Psymb_Y(z)]\,.
\end{equation*}
Following their exposition, we directly proceed to the more general case of unpaired samples with $X_1, \dotsc, X_t \sim F_X$ and $Y_1, \dotsc, Y_s \sim F_Y$.
From their proofs, we can readily extend their results on the two-sided KS statistic to obtain a time-uniform upper confidence sequence $(U_{ts})_{t,s\in \N}$ at level $\alpha\in (0,1)$ for the one-sided process 
\begin{equation*}
    M_{ts} = \sup_{z \in \calZ} \left\{\hat{F}_{X}^t(z)-\hat{F}_{Y}^s(z)\right\}- D^\mathsf{KS}(\Psymb), \quad t\in \N, s\in \N,
\end{equation*}
where $\hat{F}_{X}^t$ and $\hat{F}_{Y}^s$ denote the empirical CDFs of the samples $(X_i)_{i=1}^t$ and $(Y_i)_{i=1}^s$, respectively. This UCS $(U_{ts})_{t,s\in \N}$ for $(M_{ts})_{t,s\in \N}$ directly yields the LCS $$
L_{ts}=\sup_{z \in \calZ} \left\{\hat{F}_{X}^t(z)-\hat{F}_{Y}^s(z)\right\}-U_{ts}, \quad t\in \N,s\in \N,$$ 
for the unknown quantity of interest $D^\mathsf{KS}(\Psymb)$, satisfying
\begin{equation}\label{eq:Manole_Ramdas_lower_confidence_sequence}
    \Psymb(L_{ts} \leq D^\mathsf{KS}(\Psymb), \forall t, s \geq 1) \geq 1-\alpha, \quad \textrm{for all }\Psymb \in \frakB\,.
\end{equation}
The deduced sequential test $\psi_{ts} = \indicator{L_{ts} > 0}, t,s\in \N,$ which rejects $\calH_0$ as soon as $L_{ts}$ exceeds zero, time-uniformly controls type-I error at level $\alpha$ by the coverage guarantee~\eqref{eq:Manole_Ramdas_lower_confidence_sequence}, that is the probability that we ever reject under $\calH_0$ is bounded by $\alpha$. 

Although theoretically appealing, this approach is not competitive in terms of power as the confidence sequences are generally very conservative (see Section~\ref{sec:empirical_comparison_ks_tests}).

\section{Additional experiments and details}\label{app:additional_experiments}

\subsection{Additional details on simulations}\label{app:additional_details_sims}

All of the simulations presented in the main text and below are run with Python code 
\if\arxiv1
{(available at \url{https://github.com/yjchoe/BettingOnBets}).}
\else{(provided with the paper).}
\fi
For the classical baseline methods (BD03, LMW05, and LSW10) described in Figure~\ref{fig:comparison_classical_tests}, we use the publicly available \texttt{PySDTest}~\citep{lee2023pysdtest} package for the implementation of these baseline methods. 
In each case, we fix the number of bootstrap/subsampling trials to be $n_\text{boot}=500$, and use the resampling set size of $b=40$ for the subsampling methods. 
For the contact set estimation in LSW10~\citep{linton2010improved}, we use the default hyperparameter of $c=0.75$, which controls how quickly the size of the estimated contact set grows as the sample size increases. 
This value is different from the original paper's recommendation of $c=3.0$ or $c=4.0$, but we have found that the package default performs better in terms of mean rejection times for this simulation.

\subsection{The toy example with antimonotonicity}\label{sec:sim_anticorr}

This simulation is based on Example~\ref{ex:anticorr}, where the data distribution has a finite support and $X$ and $Y$ are maximally antimonotonic. 
We simulate $\{(X_t, Y_t)\}_{t \in \N}$ as jointly i.i.d.~with
\begin{align*}
    \prob{X=0, Y=1} &= 1/2 \quad \text{and} \quad 
    \prob{X=2/3, Y=1/3} = 1/2.
\end{align*}
The combined support of $X$ and $Y$ is $\calZ = \{0, 1/3, 2/3, 1\}$, and Figure~\ref{subfig:anticorr_cdfs} plots their marginal CDFs. 
By design, we have $X \fsdstrict Y$, and thus $Y \notfsd X$; we also have $\rho(X, Y)=-1$. 
While conceptually simple, the example highlights robustness to the dependence structure between $X$ and $Y$. 
The antimonotonic dependence is the most challenging case, in the sense that we observe a larger $X$ than $Y$ half of the time. 
Note that the support $\calZ$ is finite and known; see Remark~\ref{remark:finite_support} for our previous discussion on adaptivity. 

\begin{figure}[t]
    \centering
    \def\imageheight{4cm} 

    \begin{subfigure}[t]{0.3\textwidth}
        \centering
        \includegraphics[height=\imageheight, keepaspectratio]{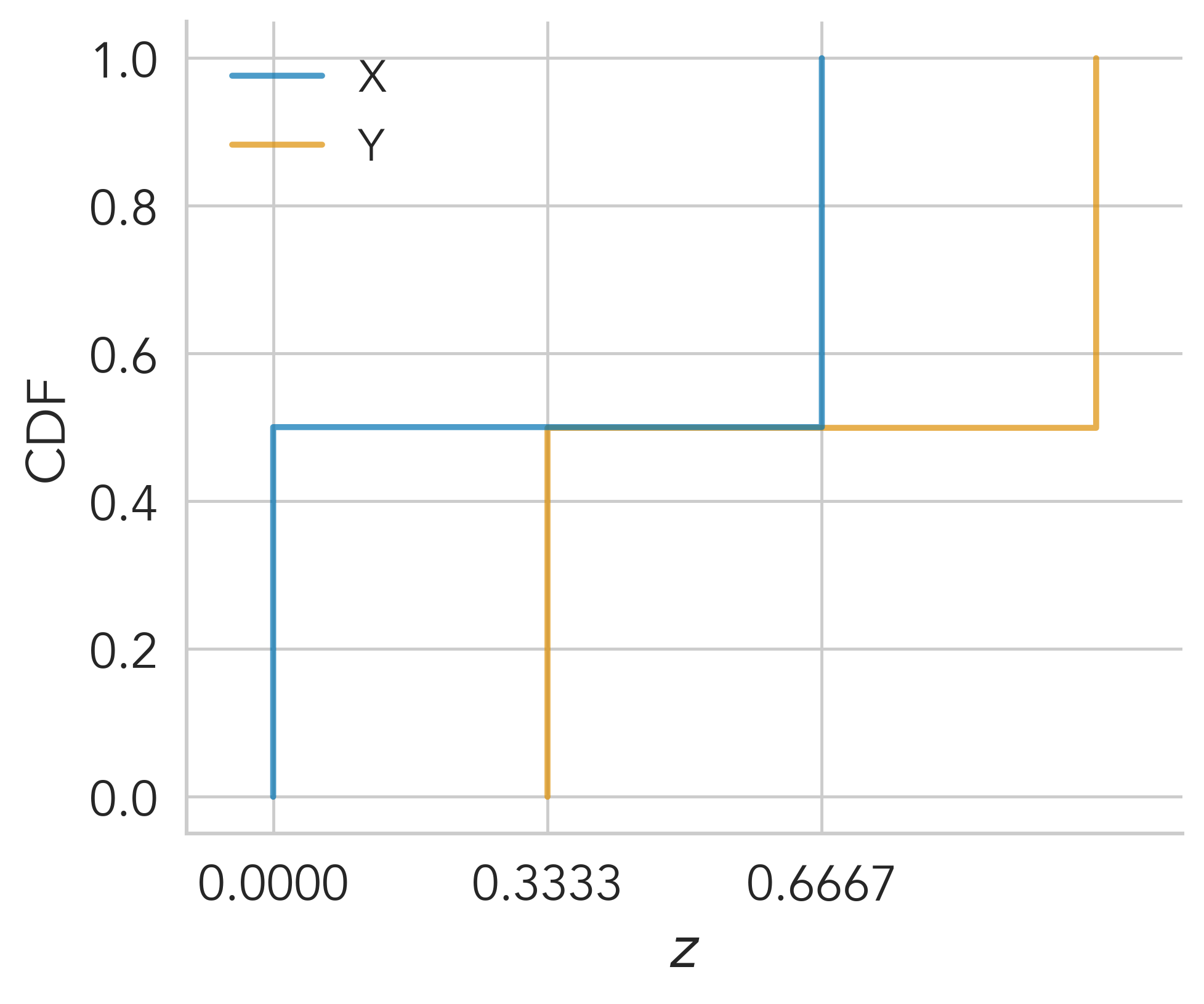}
        \caption{$F_X$ and $F_Y$.}
        \label{subfig:anticorr_cdfs}
    \end{subfigure}
    \begin{subfigure}[t]{0.34\textwidth}
        \centering
        \includegraphics[height=\imageheight, keepaspectratio]{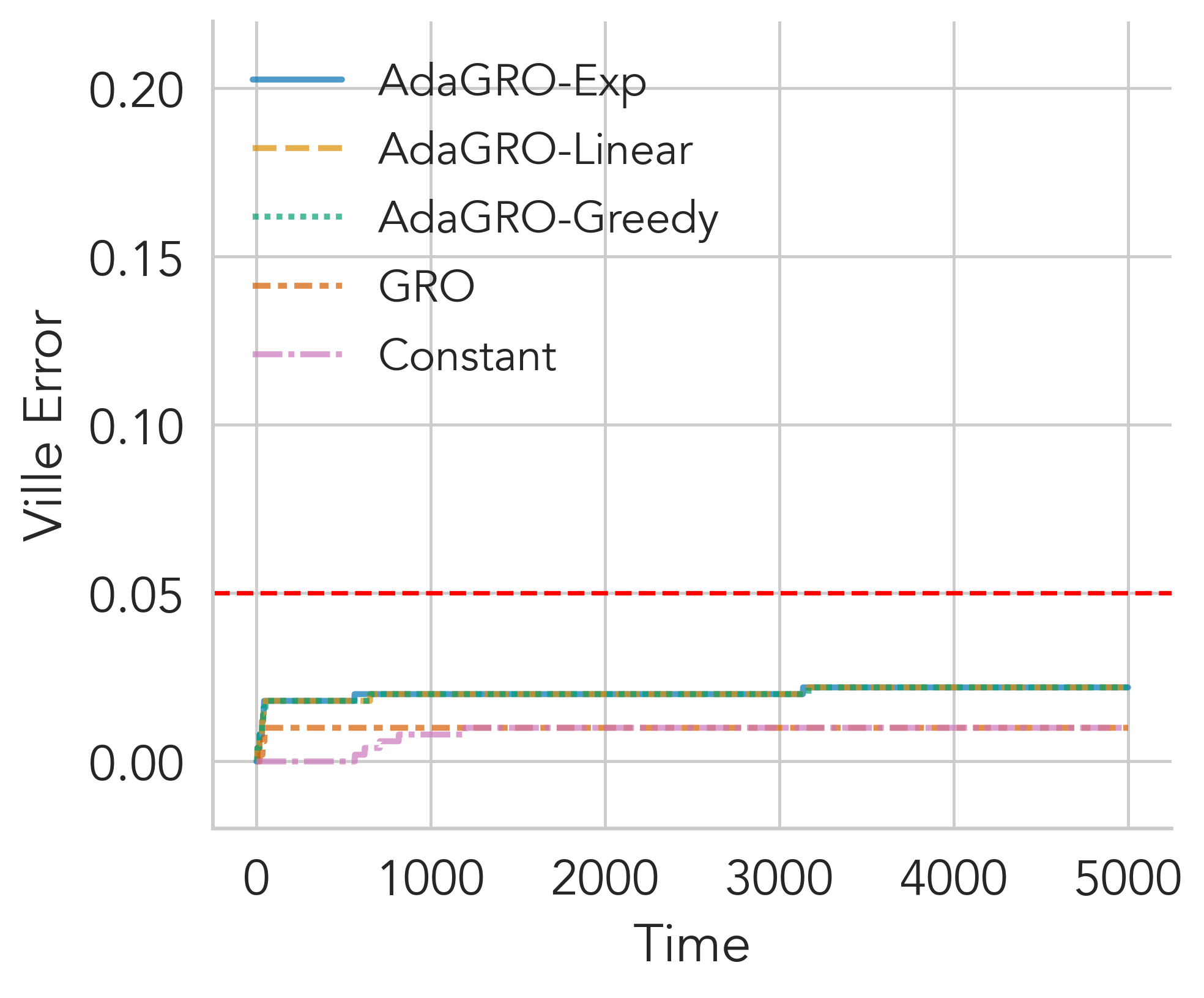}
        \caption{Ville error for $\calH_0: X \fsd Y$.}
        \label{subfig:anticorr_ville}
    \end{subfigure}
    \begin{subfigure}[t]{0.34\textwidth}
        \centering
        \includegraphics[height=\imageheight, keepaspectratio]{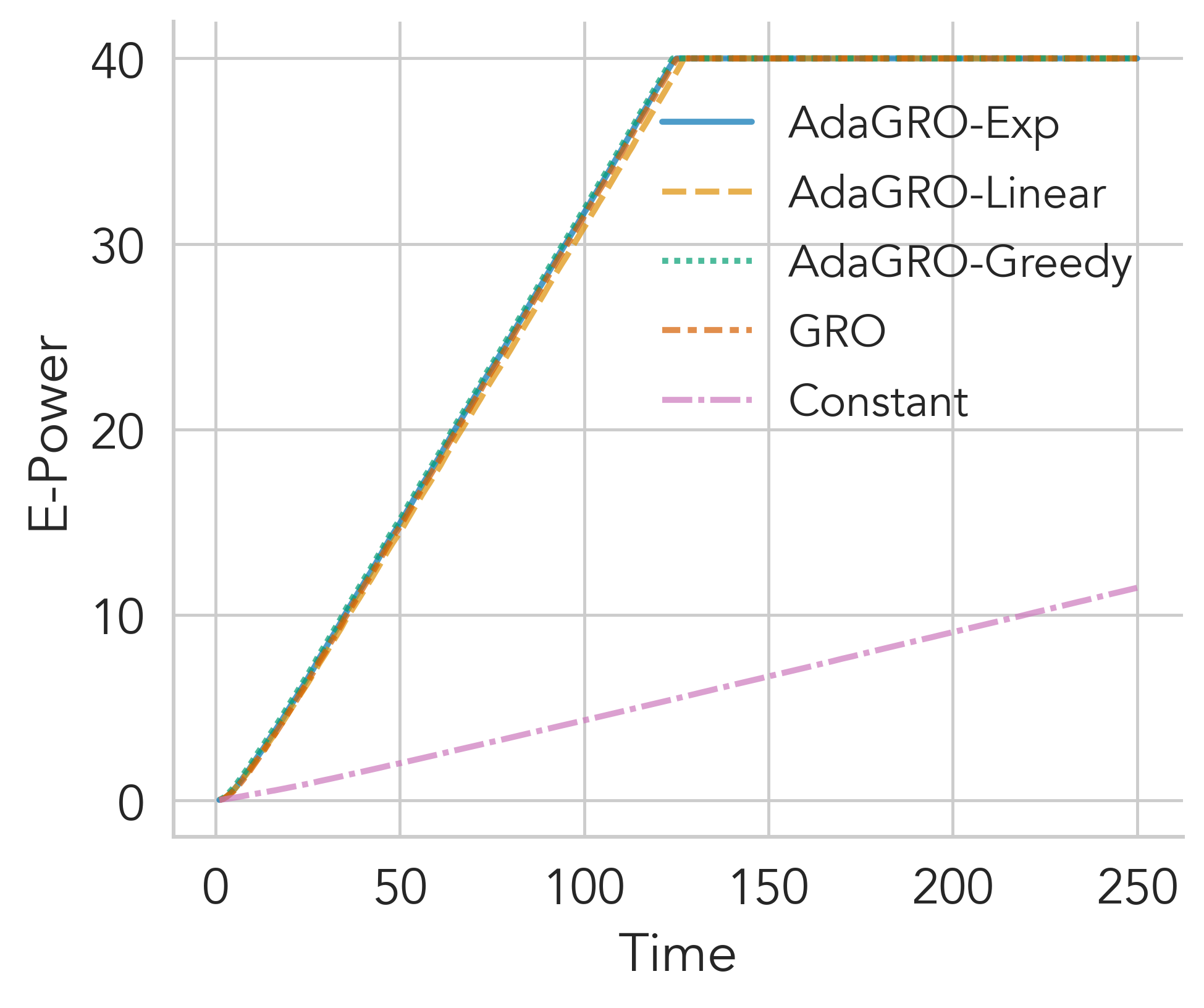}
        \caption{E-power against $\calH_0: Y \fsd X$.}
        \label{subfig:anticorr_epower}
    \end{subfigure}
    \caption{\textit{All GRO e-processes grow quickly under antimonotonicity while maintaining anytime-validity.}  
    Plots show simulations with finite support (4 outcomes) and maximal antimonotonicity ($\rho(X,Y)=-1$). Each line is averaged over $500$ repeated simulations. 
    The Ville error plot is drawn for $t = 1, \dotsc, 5,000$, while the e-power plot is truncated at $t=500$ and at $40$ for better visualization. 
    Over time, a linearly growing e-power corresponds to an exponentially growing e-process.}
    \label{fig:anticorr}
\end{figure}

Over 500 repeated simulations, we plot the mean Ville error for $\calH_0: X \fsd Y$ (Figure~\ref{subfig:anticorr_ville}) and the mean e-power against $\calH_0: Y \fsd X$ (Figure~\ref{subfig:anticorr_epower}) of the five e-processes. 
For all e-processes, the mean Ville error does not climb above $\alpha=0.05$ (red horizontal line), illustrating their anytime-validity. 
Further, the mean e-power grows very quickly for all GRO e-processes (adaptive or not), reaching $\geq 30$ within $t=100$ data points and much quicker than the constant bet e-process. 

For this toy example, we do not see the added benefits of using adaptive methods. 
As mentioned in Remark~\ref{remark:finite_support}, with finite and known support, the adaptive GRO e-process with predictable weights can at best improve the non-adaptive GRO e-process by a constant factor in e-power terms. 
In this example, two out of three thresholds are ``profitable'' bets ($z=0$ and $z=2/3$), and thus even the non-adaptive GRO e-process and other adaptive baselines perform just as well in terms of e-power. 
Without such knowledge \emph{a priori}, we still recommend using the adaptive GRO e-process with exponential weights for its flexibility.

\subsection{Comparison with KS-type sequential tests for FSD}\label{sec:empirical_comparison_ks_tests}

We now compare adaptive and non-adaptive GRO e-processes (Section~\ref{sec:FSD}) against the aforementioned KS-type sequential tests for FSD: \citet[DR68]{darling1968some}, \citet[HR22]{howard2022sequential}, and \citet[MR23]{manole2023martingale}, as discussed in Section~\ref{app:ks_test}, alongside the recently developed CDF bands by \citet[CFR26]{clerico2026uniform} which improve upon HR22. 
We focus on two setups from our main simulations (Section~\ref{sec:sim_kuniform}), one where there is substantial contact between $F_X$ and $F_Y$ ($z_0=0.2$) and another where there is no contact ($z_0=1.0$). 
Since KS-type tests do not explicitly yield e-processes, we compare the distributions of the stopping times $\tau_\alpha = \inf\{t \in \N: \phi_t = 1\}$ at level $\alpha=0.05$, over repeated simulations under each alternative (200 repeats of maximum length $T=5,000$ each). 
Smaller rejection times indicate higher power. 
For CFR26's method, for computational reasons, we use 20 Newton steps for the optimization step at each sample size, rather than the suggested 50, noting the results were identical even when using just 10 Newton steps. 

\begin{figure}[t]
    \centering
    \begin{subfigure}[b]{0.42\textwidth}
        \centering
        \includegraphics[width=\textwidth]{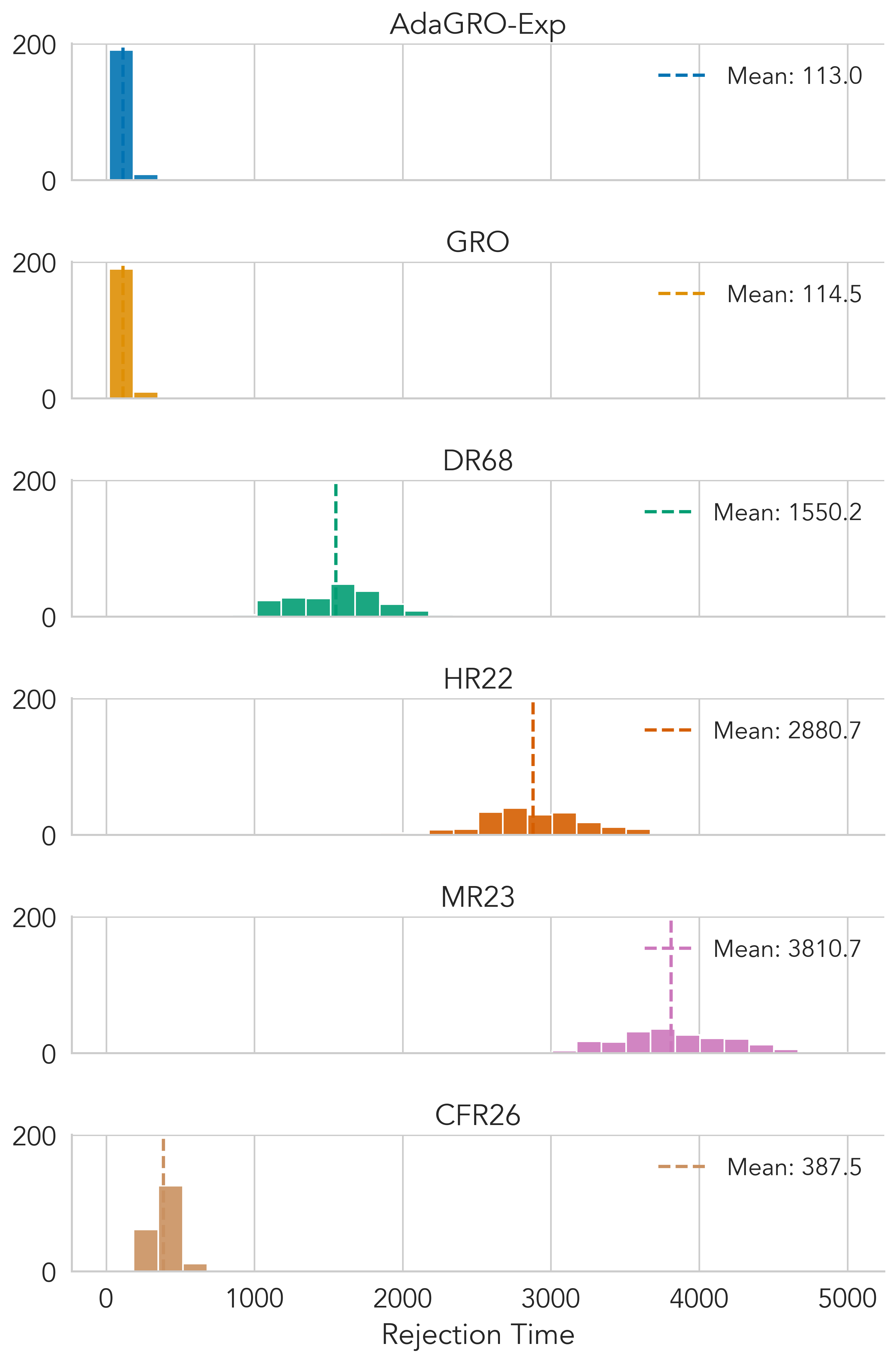}
        \caption{Contact between CDFs ($z_0=0.2$).}
        \label{fig:ks_tests_comparison_z0.2}
    \end{subfigure}
    ~
    \begin{subfigure}[b]{0.42\textwidth}
        \centering
        \includegraphics[width=\textwidth]{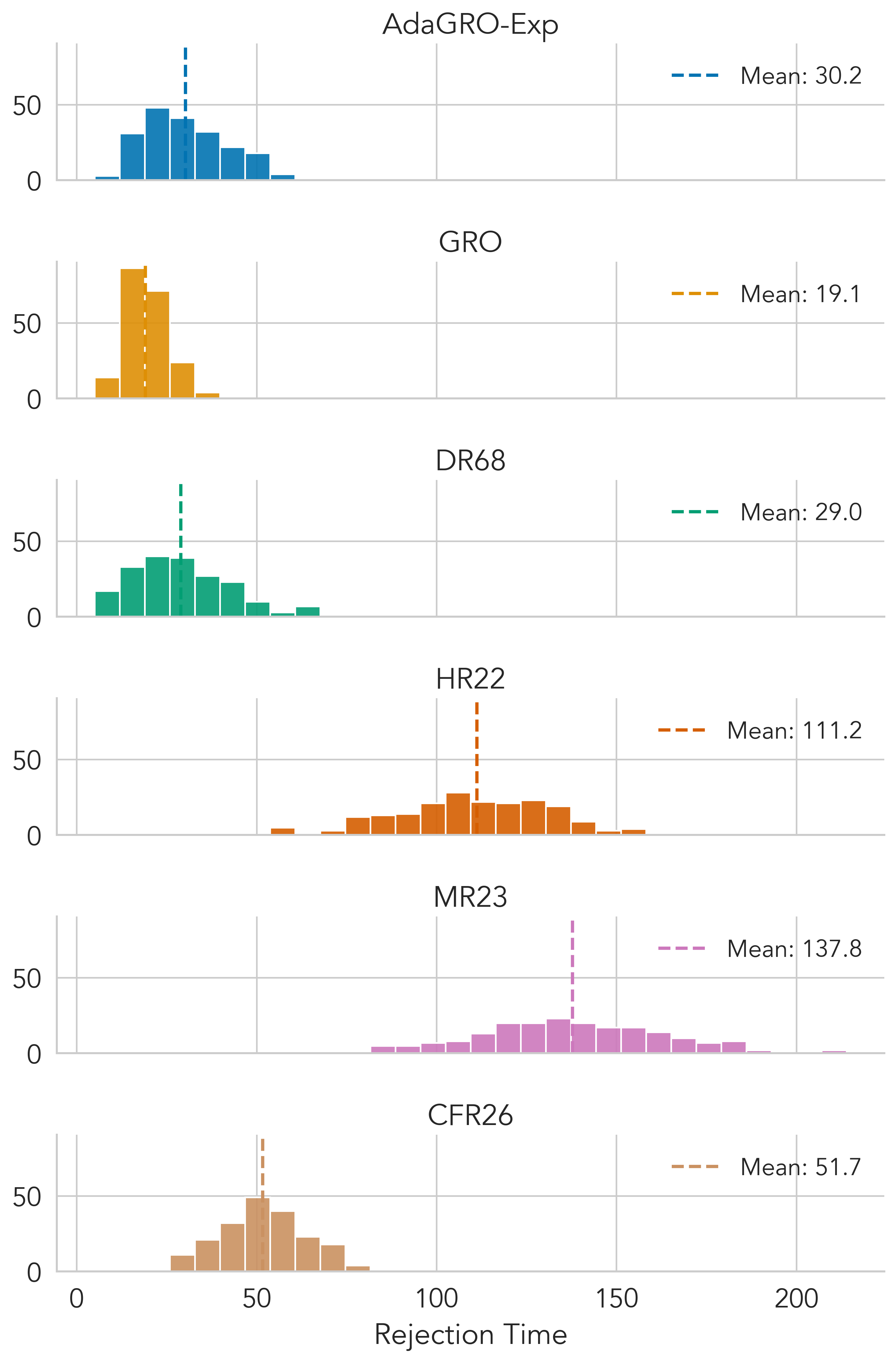}
        \caption{No contact between CDFs ($z_0=1.0$).}
        \label{fig:ks_tests_comparison_z1.0}
    \end{subfigure}
    \caption{Histograms and means of the stopping times $\tau_\alpha$, at level $\alpha=0.05$, for the sequential tests induced by adaptive and non-adaptive GRO e-process for FSD, alongside various KS-type sequential tests for FSD (DR68, HR22, MR23, and CFR26). 
    The simulation setups are identical to two setups from Section~\ref{sec:sim_kuniform}: one where there is substantial contact between $F_X$ and $F_Y$ ($z_0=0.2$) and another where there is no contact ($z_0=1.0$). 
    Results are from 200 repeated simulations of length $T=5,000$ each.}
    \label{fig:ks_tests_comparison}
\end{figure}

In Figure~\ref{fig:ks_tests_comparison}, we plot both the histograms and means of the stopping times $\tau_\alpha$ for the KS-type sequential tests (DR68, HR22, MR23, and CFR26) alongside the sequential e-tests induced by adaptive and non-adaptive GRO e-processes (AdaGRO-Exp and GRO, respectively).
We find that, regardless of whether there is contact between $F_X$ and $F_Y$, sequential e-tests (AdaGRO-Exp and GRO) are generally more powerful than KS-type sequential tests, and the difference is large when there is contact (and thus, a smaller non-dominance region). 

In the first case, where there is substantial contact (Fig.~\ref{fig:ks_tests_comparison_z0.2}), both GRO e-processes require 113.0 and 114.5 observations to reject on average, respectively, while the best-performing KS-type test (CFR26) requires more than three times the sample size (387.5) to reject in the same setup on average. 
The other baselines (DR68, HR22, and MR23) are not competitive, as they require much more on average (from 1,550 to 3,811 data points). 
When there is no contact (Fig.~\ref{fig:ks_tests_comparison_z1.0}), the mean stopping times for AdaGRO-Exp and GRO (30.2 and 19.1, respectively) are still substantially smaller than all baselines, except for DR68 (29.0) which is now comparable with the sequential e-tests. 
This is because we are in a particularly favorable case where the CDF difference is strictly positive everywhere except $z=1$ and adaptivity (as in Section~\ref{sec:sim_normal}) is not necessary. 
Even in this regime, the non-adaptive GRO e-test still achieves a smaller mean stopping time with less spread, and the adaptive GRO e-test achieves a comparable stopping time distribution. 
Other KS-type baselines are less competitive (111.2, 137.8, and 51.7 for HR22, MR23, and CFR26, respectively).

\subsection{Testing second- and third-order SD}\label{app:sim_kuniform_ksd}

In this simulation, we test for 2-SD and 3-SD in the same scenarios as the main simulation from Section~\ref{sec:sim_kuniform}. 
Since $X \fsdstrict Y$ implies $X \ssdstrict Y$ and $X \tsdstrict Y$, we should achieve e-power against both the 2-SD null, $\calH_0^{[2]}: Y \ssd X$, and the 3-SD null, $\calH_0^{[3]}: Y \tsd X$. 
Here, instead of the GRO e-process (Theorem~\ref{thm:Test-supermartingale-for-the-FSD_null}), we use the UP e-process (Theorem~\ref{thm:Test-supermartingale-for-the-kSD_null}) using the lower bound of $a=0$. 
Recall that, for each threshold $z$, the UP e-process is only asymptotically optimal. 
Thus, we will see an interaction between the suboptimality of UP bets in the early rounds and the adaptivity of predictable weights. 
The choice of predictable weights is kept the same as in the GRO e-processes for testing FSD, with 21 equidistant thresholds in $[0,1]$ at the start and then updated adaptively based on the observed quantiles of data. 

\begin{figure}[t]
    \centering
    \begin{subfigure}[t]{\textwidth}
        \centering
        \includegraphics[width=\textwidth, keepaspectratio]{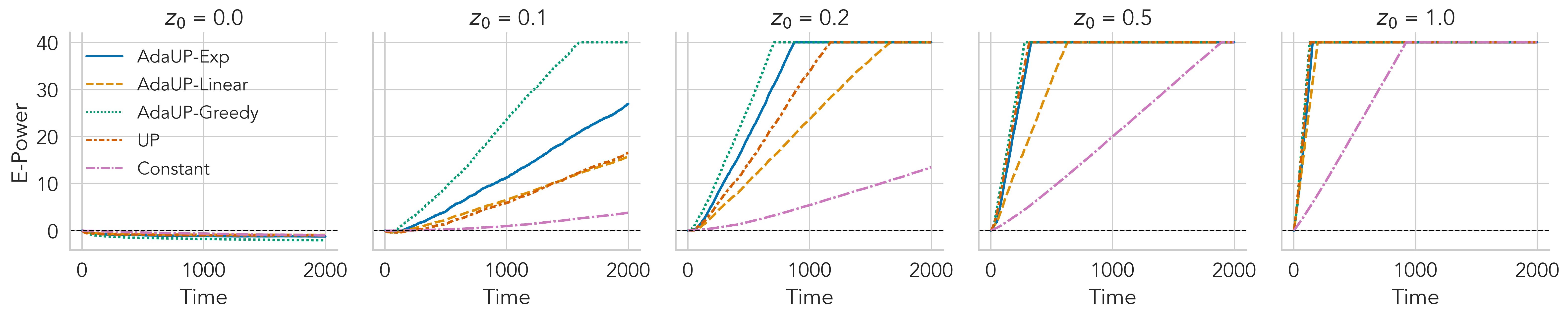}
        \caption{E-power against $\calH_0: Y \ssd X$ (2-SD).}
        \label{subfig:kuniform_epower_2sd}
    \end{subfigure}
    \begin{subfigure}[t]{\textwidth}
        \centering
        \includegraphics[width=\textwidth, keepaspectratio]{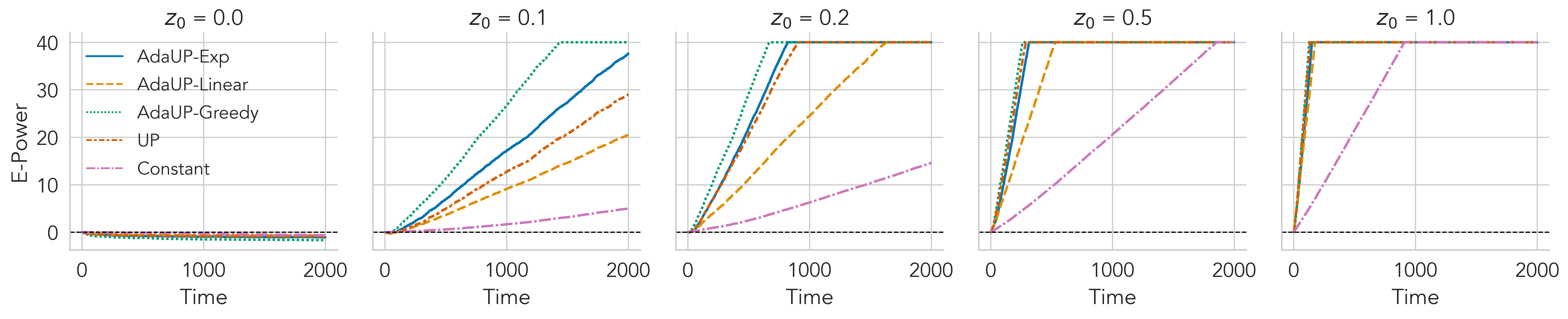}
        \caption{E-power against $\calH_0: Y \tsd X$ (3-SD).}
        \label{subfig:kuniform_epower_3sd}
    \end{subfigure}
    \caption{\textit{For higher-order SD testing, adaptive UP e-processes grow quickly relative to non-adaptive UP or constant-bet counterparts, particularly when the contact set between the CDFs is large.}
    These simulations have the same setup as for Figure~\ref{fig:kuniform}, but for testing second- and third-order SD. The e-power is truncated at 40 for ease of visualization.}
    \label{fig:kuniform_ksd}
\end{figure}

Figure~\ref{fig:kuniform_ksd} summarizes the e-power against the 2-SD null (Figure~\ref{subfig:kuniform_epower_2sd}) and the 3-SD null (Figure~\ref{subfig:kuniform_epower_3sd}). 
All UP e-processes grow in e-power even with large contact sets, although this time, the adaptive UP e-processes with \emph{greedy} weights have the highest e-power across time. 
Intriguingly, in the $z_0=0.1$ case (large contact set), the adaptive UP e-process with greedy weights now has substantial advantage over the one with exponential weights, which itself has nontrivial advantage over the non-adaptive UP e-process. 
One possible explanation is that the UP e-process is not necessarily optimal during the early rounds, and that the effect of greedy weights adapting more quickly to ``profitable thresholds'' is more pronounced. 

\begin{figure}[t]
    \centering
    \begin{subfigure}[t]{\textwidth}
        \centering
        \includegraphics[width=\textwidth, keepaspectratio]{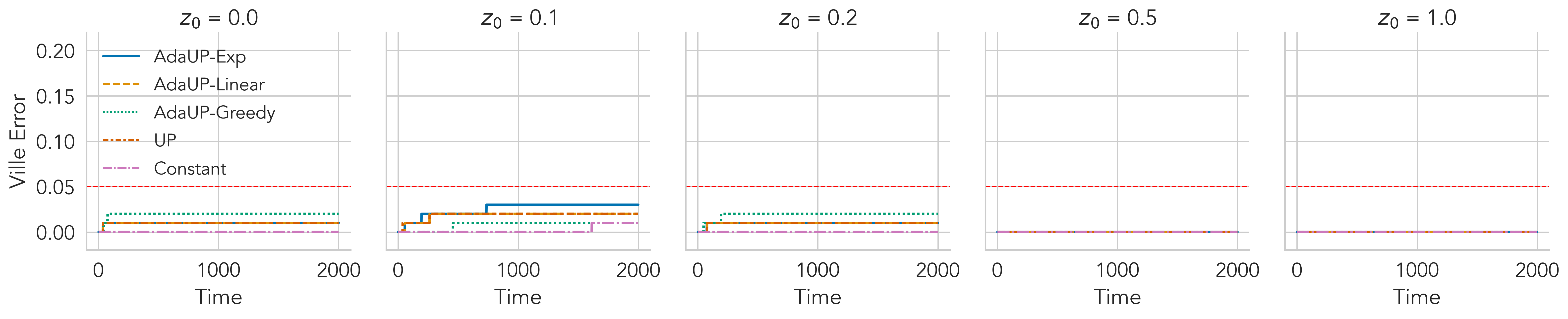}
        \caption{Ville error for testing $\calH_0: X \ssd Y$ (2-SD).}
        \label{subfig:kuniform_ville_errors_2sd}
    \end{subfigure}
    \begin{subfigure}[t]{\textwidth}
        \centering
        \includegraphics[width=\textwidth, keepaspectratio]{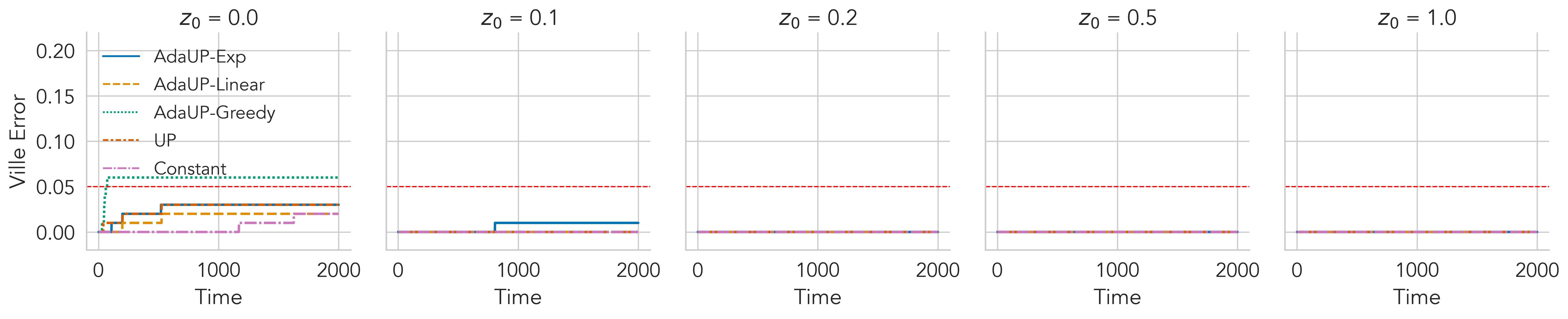}
        \caption{Ville error for testing $\calH_0: X \tsd Y$ (3-SD).}
        \label{subfig:kuniform_ville_errors_3sd}
    \end{subfigure}
    \caption{All UP e-processes (approximately) control the Ville error at level $\alpha=0.05$ across all values of $z_0$. The Ville error at time $t$ is calculated as the cumulative type I error up to time $t$, averaged over $100$ repeated simulations.} 
    \label{fig:kuniform_ksd_validity}
\end{figure}

Next, in Figure~\ref{fig:kuniform_ksd_validity}, we further verify that all UP e-processes for testing 2-SD and 3-SD control (or get close to controlling) the Ville error at level $\alpha=0.05$ across all values of $z_0$. 
Note that the $z_0=0.0$ case is the ``hard'' case, where the two CDFs are identical and it is the easiest to make false rejections; at $t=2,000$, the greedy variant incurs a Ville error of $0.06$ with a 95\% Wald interval (over 100 repeated simulations) of $(0.014, 0.107)$. We see that the Ville error is zero for the ``easy'' cases where $z_0 \geq 0.5$, for both 2-SD and 3-SD.

Thus, the UP e-process is anytime-valid, and it empirically achieves a positive growth rate when testing higher-order SD, with adaptive strategies based on greedy or exponential weights providing further power gains when the violation region is relatively small.

\section{Additional details and results for 3TTO testing}\label{app:real_data_exp}

\subsection{Data pre-processing and visualization}\label{app:baseball_data}

All data sources are downloaded directly from the publicly available MLB Statcast database (\url{https://baseballsavant.mlb.com/statcast_search}). 
For each pitcher, we downloaded the at-bat-level data for all available regular season games from 2016 (or whenever the pitcher made his MLB debut) to 2025. 

For each at-bat, there are a total of 18 unique outcomes, many of which are rare. 
Each of these outcomes is consolidated into the following five categories: 
Out (\texttt{strikeout}, \texttt{force\_out}, \texttt{field\_out}, \texttt{grounded\_into\_double\_play}, and \texttt{strikeout\_double\_play}), 
Walk (\texttt{walk} and \texttt{hit\_by\_pitch}), 
1B (\texttt{single}), 
2B/3B (\texttt{double} and \texttt{triple}), and
HR (\texttt{home\_run}). 
The following events are excluded from the analysis: \texttt{intent\_walk}, \texttt{sac\_fly}, \texttt{sac\_bunt}, \texttt{field\_error}, \texttt{fielders\_choice}, \texttt{fielders\_choice\_out}, and \texttt{catcher\_interf}. 
For simplicity, if either of the 1TTO or 3TTO outcomes for an at-bat falls into the excluded categories, we exclude that paired outcome from the analysis. 

\begin{figure}[t]
    \centering
    \includegraphics[width=0.5\textwidth]{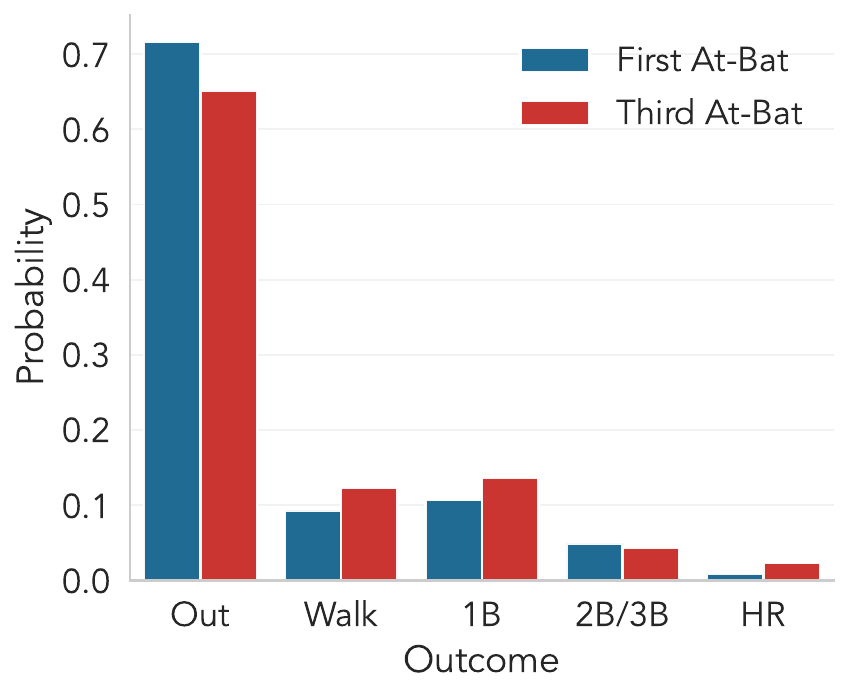}
    \caption{A histogram view of at-bat outcomes in 1TTO and 3TTO for Blake Snell.}
    \label{fig:baseball_tto_snell_dist}
\end{figure}

Figure~\ref{fig:baseball_tto_snell_dist} shows the histograms of outcomes at 1TTO and 3TTO for Blake Snell, which are the two data streams we compare in our main analysis in Section~\ref{sec:real_data_exp}. 
The histograms show that there is a notable drop in the number of outs in 3TTO and a corresponding increase in the number of walks, singles (1B), and home runs (HR) in 3TTO.

\subsection{What ``batter upside'' means and a fine-grained analysis}\label{app:batter_upside}

The baseball example illustrates an effective use of the SD framework for evaluating pitcher performance, given that the usual at-bat outcomes are ordinal. 
The null hypothesis $\calH_0: Y \fsd X$, where $X$ is the random outcome at 1TTO and $Y$ is the random outcome at 3TTO, has an intuitive interpretation in the problem context. 
The intersection representation for this null hypothesis can be written as: $\calH_0 = \cap_{z=0}^3 \calH_0(z)$, where $\calH_0(z) = \incurly{ \Psymb: \Psymb(X > z) \geq \Psymb(Y > z)}$.
Then, rejecting the null hypothesis corresponds to saying that \emph{at least one} of these statements ($\calH_0(z)^c$ for some $z=0,1,2,3$) is true: 
in their third time facing this pitcher (3TTO), relative to their first time (1TTO),
\begin{itemize}
    \item The batter is more likely to get \emph{on base (OB)}: walk, 1B, 2B, 3B, or HR.  
    \item The batter is more likely to get a \emph{hit (H)}: 1B, 2B, 3B, or HR. 
    \item The batter is more likely to get an \emph{extra-base hit (XBH)}: 2B, 3B, or HR. 
    \item The batter is more likely to get a \emph{home run (HR)}.  
\end{itemize}
These statements respectively correspond to $\Psymb(X > z) < \Psymb(Y > z)$ for $z=0,1,2,3$; if there is strong evidence for any of these statements, then the opposing batter has an upside in 3TTO relative to 1TTO, signaling to the manager that they need to consider taking the pitcher out. 
In practice, of course, a more rigorous analysis would take exogenous features into account in the decision, such as the quality of the pitcher's replacement; the pitcher's stamina and health status in the season; performance in 2TTO; and so on.

\begin{figure}[t]
    \centering
    \includegraphics[width=0.85\textwidth]{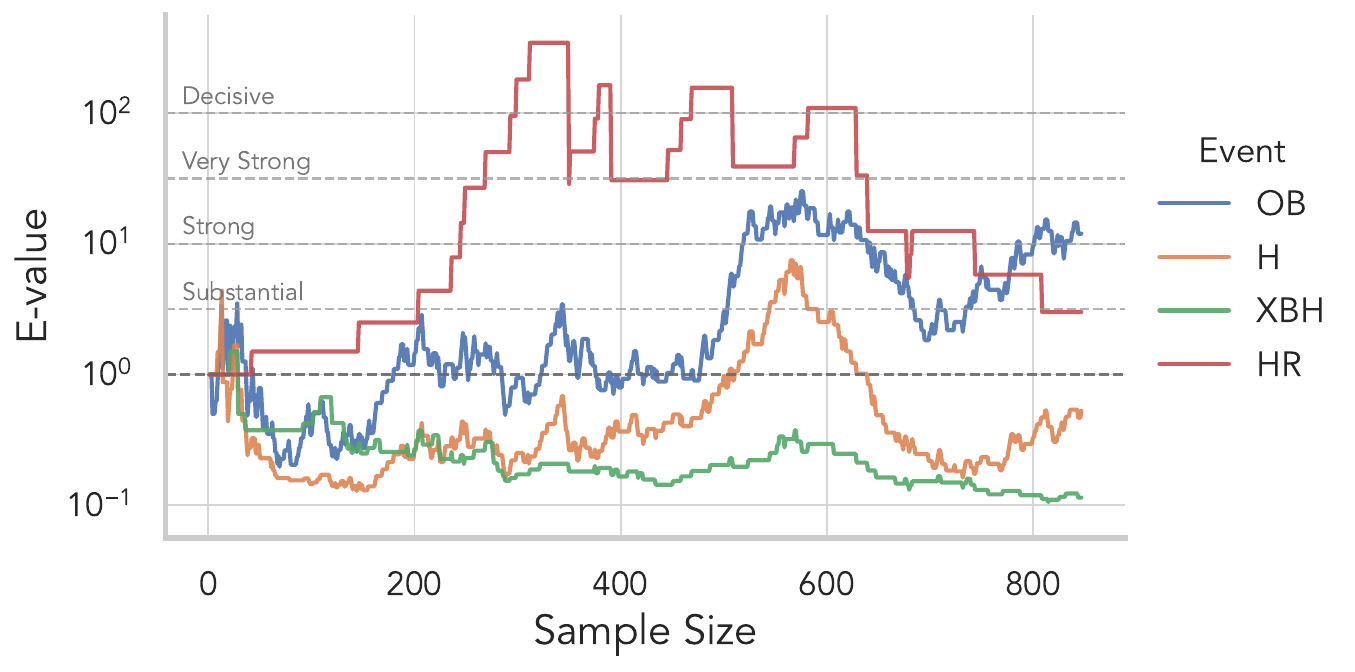}
    \caption{GRO e-processes against individual component nulls, $\calH_0(z) = \{\Psymb: \Psymb(X \leq z) \leq \Psymb(Y \leq z)\} = \{\Psymb: \Psymb(X > z) \geq \Psymb(Y > z)\}$, where $X$ and $Y$ are outcomes in 1TTO and 3TTO, respectively. 
    For $z=0,1,2,3$, each e-process quantifies evidence for the batter's upside in 3TTO over 1TTO for the cumulative events of getting on base (OB), a hit (H), an extra-base hit (XBH), and a home run (HR), respectively.}
    \label{fig:baseball_tto_snell_indiv_z}
\end{figure}

In Figure~\ref{fig:baseball_tto_snell_indiv_z}, we plot the GRO e-processes for these individual statements in their null form, $\calH_0(z)$ for $z=0,1,2,3$. 
The ``event'' of interest is the cumulative outcome for $X$ (and for $Y$): OB for $X>0$, H for $X>1$, XBH for $X>2$, and HR for $X>3$ (i.e., $X=4$). 
The plot suggests that, at the end of the 2020 regular season ($t=395$), there was very strong evidence ($E \approx 31.6$) that Snell's opponents had a higher probability of hitting a home run in their 3TTO than in 1TTO; in contrast, there was little such evidence for other cumulative events (OB, H, and XBH). 
After 2020, there is further strong evidence that the batters are more likely to get on base in 3TTO than in 1TTO. 
On the other hand, there is little evidence throughout the monitoring time frame that Snell's opponents are more likely to get XBHs in 3TTO than in 1TTO. 
These individual GRO e-processes further provide fine-grained evidence for the batter's upside in an intuitive manner.

\subsection{Do other top pitchers have a relative downside in 3TTO?}\label{app:other_pitchers}

\begin{table}[t]
    \centering
    \begin{tabular}{l|rr|rr|rr|rr}
        \toprule
        \multirow{2}{*}{\bf Pitcher} & \multirow{2}{*}{\bf GS} & \multirow{2}{*}{\bf IP} & \multirow{2}{*}{\bf ERA} & \multirow{2}{*}{\bf fWAR} &
        \multicolumn{2}{c}{\bf End of 2020} & \multicolumn{2}{|c}{\bf End of 2025} \\
         &  &  &  &  & $\tau_{2020}$ & $E_{\tau_{2020}}$ & $\tau_{2025}$ & $E_{\tau_{2025}}$ \\
        \midrule
        Max Scherzer & 243 & 1495 & 2.98 & 40.9 &
        904 & 0.80 & 1412 & 2.84 \\
        Aaron Nola & 272 & 1638 & 3.84 & 36.8 &
        669 & 0.52 & 1515 & {8.67} \\
        Gerrit Cole & 244 & 1491 & 3.21 & 36.4 & 
        794 & 0.51 & 1434 & 1.92 \\
        Clayton Kershaw & 209 & 1242 & 2.67 & 31.9 & 
        650 & 0.33 & 1027 & 0.26 \\
        Blake Snell & 222 & 1158 & 3.15 & 26.4 & 
        395 & \bf 17.41 & 847 & \bf 57.84 \\
        Yu Darvish & 214 & 1233 & 3.82 & 23.1 &  
        534 & 0.74 & 1085 & 1.40 \\
        Kyle Hendricks & 256 & 1468 & 3.85 & 22.4 & 
        709 & 0.20 & 1296 & 0.33 \\
        Marcus Stroman & 237 & 1334 & 3.85 & 20.9 & 
        707 & \bf 186.71 & 1227 & \bf 1087.79 \\
        Jack Flaherty & 184 & 984 & 3.82 & 15.9 & 
        303 & \bf 107.44 & 777 & \bf 12.77 \\
        Matthew Boyd & 187 & 1010 & 4.47 & 15.0 & 
        598 & 0.49 & 892 & 0.71 \\
        \bottomrule
    \end{tabular}
    \caption{Stopped e-values for batter upside in 3TTO over 1TTO for some of the top pitchers in MLB, using regular season data from 2016 to 2025. For each starting pitcher, we show the stopped e-values at the end of the 2020 regular season ($\tau_{2020}$), before the Snell incident, and at the end of 2025 ($\tau_{2025}$). Some other statistics are also shown, including games started (GS), innings pitched (IP), earned run average (ERA), and the Fangraphs wins above replacement (fWAR) metric. E-values larger than 10 are highlighted in bold.}
    \label{tbl:evalues_top_pitchers}
\end{table}

Despite exhibiting the 3TTO downside (or batter upside), Blake Snell is considered one of the top pitchers in MLB: out of all starting pitchers who pitched from 2016 to 2025 (ten regular seasons), Snell ranked 16th in terms of earned run average (ERA), a traditional measure of pitcher performance, and 12th in terms of Fangraphs wins above replacement (fWAR) metric, a popular sabermetric measure.\footnote{Source: Fangraphs, \url{https://www.fangraphs.com/}}
Given the now-pervasive belief that the 3TTO effect is a common phenomenon for pitchers, we ask whether other top pitchers also exhibit evidence of batter upside in 3TTO relative to 1TTO.
Because managers are less likely to take out top pitchers for their 3TTO, we are generally less likely to encounter the aforementioned issues of selection bias. 

In Table~\ref{tbl:evalues_top_pitchers}, we report the GRO e-processes (specifically, AdaGRO-Exp) for testing the 3TTO downside for ten MLB pitchers, using regular season data from 2016 to 2025. 
To provide a compact summary, we report the stopped e-values at the end of the 2020 regular season ($\tau_{2020}$), right before the 2020 World Series, and at the end of 2025 ($\tau_{2025}$), the last full season at the time of writing. 
We select 10 high-performing pitchers (all with fWAR $\geq 15.0$), including Snell, over 2016--2025. 
In addition to the e-values, we also report other standard metrics for starting pitchers, including games started (GS), innings pitched (IP), ERA, and fWAR. 
The table is sorted by fWAR, which estimates how many wins a pitcher contributes to his team above a replacement-level pitcher.

The results are somewhat heterogeneous across these top pitchers. 
There is no strong evidence for the 3TTO downside observed for other best-performing pitchers in terms of fWAR (Scherzer, Nola, Cole, and Kershaw), making Snell an exception. 
Yet, among pitchers with lower fWAR than Snell, some also do not have observed 3TTO downside (Darvish, Hendricks, and Boyd), while others exhibit strong (if not decisive) evidence for the 3TTO downside (Stroman and Flaherty). 
We remark that we only show a sample of 10 pitchers, mainly to contextualize our findings in the main section, and a more comprehensive analysis across all MLB pitchers is left for future work.

\section{Extensions}\label{sec:extensions}

In this section, we elaborate on extensions to general integral stochastic orders and to multivariate, non-i.i.d., and unpaired data.

\subsection{General integral stochastic orders}\label{app:integral_stochastic_orders}

The e-process constructions in Theorems~\ref{thm:Test-supermartingale-for-the-FSD_null} and \ref{thm:Test-supermartingale-for-the-kSD_null} extend to testing SD relations under arbitrary classes of utility generators, also known as \emph{integral stochastic orders}~\citep{denuit2002smooth,feng2025integral}. 
Given an arbitrary class of utility functions $\calU = \{u: \calZ \to \R\}$, we may define $Y \sd_{\calU} X$ if and only if $\ex{u(Y)} \leq \ex{u(X)}$ for all $u \in \calU$. Below, we elaborate on two particularly important examples: the \emph{increasing convex order} and the \emph{infinite-order} SD.

\textbf{Increasing convex orders.} 
By complete analogy to the construction for 2-SD (increasing concave order), we can also construct e-variables and test supermartingales for testing SD in the increasing convex order (for risk-seeking DMs, in the expected-utility framework). We now require an upper bound instead of a lower bound on the support, that is $\calZ=(-\infty,b]$ for some $b<\infty$.
The associated utility generators are
\begin{equation*}\label{eq:utility_generator_icx}
    u_z^{\mathsf{icx}}(x) = \frac{(x - z)_+}{b - z}, \quad x \in \calZ,
\end{equation*}
for $z\in (-\infty,b)$, and $u_b^{\mathsf{icx}}(x)\equiv0$.
Essentially, increasing convex order relies on the survival function $\bar{F}(z) = 1 - F(z)$ in its definition, leading to the above generator; see \citet[][Theorem 1.5.7]{muller2002comparison} for a proof. 
The rest of the construction and proof is completely analogous to that for Theorem~\ref{thm:Test-supermartingale-for-the-kSD_null} in the $k=2$ case (even the analogous $k$-th icx order can be handled similarly). 

\textbf{Infinite-order SD (Laplace transform order).} 
Another popular integral stochastic order is the \emph{infinite-order SD}, also known as the \emph{Laplace transform order}~\citep{fishburn1980continuaunbounded}, defined as
\begin{equation*}\label{eq:infinite_sd}
    Y \sd_{\infty} X \quad \text{if and only if} \quad \ex{e^{-rX}} \leq \ex{e^{-rY}}, \, \forall r \in \R_{\geq 0},
\end{equation*}
assuming that the expectations are well-defined. This order is of particular interest in economics for modeling \emph{constant absolute risk aversion (CARA)} and in decision theory for modeling the preference of DMs to ``combine good lotteries with bad ones''~\citep[e.g.,][]{eeckhoudt2010simplifying}. 

Under the assumption of a lower bound, e.g.\ $\mathcal{Z}=[0,\infty)$, all given expectations exist and the infinite-order SD is generated by the exponential utility functions, $u_r^{[\infty]}(x) = 1-e^{-rx} \in [0,1],\, x \in \calZ$, for $r \geq 0$. In this case, the building-block e-variables and e-process construction from Theorem~\ref{thm:Test-supermartingale-for-the-kSD_null} directly extend with the utility function class indexed over $r \geq 0$. 

\textbf{Multivariate stochastic orders. }
Given the variety of notions of multivariate stochastic dominance proposed in the literature, see, e.g.,\ \citet[Section 6]{shaked2007stochastic}, we limit ourselves to a brief overview and a sketch of how our methods could be extended to obtain sequential SD tests for multivariate data. 

As in the univariate case, we say that a random vector $\bY$ is dominated by $\bX$, or $\bY\preceq_\calU \bX$ if $\Ex{}{u(\bY)} \leq \Ex{}{u(\bX)}$ for all $u\in \calU$, where $\calU = \{u:\R^d \to \R\}$ is a multivariate utility class. 
For example, the \emph{usual multivariate stochastic order} is obtained for the utility class $\calU=\{\indicator{\boldsymbol{x} \in U}\mid U\subseteq \R^d\textrm{ an upper set}\}$. 
In this case, all utility functions are bounded, and our results extend for the building-block e-variables $1+\lambda (\indicator{\bY \in U}-\indicator{\bX \in U})$. 

\subsection{Non-i.i.d.\ data}\label{sec:non_iid_data}

We continue to assume that the data arrive in pairs $(X_t,Y_t), t\in \mathbb{N}$, but no longer require that both $(X_t)_{t\in \N}$ and $(Y_t)_{t\in \N}$ are identically distributed.
Inspired by \citet{mineiro2023nonstationarity}, who construct time-uniform CDF bounds for non-stationary data, we define the \emph{strong FSD null (for non-i.i.d.~data)} as 
\begin{equation}\label{eq:def_strong_FSD_null_for_non_stationary_data}
    \calH_0^\textrm{s} = \incurly{
        \Psymb \in \frakB \mid \Psymb(X_t\leq z\mid \calF_{t-1})\leq \Psymb(Y_t\leq z\mid \calF_{t-1}),\; \forall z\in \calZ,\forall t\in \N
    } \subseteq \calH_0\,.
\end{equation}
In words, the null states that $X_t$ stochastically dominates $Y_t$ at each time step $t$, conditioned on the past. 
Here, the methods developed in Section~\ref{sec:FSD} readily extend to this setting, once again due to their martingale-based construction.  
In particular, for any $z \in \calZ$ and predictable $\lambda\in [0,1]$, $S_t(\lambda, z)$ defined in \eqref{eq:def_coin_betting_e-variable} remains a sequential e-variable for $\calH_0^{\textrm{s}}$, and the construction based on predictable mixtures over $\calZ$ (Theorem~\ref{thm:Test-supermartingale-for-the-FSD_null}) continues to apply. 

By conditioning, the inequality in \eqref{eq:def_strong_FSD_null_for_non_stationary_data} is read in an almost-sure sense, and for some given data streams it may happen that $\calH_0^\textrm{s}=\emptyset$. 
That is, if there is no continuity within each data stream and $\Psymb_{X_t}$ and $\Psymb_{Y_t}$ may change arbitrarily over time, inference for $\calH_0^\textrm{s}$ may become meaningless. 
Even in this setting, it may still hold that, conditional on the past, $X$ dominates $Y$ \emph{on average} at each $t$. 
This motivates the \emph{weak FSD null (for non-i.i.d.~data)}:
\begin{equation}\label{eq:def_weak_FSD_null_for_non_stationary_data}
     \calH_0^\textrm{w} =\left\{
        \Psymb \in \frakB \;\middle|\; \frac{1}{t}\sum_{\ell=1}^t\Psymb(X_\ell\leq z\mid \calF_{\ell-1})\leq \frac{1}{t}\sum_{\ell=1}^t\Psymb(Y_\ell\leq z\mid \calF_{\ell-1}),\; \forall z\in \calZ, \forall t\in \N
    \right\}\,.
\end{equation}

The difference between the strong null~\eqref{eq:def_strong_FSD_null_for_non_stationary_data} and the weak null~\eqref{eq:def_weak_FSD_null_for_non_stationary_data} mirrors that of the strong and weak nulls for comparing sequential forecasters in \citet{henzi2022valid} and \citet{choe2023comparing}, respectively. 
(The strong-versus-weak terminology itself originates from the classical distinction between Fisherian and Neymanian tests of individual versus average treatment effects.)
The question of what stochastic dominance exactly means for non-stationary data remains context-dependent. 
For the weak hypothesis, the building-block e-variables in Section~\ref{sec:FSD} would have to be replaced with e-variables of exponential form, such as empirical-Bernstein e-variables~\citep{howard2021time} employed by \citet{mineiro2023nonstationarity}, since $\calH_0^\textrm{w}$ only imposes a condition on the averaged conditional mean differences. 
Interestingly, unlike in the case of forecast comparison, the strong and weak nulls do \emph{not} coincide even when we assume weak stationarity of each data stream. 
In other words, the weak null is strictly larger than the strong null even under weak stationarity.

\subsection{Unpaired data}\label{sec:unpaired_data}

If the data streams $(X_t)_{t\in \N}$ and $(Y_t)_{t\in \N}$ do not arrive in pairs, we may instead compute the building-block e-variables from Section \ref{sec:FSD} on whole batches.
Suppose, at $t\in \N$, we observe $\boldsymbol{X}_t \in \R^{b_t}$ and $\boldsymbol{Y}_t \in \R^{b_t'}$, for $b_t,b_t'\geq 1$.
Then, for any $z\in \calZ$ and predictable $\lambda \in [0,1]$, the quantity $1+\lambda(\hat{F}_{{\boldsymbol{X}_t}}(z)-\hat{F}_{{\boldsymbol{Y}_t}}(z))$ is a sequential e-variable for the SD null~\eqref{eq:def_FSD_null} with respect to the batch filtration $\calF_t= \sigma(\{(\boldsymbol{X}_\ell ,\boldsymbol{Y}_\ell ), \ell \leq t\})$, and the sequential construction of Theorem~\ref{thm:Test-supermartingale-for-the-FSD_null} directly applies. 
This further applies to higher-order SD, where for $z\in \calZ$, $\hat{F}_{{\boldsymbol{X}_t}}(z)$ is replaced by $\bar{u}_{z}^{[k]}(\boldsymbol{X}_t)=\frac{1}{b_t}\sum_{j=1}^{b_t} u^{[k]}_z(X_{j,t})$ and analogously for $\boldsymbol{Y}_t$.